\documentclass[twoside,11pt]{article}

\usepackage{blindtext}
\usepackage[preprint]{jmlr2e}

\usepackage{lastpage}

\usepackage[ruled,vlined]{algorithm2e}
\usepackage{amsmath}
\usepackage{amsfonts}
\usepackage{tabularx}
\usepackage{diagbox}
\usepackage{xcolor}
\usepackage{bm}
\DeclareMathOperator{\tr}{trace}
\usepackage{enumitem}
\newtheorem{assumption}{Assumption}

\jmlrheading{24}{2023}{1-\pageref{LastPage}}{11/21}{2/23}{21-1313}{Henry Lam and Haofeng Zhang}

\ShortHeadings{Doubly Robust Stein-Kernelized Monte Carlo Estimator}{Lam and Zhang}
\firstpageno{1}

\begin{document}

\title{Doubly Robust Stein-Kernelized Monte Carlo Estimator: Simultaneous Bias-Variance Reduction and Supercanonical Convergence}

\author{\name Henry Lam \email khl2114@columbia.edu \\
\name Haofeng Zhang \email hz2553@columbia.edu \\
       \addr  Department of Industrial Engineering and Operations Research\\
        Columbia  University\\
       New York, NY 10027, USA
       %\AND
       %\name Michael I.\ Jordan \email %jordan@cs.berkeley.edu \\
       %\addr Division of Computer Science and %Department of Statistics\\
       %University of California\\
       %Berkeley, CA 94720-1776, USA
       }

\editor{Aryeh Kontorovich}

\maketitle

\begin{abstract}%   <- trailing '%' for backward compatibility of .sty file

%Monte Carlo methods are often used to compute risk or performance measures, especially for complex models whose expectation cannot be calculated in closed form. , Monte Carlo computation can sometimes encounter challenges in its convergence speed or the need to generate samples from biased models. We investigate the challenges of applying such techniques in some general situations with multiple input models each under different knowledge levels that occur in practice. 
Standard Monte Carlo computation is widely known to exhibit a canonical square-root convergence speed in terms of sample size. Two recent techniques, one based on control variate and one on importance sampling, both derived from an integration of reproducing kernels and Stein's identity, have been proposed to reduce the error in Monte Carlo computation to supercanonical convergence. This paper presents a more general framework to encompass both techniques that is especially beneficial when the sample generator is biased and noise-corrupted. We show our general estimator, which we call the doubly robust Stein-kernelized estimator, outperforms both existing methods in terms of mean squared error rates across different scenarios. We also demonstrate the superior performance of our method via numerical examples.

\end{abstract}


\section{Introduction}
\label{sec:intro}

%In a potentially large-scale simulation analysis, typically the modeler faces several, possibly many input variables with some untraceable noise, or inconsistent sampling distributions. Analytic intractability of sophisticated models has inspired the development of advanced Monte Carlo methodologies to facilitate computation. For the intro, here's a possible flow (which you already have most of it): ($\mathbb{E}_\pi[\hat{\theta}]=\theta$) \mathbb{E}_\pi[(\hat{\theta}-\theta)^2]=  which converges as the reciprocal of $n$

We consider the problem of numerical integration via Monte Carlo simulation. %In the most basic form of Monte Carlo methods, to
As a generic setup, we aim to estimate the expectation $\theta=\mathbb E_\pi[f(X)]$ using independent and identically distributed (i.i.d.) samples $x_i$, $i\in [n]:=\{1,\cdots,n\}$, drawn from the target distribution $\pi$. A natural Monte Carlo estimator is the sample mean $\hat{\theta}=\frac{1}{n}\sum_{i=1}^{n} f(x_i)$, which is widely known to be unbiased with mean squared error (MSE) $O(n^{-1})$, a rate commonly referred to as the canonical rate. 
% Similarly, we refer to a convergence rate higher than $O(n^{-1})$ as a supercanonical rate and a convergence rate lower than $O(n^{-1})$ as a sub-canonical rate. 
% Many Monte Carlo methods have been developed to facilitate the convergence rate of MSE. (which is common in practice, e.g., when we know the density function up to a normalizing constants)

In this paper, we are interested in Monte Carlo estimators that are supercanoncial, namely with MSE $o(n^{-1})$. We focus especially on recently proposed Stein-kernelized-based methods. These methods come in two forms, one based on control variate, or more generally control functional (CF), and one based on importance sampling (IS). They are capable of reducing bias or variance of Monte Carlo estimators to the extent that the convergence speed becomes supercanonical. On a high level, these approaches utilize knowledge on the analytical form of the sampling density function (up to a normalizing constant), and apply a ``kernelization" of Stein’s identity induced by a reproducing kernel Hilbert space (RKHS) to construct functions or weights that satisfy good properties for CF or IS purpose. 

Our main goal in this paper is to study a more general framework to encompass both approaches, by introducing a doubly robust Stein-kernelized (DRSK) estimator. The ``doubly robust" terminology borrows from existing approaches in off-policy learning \citep{dudik2011doubly,dudik2014doubly,jiang2016doubly,farajtabar2018more} where one simultaneously applies control variate and IS to reduce estimation error, and the resulting estimator is no worse than the more elementary estimators. In our setting, we will show that DRSK indeed outperforms kernel-based CF and IS across a range of scenarios, especially when the sample generator is imperfect in that the samples are biased or noise-corrupted. Importantly, distinct from the off-policy learning literature, the superiority of our estimator manifests in terms of faster MSE rates in $n$.
% superiority of our approach even under these adversarial situations.
% , in a sense, imperfect, in that  . 

% Previous work assumes that the gradient of the sampling density is fully known. 

More specifically, we consider a general estimation framework with target quantity $\theta=\mathbb{E}_\pi[f(X,Y)]$, where $X$ is a partial list of input variables that is analytically tractable (i.e., its density $\pi_X$ known up to a normalizing constant), while $Y$ is a noise term that is not traceable and is embedded in the samples $(x_i, f(x_i,y_i))$. Moreover, instead of having a generator for $\pi$, we may only have access to a possibly biased generator with distribution $q$. In this setting, we will demonstrate that applying the existing methods of kernel-based CF and IS both encounter challenges. In fact, because of these complications, these estimators could have a \emph{subcanonical} convergence. On the other hand, our DRSK still exhibits \emph{supercanonical} rate. 

Other than theoretical interest, a motivation of studying our considered general setting pertains to estimation problems that involve epistemic and aleatory variables, which arise when simulation generators are ``corrupted". More concretely, consider a target performance measure that is an expected value of some random outputs (aleatory noise), which are in turn generated from some input models. When the input models themselves are estimated from extrinsic data (epistemic noise), then estimating the target performance measure would involve handling the two sources of noises simultaneously, a common problem in stochastic simulation under input uncertainty \citep{xie2014bayesian,zouaoui2003accounting,song2014advanced,lam2016advanced,corlu2020stochastic,barton2022input}. Here, when the epistemic uncertainty is represented via a Bayesian approach, a natural strategy is to estimate the posterior performance measure, in which $X$ can represent the epistemic uncertainty and $Y$ the aleatory uncertainty. Typically, $X$ follows a posterior distribution that is known up to a normalizing constant, and may need to be generated using variational inference \citep{wainwright2008graphical,blei2017variational}\footnotemark\footnotetext{In variational inference \citep{blei2017variational}, the family of variational distributions generally does not contain the true posterior and thus the resulting approximate distribution is systemically biased/different from the true posterior.} or Markov chain Monte Carlo (MCMC) with independent parallel chains \citep{rosenthal2000parallel,LL}\footnotemark\footnotetext{Parallel computing of MCMC has attracted attention due to the enhancement of parallel processing units in GPU \citep{jacob2011using}. More precisely, here we mean that the $n$ samples are the end point of $n$ independent Markov chains respectively where each Markov chain is initialized from an i.i.d. starting point and run for the same amount of (possibly small) iterations; See Section \ref{sec:multipledistributions} for an example. These $n$ samples are independent but may exhibit a high bias as pinpointed by \citet{rosenthal2000parallel}, which falls into the scope of our study.},
so that the realized samples follow a biased distribution $q$ instead of $\pi$ \citep{liu2017stein2}\footnotemark\footnotetext{There are other problem settings where the biased distribution $q$ also arises, such as in parametric bootstrap or perturbed maximum a posteriori; See \citet{liu2017stein2} for details.}. Our work primarily focuses on the case where the generated samples of $X$ are independent, while other studies for the case where the samples are not necessarily independent (but without studying supercanonical rates) can be found in, e.g., \citet{south2022postprocessing,belomestny2021variance} and references therein. On the other hand, $Y$ could be generated from a black-box simulator that lacks analytical tractability. Such problems with biased epistemic generators and black-box aleatory noises comprise precisely the setup where DRSK is well-suited to enhance the estimation rates.
% of estimators in such problems.

To be more concrete, below we use a simple generic problem in stochastic simulation to illustrate our problem setting. A more sophisticated example of a computer communication network can be found in Section \ref{sec:multipledistributions} in our numerical experiments.

\begin{example} \label{example} Consider an M/M/1 queue with known arrival rate $1$ and unknown service rate $x$. Since the ground-truth $x$ is unknown, we may estimate $x$ via Bayesian inference based on historical data (say, the actual service time of several customers) and obtain its posterior distribution up to a normalizing constant. By leveraging MCMC or variational inference, only samples from an approximate posterior (instead of the exact posterior) can be drawn. Suppose we are interested in the mean of the waiting time of the first 10 customers. To do this, for each rate $x$, we simulate a fixed number of M/M/1 queues, obtain the waiting time of the first 10 customers in each queue, and output their average $f(x,y)$. Here $y$ represents the intrinsic noise (aleatory uncertainty) in the simulation model since the sample average waiting time given $x$ is still a random proxy for its expectation. The goal is to calculate the expectation of $f(X,Y)$ under the posterior of $X$ and the intrinsic noise of $Y$. Note, moreover, that when using a large or even moderate number of simulated queues, the noise level of $Y$ given $X$ in our estimate is small because of the averaging effect.
\end{example}

Finally, in terms of computational complexity, DRSK costs no more than kernel-based CF and IS. The main computational expense is to solve a kernel ridge regression (KRR) problem (as in kernel-based CF) and a convex quadratic program (as in kernel-based IS), both of which involve a Gram matrix whose dimension is related to the number of the samples. Although computational complexity is not the main focus of this work, we point out that KRR is known to not scale well with the growth of the sample size due to the matrix computation. Hence advanced computational techniques such as divide-and-conquer KRR \citep{pmlr-v30-Zhang13} may be applied when the sample size is large.

In the following, we first introduce some background and review the kernel-based CF and IS (Section \ref{sec:mainres}). With these, we introduce our main DRSK estimator, present its convergence guarantees and compare with the existing methods (Section \ref{sec:DRSK main}). After that, we demonstrate some numerical experiments to support our method (Section \ref{sec:numerics}). Finally, we develop the theoretical machinery for regularized least-square regression on RKHS needed in our analysis (Section \ref{sec:RLS}) and detail the proofs of our theorems (Section \ref{sec:analysis}).

\section{Background and Existing Methods} \label{sec:mainres}
We first introduce our setting and notations (Section \ref{sec:PS}), then review the technique of kernelization on Stein’s identity in an RKHS (Section \ref{sec:Stein}),  kernel-based CF (Section \ref{sec:CF0}) and IS (Section \ref{sec:BBIS0}), followed by a discussion on other related work (Section \ref{sec:literature}).

\subsection{Setting and Notation} \label{sec:PS}

Consider a random vector $(X,Y)$ where $X$ takes values in an open set $\Omega \subset \mathbb{R}^d$ and $Y$ takes values in an open set $\Gamma \subset \mathbb{R}^p$. Our goal is to estimate the expectation of $f(X,Y)$ under a distribution $(X,Y)\sim \pi$, which we denote as $\theta:=\mathbb{E}_\pi[f(X,Y)]$. The point estimator will be denoted as $\hat{\theta}$. We assume $X$ admits a positive continuously differentiable marginal density with respect to $d$-dimensional Lebesgue measure, which we denote as $\pi_X(x)$. Similarly, we denote $\pi_{Y|X}(y|x)$ as the conditional distribution of $Y$ given $X$ (which is not required to have density).

% We emphasize that it is fine that we do not know the value of $y_i$. 
Our premise is that we can run simulations and have access to a collection of i.i.d. samples $D = \{(x_i,f(x_i,y_i)): i=1,\cdots,n\}$ where $(x_i,y_i)$ are drawn from some distribution $q$ (which might be unknown and distinct from $\pi$). It is sometimes useful to think of $X$ as the ``dominating" factor in the simulation, contributing the most output variance, whereas $Y$ is an auxiliary noise and contributes a small variance (we will rigorously define these in the theorems later). The small variance from $Y$ can be justified in, e.g., the stochastic simulation setting where typically the modeler simulates and then averages a large number of simulation runs to estimate an expectation-type performance measure; Recall Example \ref{example}.

% We assume that the function values $f(x_i,y_i)$ and gradients $\nabla_x \log \pi_X(x_i)$ are pre-computed. The latter means that we know $\pi_X$ is given via an un-normalised density, which is satisfies by many cases in pratice.

For convenience, for any measurable function $g : \Omega \times \Gamma \to \mathbb{R}$, we write $\mu(g)=\mathbb E_\pi[g(X,Y)]$, and for any measurable function $g : \Omega \to \mathbb{R}$, we write $\mu_X(g)=\mathbb E_{\pi_X}[g(X)]$. If $g$ is constructed from training data, then $\mu(g)$ and $\mu_X(g)$ are understood as the conditional expectation of $g$ given training data. Let $L^2(\pi_X)$ denote the space of measurable functions $g : \Omega \to \mathbb{R}$ for which $\mu_X(g^2)$ is finite, with the norm written as $\|\cdot\|_{L^2(\pi_X)}$. Let $C^k(\Omega,\mathbb{R}^j)$ denote the space of (measurable) functions from $\Omega$ to $\mathbb{R}^j$ with continuous partial derivatives up to order $k$. The region $\Omega$ can be bounded or unbounded; in the former case, the boundary $\partial \Omega$ is assumed to be piecewise smooth (i.e., infinitely differentiable). 

%Beyond notations in Section \ref{sec:PS}, a few additional notations are needed here. 

Similarly, for any measurable function $g : \Omega \times \Gamma \to \mathbb{R}$, we write $\nu(g)=\mathbb E_q[g(X,Y)]$, and for any measurable function $g : \Omega \to \mathbb{R}$, we write $\nu_X(g)=\mathbb E_{q_X}[g(X)]$. If $g$ is constructed from training data, then $\nu(g)$ and $\nu_X(g)$ are understood as the conditional expectation of $g$ given training data. Let $L^2(q_X)$ denote the space of measurable functions $g : \Omega \to \mathbb{R}$ for which $\nu_X(g^2)$ is finite, with the norm written as $\|\cdot\|_{L^2(q_X)}$.
%(See Section \ref{sec:RLS} for details.)

Throughout this paper, we assume that $\pi_X\in C^1(\Omega, \mathbb{R})$. The \emph{score function} of the density $\pi_X$, $\bm{u}(x):=\nabla_x \log \pi_X(x)$ is well-defined and is computable for given $x_i$'s. This is equivalent to saying $\pi_X(x)$ has a parametric form that is known up to a normalizing constant.
%, so that we can apply CF on $X$ as we will discuss. 
We also assume that the target function $f : \Omega \times \Gamma \to \mathbb{R}$ satisfies $\mathbb{E}_\pi[f(X,Y)^2]<\infty$. 

\subsection{Stein-Kernelized Reproducing Kernel Hilbert Space} \label{sec:Stein}
We briefly introduce the technique of kernelization on Stein’s identity in an RKHS %The statement here is based on 
\citep{liu2016kernelized,CF1}. % or Appendix \ref{sec:setup}.

We say that a real-valued function $g(x): \Omega\subset \mathbb{R}^d \to \mathbb{R}$ is in the Stein class of $\pi_X(x)$ \citep{ley2017stein,liu2016kernelized} if $g(x)$ is continuously differentiable and satisfies
$$\int_{\Omega} \nabla_x(\pi_X(x)g(x))dx = \bm{0} \in \mathbb{R}^d.$$
% acting on the Stein class of $\pi_X$,
This condition can be easily checked using integration by parts or the divergence theorem; in particular, it holds if   $\pi_X(x)g(x) = 0$, $\forall x \in \partial \Omega$ when the closure of $\Omega$ is compact, or $\lim_{\|x\|\to \infty} \pi_X(x)g(x) = 0$ when $\Omega = \mathbb{R}^d$. The ``canonical" Stein operator of $\pi_X$, $\mathcal{A}_{\pi_X}$, acting on the Stein class of $\pi_X(x)$ is a (linear) operator defined as 
\begin{equation}\label{equ:steinoperator}
\mathcal{A}_{\pi_X}g(x) =\bm{u}(x)g(x) +\nabla_x g(x) \in \mathbb{R}^d.    
\end{equation}
where $\bm{u}(x):=\nabla_x \log \pi_X(x)$ is the score function of $\pi_X(x)$ as introduced earlier. Note that the general definition of the Stein operator typically depends on a class of functions that Stein operator acts on \citep{gaunt2019algebra}; Yet, if this class of functions is exactly the Stein class of $\pi_X(x)$, the ``canonical" Stein operator is defined uniquely by \eqref{equ:steinoperator} as suggested by previous studies \citep{stein2004use,liu2016kernelized,ley2017stein,mijoule2018stein}.

For any $g(x)$ in the Stein class of $\pi_X(x)$,
we have the well-known Stein’s identity \citep{liu2016kernelized,mijoule2018stein} as follows:
\begin{equation} \label{equ:steinidentity}
\mathbb{E}_{\pi_X}[\mathcal{A}_{\pi_X}g(X)] = \bm{0} 
\end{equation}
since $\nabla_x(\pi_X(x)g(x))=\big(\bm{u}(x)g(x) +\nabla_x g(x)\big)\pi_X(x)$. In addition, a vector-valued function
$\bm{g}(x) = [g_1(x), \cdots , g_{d'}(x)]$ is said to be in the Stein class of $\pi_X(x)$ if every $g_i$, $\forall i \in [d']$ is in the Stein class of $\pi_X(x)$.

Recall that a Hilbert space $\mathcal{H}$ with an inner product $\langle \cdot,\cdot \rangle : \mathcal{H}\times \mathcal{H}\to \mathbb{R}$ is a RKHS if there exists a symmetric positive definite function $k : \Omega \times \Omega \to \mathbb{R}$, called a (reproducing) kernel, such that for all $x \in \Omega$, we have $k(\cdot,x) \in \mathcal{H} $ and for all $x \in \Omega$ and $h \in \mathcal{H}$, we have $h(x) = \langle h(\cdot),k(\cdot,x) \rangle$. A kernel $k(x,x')$ is said to be in the Stein class of $\pi_X$ if $k(x,x')\in C^2(\Omega \times \Omega,\mathbb{R})$, and both $k(x,\cdot)$ and $k(\cdot,x')$ are in the Stein class of $\pi_X$ for any fixed $x,x'$. Note that if $k(x,x')$ is in the Stein class of $\pi_X$, so is any $h \in \mathcal{H}$ \citep{liu2016kernelized}. 

Suppose $k(x,x')$ is in the Stein class of $\pi_X$. %By first applying Stein’s operator on $k(x,x')$ with fixed $x'$ and subsequently with fixed $x$, we can get the following kernelized versions of Stein’s identity by the fact that
Then it is easy to see that $\nabla_{x'}(\pi_X(x') k(\cdot, x'))$ is also in the Stein class of $\pi_X$ for any $x'$ \citep{liu2016kernelized}:
$$\int_{\Omega} \nabla_x\left(\pi_X(x)\nabla_{x'}(\pi_X(x') k(x, x'))\right)dx =\nabla_{x'} \left(\pi_X(x') \int_{\Omega} \nabla_x(\pi_X(x)k(x, x'))dx \right) = \bm{0} \in \mathbb{R}^{d\times d}$$
since $k(\cdot, x')$ is in the Stein class of $\pi_X$ for any $x'$. Taking the trace of the matrix $$\text{trace}\left(\nabla_x\left(\pi_X(x)\nabla_{x'}(\pi_X(x') k(x, x'))\right)\right),$$
we get the following kernelized version of Stein’s identity (cf. \eqref{equ:steinidentity}):
\begin{equation} \label{equ:stein}
\mathbb{E}_{X\sim\pi_X}[k_0(X,x')] = 0, \quad \forall x'\in \Omega    
\end{equation}
where $k_0(x,x')$ is a new kernel function defined via 
$$k_0(x,x') := \nabla_x\cdot\nabla_{x'} k(x,x')+\bm{u}(x)\cdot \nabla_{x'}k(x,x')+\bm{u}(x')\cdot\nabla_x k(x,x')+ \bm{u}(x)\cdot \bm{u}(x')k(x,x').$$

Let $\mathcal{H}_0$ be the RKHS related to $k_0(x,x')$. Then all the functions $h(x)$ in $\mathcal{H}_0$ are orthogonal to $\pi_X(x)$ in the sense that $\mathbb{E}_{\pi_X}[h(X)] = 0$ \citep{LL}. This proposition is fundamental in the construction of Stein-kernelized CFs. It is easy to check that the commonly used radial basis function (RBF) kernel $k(x,x') = \exp(-\frac{1}{ h^2}\|x-x'\|^2_2)$ is in the Stein class for continuously differentiable densities supported on $\mathbb{R}^d$. 

In addition that the kernel $k$ is in the Stein class, we also assume the gradient-based kernel $k_0$ satisfies $\sup_{x\in \Omega} k_0(x,x) < \infty.$ In this paper, we always assume $\mathcal{H}_0$ satisfies these two conditions. For instance, the RBF kernel $k(x,x') = \exp(-\frac{1}{ h^2}\|x-x'\|^2_2)$ satisfies the above two conditions for continuously differentiable densities supported on $\mathbb{R}^d$ whose score function $\bm{u}(x)$ is of polynomial growth rate. %(CHECK IF THE REWRITING IS CORRECT). 

Next, let $\mathcal{C}$ denote the RKHS of constant functions with kernel $k_{\mathcal{C}}(x,x') = 1$ for all $x,x' \in \Omega$. The norms associated to $\mathcal{C}$ and $\mathcal{H}_0$ are denoted by $\|\cdot\|_\mathcal{C}$ and $\|\cdot\|_{\mathcal{H}_0}$ respectively.
$\mathcal{H}_+ = \mathcal{C} +\mathcal{H}_0$ denotes the set $\{c + \psi: c \in \mathcal{C}, \psi \in \mathcal{H}_0\}$. Equip $\mathcal{H}_+$ with the structure of a vector space, with addition operator $(c + \psi) + (c' + \psi') = (c + c') + (\psi + \psi')$ and multiplication operator $\lambda (c + \psi)  = (\lambda c) + (\lambda \psi)$, each well-defined due to uniqueness of the representation $f = c + \psi, f' = c' + \psi'$ with $c,c' \in \mathcal{C}$ and $\psi,\psi' \in \mathcal{H}_0$. It is known that $\mathcal{H}_+$ can be constructed as an RKHS with kernel $k_+(x,x') := k_{\mathcal{C}}(x,x') + k_0(x,x')$ and with norm $\|f\|_{\mathcal{H}_+}^2:=\|c\|_\mathcal{C}^2+\|\psi\|_{\mathcal{H}_0}^2$ \citep{berlinet2011reproducing}. Let $\kappa:=\sup_{x\in \Omega} \sqrt{k_+(x,x)}$. Note that $\kappa<\infty$  since $\sup_{x\in \Omega} k_0(x,x) < \infty$.

Finally, we remark that different choices of the kernel $k$ lead to different constructions of the RKHS $\mathcal{H}_+$, and may lead to different performance of subsequent approaches since their performance depends on the regularity of the ground-truth regression function in the space $\mathcal{H}_+$. An ideal requirement is that $\mathcal{H}_+$ should embody the ground-truth regression function%although this is hard to verify in practice as the regression function is unknown
.

%The classical control variates method constructs a new random variable with a smaller MSE. An alternative way is to construct a new function that has a smaller MSE then the original function $f$. This idea has been introduced by \citet{CF1} where they call their method control functionals. To be more precise, we call it Stein control functionals method. suggests to construct an auxiliary function based on the data that can lead to 

\subsection{Control Functionals}\label{sec:CF0}
Control functional (CF), or more precisely Stein-kernelized control functional, first introduced by \citet{CF1} and \citet{CF2}, is a systematically constructed class of control variates for variance reduction (In fact, CF can also partially perform bias reduction as we will show in Section \ref{sec:DRSK main}). %It reduces the MSE to a higher order than $O(n^{-1})$, which we will refer to as a supercanonical rate. 
%by utilizing knowledge on the analytical form of the gradient of the probability density function of $\pi$. %(up to normalizing constants)
Specifically, CF constructs a function $s_m(\cdot)$ applied on $X$ such that $\mu_X(s_m)$ is known, and that $f(X,Y)-s_m(X)$ has a very low variance so that the CF-adjusted sample
$$f(X,Y)-s_m(X)+\mathbb \mu_X(s_m)$$
is supercanonical for estimating $\mathbb E_\pi[f(X,Y)]$. This function $s_m(\cdot)$ is constructed as a functional approximation for $f(\cdot)$ by utilizing a training set of samples, where the function lies in the RKHS $\mathcal{H}_+$ constructed in Section \ref{sec:Stein} which has known mean under the distribution $\pi$. 
%that is induced by applying a Stein operator, with respect to $\pi$, to a general ``primary" RKHS. 
% It lies in a reproducing kernel Hilbert space (RKHS) which is a summation of two RKHSs: one RHKS of constant functions and one RKHS of functions that has mean zero under $\pi$. %with the special property that any element can be disintegrated into a constant and a function that has mean zero under $\pi$. 
% The latter RKHS is a consequence of applying a Stein operator, with respect to $\pi$, to a general ``primary" RKHS. %in order to obtain the approximation basis for $s_m(\cdot)$.

In more detail, following \citet{CF1}, we divide the data $D$ into two disjoint subsets as $D_0 = \{(x_i,y_i)\}^m_{i=1}$ and $D_1 = \{(x_i,y_i)\}^n_{i=m+1}$, where $1 \le m < n$. $D_0$ is used to construct a function $s_m(\cdot)\in L^2(\pi_X)$ that is a partial approximation to $f$ (that only depends on $x$). $s_m(x)$ is given by the following regularized least-square (RLS) functional regression on the RKHS $\mathcal{H}_+$, also known as kernel ridge regression (KRR):
$$s_m(x):= \mathop{\arg\min}_{g\in \mathcal{H}_+} \left\{ \frac{1}{m} \sum_{j=1}^{m} (f(x_j,y_j)-g(x_j))^2 +\lambda \|g\|^2_{\mathcal{H}_+} \right\}$$
where $\lambda>0$ is a regularization parameter which typically depends on the cardinality of the data set $D_0$. Note that in CF, the RKHS space $\mathcal{H}_+$ instead of $\mathcal{H}_0$ is the hypothesis class used in the KRR because $\mathcal{H}_0$ only contains functions with mean zero that cannot effectively approximate the function $f$ with a general mean. Consider the function
% e first $m < n$ samples to form a CF $s_m(X,Y)$, and the rest of $n-m$ samples to evaluate $f-s_m$ and output their average. Denote 
% The main differece between our setting and \citet{CF1} is that the reproducing kernel Hilbert space we construct here is only based on part of the random vectors, i.e. the $X$ part.\\
% First we separate $D = \{(x_i,y_i)\}^n_{i=1}$ into two 
% these Consider functions of the following form
\begin{equation} \label{equ:CF}
f_m(x,y)=f(x,y)-s_m(x)+\mu_X(s_m).    
\end{equation}
Then the CF estimator is given by a sample average of $f_m(\cdot,\cdot)$ on $D_1$, i.e.,
$$\hat{\theta}_{CF}:=\frac{1}{n-m} \sum_{j=m+1}^{n} f_m(x_j,y_j).$$
% Note that $D_0$ is a collection of independent and identically distributed samples.  Recall that $\mu_X(g)$ (or $\mu(g)$) is the expectation $\mathbb{E}_\pi[g(X,Y)|D_0]$ (or $\mathbb{E}_\pi[g(X)|D_0]$, respectively) with respect to the distribution of the random variable $(X,Y)$, conditional on $D_0$. , where $\mathbb{E}_\pi[\cdot|D_0]$ and $\mathbb E_\pi[\cdot]$ are with respect to the data distribution

If $(x_i,y_i)$ are i.i.d. drawn from the original distribution $\pi$, it is clear that we have unbiasedness, since $\mathbb{E}_\pi[\hat{\theta}_{CF}|D_0] =\mu(f_m)=\theta$ for any given $D_0$ and hence $\mathbb E_\pi[\hat{\theta}_{CF}]=\theta$. On the other hand, suppose our data $D$ are drawn from a distribution $q$ which may be different from the original distribution $\pi$. Then given $D_0$ (and thus given $s_m$), one component $f_m(X,Y)$ is in general biased for $\theta$ and the resulting CF by taking their average can suffer as a result.

It is well known in the KRR literature that the $s_m$ can be written explicitly in a closed form. We will review this result in Section \ref{sec:CF} and in particular, rewrite the existing closed-form solution from \citet{CF1} in terms of $k_+$
rather than $k_0$. In summary, CF is described in Algorithm \ref{CFalgorithm}. In addition, \citet{CF1} suggested a simplified version of CF, called the \textit{simplified CF estimator}, which is simply the $\mu_X(s_n)$ by setting $m=n$ in Algorithm \ref{CFalgorithm}. This estimator is described in Algorithm \ref{SimCFalgorithm} for the sake of completeness.

\begin{algorithm} \label{CFalgorithm}
\SetAlgoLined
%\KwResult{Write here the result }
\textbf{Goal}: Estimate $\theta:=\mathbb{E}_{\pi}[f(X,Y)]$;\\
\textbf{Input}: A set of i.i.d. samples $D=\{(x_j,z_j=f(x_j,y_j))\}_{j=1,\cdots,n}$ drawn from some distribution $q$, the reproducing kernel $k_0$ of $\mathcal{H}_0$ and $k_+=k_0+1$ of $\mathcal{H}_+$;\\
\textbf{Procedure}: 
(1) We divide the dataset $D$ into two disjoint subsets as $D_0 = \{(x_i,f(x_i,y_i))\}^m_{i=1}$ and $D_1 = \{(x_i,f(x_i,y_i))\}^n_{i=m+1}$, where $1 \le m < n$;\\
(2) $D_0$ is used to construct the CF-adjusted sample
$$f_m(x,y)=f(x,y)-s_m(x)+\mu_X(s_m).$$
Let
$$\hat{z}=(f(x_1,y_1),\cdots, f(x_m,y_m))^T,$$
$$K_+=(k_+(x_i,x_j))_{i,j =1,\cdots, m},$$
$$\hat{k}_+(x)=(k_+(x_1,x),\cdots, k_+(x_m,x))^T,$$
$$\beta =(K_+ +\lambda m I)^{-1} \hat{z}.$$
Then $s_m(x)= \beta^T \hat{k}_+(x)$ and
$\mu_X(s_m)=\beta^T \mathbf{1}$.\\
\textbf{Output}: CF estimator
$$\hat{\theta}_{CF}:=\sum_{j=m+1}^{n} \frac{1}{n-m} f_m(x_j,y_j).$$
 \caption{Stein-Kernelized Control Functional (CF)}
\end{algorithm}

\begin{algorithm} \label{SimCFalgorithm}
\SetAlgoLined
%\KwResult{Write here the result }
\textbf{Goal}: Estimate $\theta:=\mathbb{E}_{\pi}[f(X,Y)]$;\\
\textbf{Input}: A set of i.i.d. samples $D=\{(x_j,z_j=f(x_j,y_j))\}_{j=1,\cdots,n}$ drawn from some distribution $q$, the reproducing kernel $k_0$ of $\mathcal{H}_0$ and $k_+=k_0+1$ of $\mathcal{H}_+$;\\
\textbf{Procedure}: 
Let
$$\hat{z}=(f(x_1,y_1),\cdots, f(x_n,y_n))^T,$$
$$K_+=(k_+(x_i,x_j))_{i,j =1,\cdots, n},$$
$$\hat{k}_+(x)=(k_+(x_1,x),\cdots, k_+(x_n,x))^T,$$
$$\beta =(K_+ +\lambda n I)^{-1} \hat{z}.$$
Then
$\mu_X(s_n)=\beta^T \mathbf{1}$.\\
\textbf{Output}: Simplified CF estimator
$$\hat{\theta}_{SimCF}:=\mu_X(s_n).$$
\caption{Simplified CF Estimator (SimCF)}
\end{algorithm}
 %(See Section \ref{sec:CF} for details.) 

\subsection{Black-Box Importance Sampling}\label{sec:BBIS0}
Another method derived from the kernelization of Stein's identity is black-box importance sampling (BBIS) introduced by \citet{LL}, which has been shown to be an effective approach for both variance and bias reduction. This method can be viewed as a relaxation of conventional IS, in that it assigns weights over the samples assuming that the data are drawn from a sampling distribution $q$ in the \textit{black box}, i.e., with no knowledge on the analytical form of $q$. In this setting, standard importance weights cannot be calculated because the likelihood ratio is unknown. In the situation where the noise $Y$ does not appear, the black-box importance weights are optimized based on a convex quadratic minimization problem: 
\begin{equation}
\mathop{\arg\min}_{w}\{w^TK_0w:w\ge 0, w^T\textbf{1}=1\}\label{BBIS explain}
\end{equation}
where $\textbf{1}=(1,1,\cdots,1)$ and $K_0$ is the kernel matrix  with respect to the RKHS $\mathcal{H}_0$ constructed in Section 2.2 which contains functions with mean zero under the distribution $\pi_X$.
%induced by the same RKHS of zero-mean functions as in the CF.
Suppose $w$ is the (unknown) likelihood ratio between $\pi_X$ and $q_X$, i.e., $w_i=\frac{1}{n}\frac{\pi_X(x_i)}{q_X(x_i)}$. Then
$$\mathbb{E}_{x_i\sim q_X} [w_ik_0(x_i,x_j)w_j]= \frac{1}{n}\mathbb{E}_{x_i\sim \pi_X} [k_0(x_i,x_j)w_j] = 0 $$
for any $x_j$ since the kernel function has mean zero under $\pi_X$; see \eqref{equ:stein}. In addition, we have $\mathbb{E}_{q_X}[w^T\textbf{1}]=1$. Therefore the quadratic form achieves the minimum mean $0$ under $q_X$ when $w$ is the likelihood ratio, which intuitively justifies \eqref{BBIS explain}.%, and as a result the black-box importance weights obtained from \eqref{BBIS explain} approximate the true importance weights 
% (ADD MORE EXPLANATION WHY? CURRENTLY IT'S UNCLEAR HOW THIS IS DEDUCEABLE). 

%To address the issue of biased generating distribution, a potential way is to use an importance sampling method. 
% Following \citet{LL}, we review the black-box importance sampling technique as follows. $\mathbb{S}(\{x_j,w_j\},\pi)$ is is used to construct the black-box importance weights $w$ that minimize the so-called
 
We provide more details in the general situation. First, we define the empirical kernelized Stein discrepancy (KSD) between the weighted empirical distribution of the samples and $\pi_X$:
$$\mathbb{S}(\{x_j,w_j\},\pi_X)=\sum_{j,k\in D} w_j w_k k_0(x_j,x_k) = w^T K_0 w.$$
where $$K_0=(k_0(x_j,x_k))_{j,k \in D}$$ is the $n\times n$ kernel matrix constructed on the entire data set $D$. The black-box importance weights then form the optimal solution to the following convex quadratic optimization problem (which can be solved efficiently):
\begin{align} \label{COP0}
\hat{w}= \mathop{\arg\min}_{w} \left\{\mathbb{S}(\{x_j,w_j\},\pi_X), \text{ s.t. } \sum_{j=1}^{n} w_j=1,
0\le w_j \le \frac{B_0}{n}\right\}
\end{align}
where $B_0$ is a pre-specified bound that will be provided in the subsequent theorems. Here, the upper bound $B_0$ on the weights is a new addition compared to the original formulation in \citet{LL}, and is needed to control the MSE when there is the noise term $Y$. %(See Section \ref{sec:IS} for details.)
The BBIS estimator is given by a weighted average of $f(\cdot,\cdot)$ on $D$, i.e.,
$$\hat{\theta}_{IS}:=\sum_{j=1}^{n} \hat{w}_j f(x_j,y_j).$$
In summary, BBIS is described in Algorithm \ref{ISalgorithm}.

\begin{algorithm} \label{ISalgorithm}
\SetAlgoLined
%\KwResult{Write here the result }
\textbf{Goal}: Estimate $\theta:=\mathbb{E}_{\pi}[f(X,Y)]$;\\
\textbf{Input}: A set of i.i.d. samples $D=\{(x_j,z_j=f(x_j,y_j))\}_{j=1,\cdots,n}$ drawn from some distribution $q$, the reproducing kernel $k_0$ of $\mathcal{H}_0$;\\
\textbf{Procedure}: 
% $D$ is used to construct the importance weights $\hat{w}$. 
Let $$K_0=(k_0(x_i,x_j))_{i,j =1,\cdots, n},$$
and $\hat{w}$ is the optimal solution to the following quadratic optimization problem
\begin{align}
    \hat{w}= \mathop{\arg\min}_{w} \left\{w^T K_0 w, \text{ s.t. } \sum_{j=1}^{n} w_j=1, 0\le w_j \le \frac{B_0}{n}\right\}.
\end{align}
If the noise $Y$ does not exist, we simply set $B_0=+\infty$. Otherwise, set $B_0=2B$ for the canonical rate and $B_0=4B$ for a supercanonical rate (see Assumption \ref{lrupper} for the definition of $B$ and Theorem \ref{mainIS0} for details).\\ 
\textbf{Output}: BBIS estimator
$$\hat{\theta}_{IS}:=\sum_{j=1}^{n} \hat{w}_j f(x_j,y_j).$$
\caption{Modified Black-Box Importance Sampling (BBIS)}
\end{algorithm}

\subsection{Related Work}\label{sec:literature}
To close this section, we discuss some other related literature. Variance and bias reduction has been a long-standing topic in Monte Carlo simulation. Two widely used methods are
control variates and importance sampling (Chapter 4 in \citet{Gl}, Chapter 5 in \citet{asmussen2007stochastic}, Chapter 5 in \citet{rubinstein2016simulation}). The control variate method reduces variance by adding an auxiliary variate with known mean to the naive Monte Carlo estimator \citep{nelson1990control}. The construction of this auxiliary variable can follow multiple approaches. In the classical setup, a fixed number of control variates are linearly combined, with the linear coefficients constructed by ordinary least squares \citep{glynn2002some}. This approach maintains the canonical rate in general. Beyond this, one can add a growing number of control variates (relative to the sample size) to get a faster rate \citep{portier2019monte}. The linear coefficients can also be obtained by using regularized least squares to increase accuracy \citep{south2022regularized,leluc2021control}. These methods require a pre-specified collection of well-behaved control variates (e.g., the control variates are linearly independent or dense in a function space), which may not always be easy to construct in practice. 
% They do not focus on the practical construction of those control variates.

%A central limit theorem when the number of control variables tends to infinity is established in \citet{portier2019monte}. However, this assumption is too restrictive for practical applications. 
To generalize the linear form and possibly obtain a supercanonical rate, the control variate can be constructed via a fitted function learned from data. This function can be constructed using adaptive control variates \citep{henderson2002approximating,henderson2004adaptive,kim2007adaptive}. In order to have a supercanonical 
rate for adaptive control variate estimators, one of the key assumptions in these papers is the existence of a “perfect” control variate (i.e., a  control variate with zero variance), and the “perfect” control variate can be approximated in an adaptive scheme. This assumption could be unlikely
to hold for some practical applications \citep{kim2007adaptive}. Another approach to construct the function is via $L^2$ functional approximation \citep{maire2003reducing}. 
%However, to obtain a supercanonical convergence rate, \citet{maire2003reducing} assumes that the $L^2$ expansion coefficients of the target function decrease at a polynomial rate, which is unlikely to hold for high-dimensional functions.
\citet{maire2003reducing} only focuses on the mono-dimensional function whose $L^2$ expansion coefficients decrease at a polynomial rate. It is unknown how to extend this work to multi-dimensional functions. 
Finally, the function can also be constructed, as described earlier, via RLS regression \citep{CF1} which provides a systematic approach to construct control variates based on the kernelization of Stein's identity. Compared with previous methods, \citet{CF1} require less restrictive assumptions for supercanonical rates and avoid adaptive tuning procedures.
% (I'VE REVISED A BIT THIS PARAGRAPH; SEE IF CORRECT..). 

% control variate takes a linear-parameterized form, which
Standard IS reduces variance by using an alternate proposal distribution to generate samples, and multiplying the samples by importance weights constructed from the likelihood ratios, or the Radon-Nikodym derivative, between the proposal and original distributions. Likewise, it can also be used to de-bias estimates through multiplying by likelihood ratios if the generating distribution is biased. IS has been shown to be powerful in increasing the efficiency of rare-event simulation \citep{bucklew2013introduction,rubino2009rare,JUNEJA2006291,BLANCHET201238}. In contrast to conventional methods, BBIS does not require the closed form of the likelihood ratio. We also note that both CF and BBIS can be viewed as versions of the weighted Monte Carlo method since both of them can be written in the form of a linear combination of $f(x_i)$, while their weights are obtained in different ways. Other methods to construct weighted Monte Carlo can be found in, e.g., \citet{doi:10.1287/opre.1040.0148,owen2000safe}. 

%For example, some of them have been discussed in Chapter 4 in \citet{Gl}. There are two aspects about eliminating MSE, namely variance reduction and biased reduction (and sometimes they overlap). 

%- lit review: how about dividing into three parts: 1) variance reduction techniques in MC including BBIS and CV (In addition to the references you cited, the following could also be put: Asmussen & Glynn's book (and chapter), Glasserman's book, Kroese's book, survey papers on BBIS like Shahabuddin & Juneja and Blanchet & Lam, CV includes the seminal papers one by Glynn and one by Barry Nelson. let me know if you couldn't find them.. 2) RKHS techniques, including regularized regression and also CF and BBIS; 3) doubly robust estimator (used in off-policy RL and covariate shift; our math form is similar to them..)

%CF is a recent  technique first proposed by \citet{CF1}, which can be viewed as a highly efficient control variate by suitably choosing a function on the input variates. For instance, suppose we want to estimate $\mathbb E_\pi[f(X)]$ by Monte Carlo simulation, where $\mathbb E_\pi[\cdot]$ denotes the expectation for $X$ under distribution $\pi$. In a more general (and more natural) setting, the input space may be not compact and data could be dependent (CHECK IF THE REWRITING IS CORRECT..). Recent development shows that ``kernelizing" Stein’s identity goes far beyond the area of CF. Relevant work includes Stein control functionals
 
Besides the Monte Carlo literature, CF utilizes a combination of two ideas: kernel ridge regression (KRR) and the Stein operator. The theory of KRR has been well developed in the past two decades. Its modern learning theory has been proposed in \citet{CS} and \citet{CS2}, and further strengthened in \citet{SZ}, \citet{SZ2} and \citet{SZ3}. \citet{CZ} provide comprehensive documentation on this topic. Most of these works assume that the input space is a compact set. \citet{SW} further study KRR on non-compact metric spaces, which provides the mathematical foundation of this paper. There are multiple studies on extending the standard KRR. For instance, \citet{SW2} study KRR with dependent samples. \citet{christmann2007consistency,debruyne2008model} study consistency and robustness of kernel-based regression. To relieve the computation cost of KRR estimators for large datasets, %Thus, the design of scalable methods for KRR (and other kernel based methods) has been the focus of intensive research in recent years. One of the most popular approaches to scaling up kernel based methods is random Fourier features sampling,  
\citet{RR} propose the random Fourier feature sampling to speed up the evaluation of the kernel matrix, and \citet{pmlr-v30-Zhang13} propose a divide-and-conquer KRR to decompose the computation.

The second idea used by CF is the kernelization of Stein’s identity, i.e., applying the Stein operator to a ``primary" RKHS. The resulting RKHS automatically satisfies the zero-mean property under $\pi$ which lays the foundation for constructing suitable control variates. This idea has been used and followed up in \citet{CF1,CF2,HZ,south2022postprocessing}, and finds usage beyond control variates, including the Stein variational gradient descent \citep{liu2016stein,liu2017stein,liu2017stein2,han2018stein,wang2019stein}, the kernel test for goodness-of-fit \citep{CSG, liu2016kernelized}
and BBIS \citep{LL,hodgkinson2020reproducing} that we have described earlier.

% , and . This particular reproducing kernel Hilbert space is used in the first idea to replace a general RKHS. This technique, first introduced in \citet{CF1}, has gained extensive attention recently. 

\section{Doubly Robust Stein-Kernelized Estimator} \label{sec:DRSK main}
%In the light of both Stein control functionals technique and black-box importance sampling technique, 
We propose an enhancement of CF and BBIS that can simultaneously perform both variance and bias reduction. We call it the \textit{doubly robust Stein-kernelized (DRSK) estimator}. In brief, %the DRSK estimator uses $D_0$ to construct control functionals and $D_1$ to construct black-box importance weights.  
we divide the data $D$ into two disjoint subsets as $D_0 = \{(x_i,y_i)\}^m_{i=1}$ and $D_1 = \{(x_i,y_i)\}^n_{i=m+1}$, where $1 \le m < n$. Based on the first subset $D_0$, we construct the same regression function $s_m(\cdot)\in L^2(\pi_X)$ to derive the CF-adjusted sample $f_m(x,y)$ as in Section \ref{sec:CF0} for variance reduction. Based on the second subset $D_1$, we generate the black-box importance weights $\hat{w}$ as in Section \ref{sec:BBIS0} for bias reduction. Finally, we take the weighted average of $f_m(x_j,y_j)$ with weights $\hat{w}$ on $D_1$ to get the DRSK estimator. The detailed procedure is described in Algorithm \ref{DRSKalgorithm}.

The terminology ``doubly robust" in DRSK is borrowed from doubly robust estimators in off-policy learning \citep{dudik2011doubly,dudik2014doubly,jiang2016doubly,farajtabar2018more}. In fact, we can rewrite the DRSK estimator as
\begin{align}
\hat{\theta}_{DRSK}:=&\sum_{j=m+1}^{n} \hat{w}_j f_m(x_j,y_j)\nonumber\\
=&\sum_{j=m+1}^{n} \hat{w}_j (f(x_j,y_j)-s_m(x_j)+\mu_X(s_m))\nonumber\\
=&\mu_X(s_m) + \sum_{j=m+1}^{n} \hat{w}_j (f(x_j,y_j)-s_m(x_j))\label{equ:doublyrobust}
\end{align}

The doubly robust estimator \citep{dudik2014doubly} is known to be a combination of two approaches: direct method (DM) and inverse propensity score (IPS). The second term in \eqref{equ:doublyrobust} is an importance-sampling weighted average of the residuals from the regression, which is similar to the part of IPS in doubly robust estimators. The first term in \eqref{equ:doublyrobust} is similar to the DM by using $s_m$ as an approximation of $f$:
$$\sum_{j=m+1}^n \frac{1}{n-m}s_m(x_j)$$
except that in our setting, there is no need to estimate the expectation of $s_m$ under $\pi$ since $s_m \in \mathcal{H}_+$ has a known expectation by our construction. Note that the first term in \eqref{equ:doublyrobust} is essentially the simplified CF estimator suggested by \citet{CF1} (Algorithm \ref{SimCFalgorithm}). In this sense, CF is similar to DM.

\begin{algorithm} \label{DRSKalgorithm}
\SetAlgoLined
%\KwResult{Write here the result }
\textbf{Goal}: Estimate $\theta:=\mathbb{E}_{\pi}[f(X,Y)]$;\\
\textbf{Input}: A set of i.i.d. samples $D=\{(x_j,z_j=f(x_j,y_j))\}_{j=1,\cdots,n}$ drawn from some distribution $q$, the reproducing kernel $k_0$ of $\mathcal{H}_0$ and $k_+=k_0+1$ of $\mathcal{H}_+$;\\
\textbf{Procedure}: 
(1) We divide the dataset $D$ into two disjoint subsets as $D_0 = \{(x_i,f(x_i,y_i))\}^m_{i=1}$ and $D_1 = \{(x_i,f(x_i,y_i))\}^n_{i=m+1}$, where $1 \le m < n$;\\
(2) $D_0$ is used to construct the CF-adjusted sample
$$f_m(x,y)=f(x,y)-s_m(x)+\mu_X(s_m).$$
Let
$$\hat{z}=(f(x_1,y_1),\cdots, f(x_m,y_m))^T,$$
$$K_+=(k_+(x_i,x_j))_{i,j =1,\cdots, m},$$
$$\hat{k}_+(x)=(k_+(x_1,x),\cdots, k_+(x_m,x))^T,$$
$$\beta =(K_+ +\lambda m I)^{-1} \hat{z}.$$
Then $s_m(x)= \beta^T \hat{k}_+(x)$ and
$\mu_X(s_m)=\beta^T \mathbf{1}$.
(See Section \ref{sec:CF} for details.)\\
(3) $D_1$ is used to construct the importance weights $\hat{w}$. Let $$K_0=(k_0(x_i,x_j))_{i,j =m+1,\cdots, n},$$
and $\hat{w}$ is the optimal solution to the following quadratic optimization problem
\begin{align} \label{COP3}
    \hat{w}= \mathop{\arg\min}_{w} \left\{w^T K_0 w, \text{ s.t. } \sum_{j=m+1}^{n} w_j=1, 0\le w_j \le \frac{B_0}{n-m}\right\}
\end{align}
If the noise $Y$ does not exist, we simply set $B_0=+\infty$. Otherwise, set $B_0=2B$ for the canonical rate and $B_0=4B$ for a supercanonical rate (see Assumption \ref{lrupper} for the definition of $B$ and Theorems \ref{mainDRSK0} and \ref{mainDRSK0-2} for details).\\   
\textbf{Output}: DRSK estimator
$$\hat{\theta}_{DRSK}:=\sum_{j=m+1}^{n} \hat{w}_j f_m(x_j,y_j).$$
\caption{Doubly Robust Stein-Kernelized Estimator (DRSK)}
\end{algorithm}

\subsection{Main Findings and Comparisons} \label{sec:discussion}
We summarize our main findings and comparisons of DRSK with the existing CF and BBIS methods. To facilitate our comparisons, recall that in Section \ref{sec:PS} we have proposed a general problem setting where input samples are both partially known (meaning we have a noise term $Y$) and biased (meaning $(X,Y)$ may not be drawn from the original distribution $\pi$). Here, we elaborate and consider the following four scenarios in roughly increasing level of complexity, where the last case corresponds to the general setting introduced earlier:\\
(1) ``Standard": there is no noise term $Y$ and $X$ is drawn from $\pi$.\\
(2) ``Partial": there is a noise term $Y$ and $(X,Y)$ is drawn from $\pi$.\\
(3) ``Biased": there is no noise term $Y$ and $X$ is drawn from $q$.\\
(4) ``Both": there is a noise term $Y$ and $(X,Y)$ is drawn from $q$.

We will frequently use the abbreviations ``Standard", ``Partial", ``Biased", ``Both" to refer to each case. %In Section \ref{sec:analysis}, we establish multiple results of the MSE convergence rate when we apply CF estimator, BBIS estimator and DRSK estimator in these scenarios respectively.
Tables \ref{Pro} summarizes the MSE rates of three different methods in each scenario with some common assumptions specified in Section \ref{sec:assumptions}. In particular, the ground-truth regression function $\bar{f}:=\mathbb{E}_\pi[f(X,Y)|X=x]\in \text{Range}(L_q^{r})$ with $\frac{1}{2}\le r\le 1$ indicating the regularity of $\bar{f}$ in the space $\mathcal{H}_+$ where the positive self-adjoint operator $L_q$ is formally defined later in \eqref{equ:integraloperator}. $M_0$ is given in Assumptions \ref{noise1} that is the bound on the noise level of $Y$ given $X$. %(See Section \ref{sec:analysis} for the details of assumptions.)

\begingroup
\tabcolsep = 3.0pt
\begin{table}
%\scriptsize
\begin{center}
\caption{This table displays the MSE rates of three different methods in each scenario with some common assumptions specified in Section \ref{sec:assumptions}. The ground-truth regression function $\bar{f}:=\mathbb{E}_\pi[f(X,Y)|X=x]\in \text{Range}(L_q^{r})$ with $\frac{1}{2}\le r\le 1$. $M_0$ is given in Assumptions \ref{noise1}.} \label{Pro} 
%\small
\begin{tabular}{|c|c|c|c|c|}
\hline 
MSE & Standard & Partial & Biased & Both  \\ \hline
CF (Ass. \ref{CSA}, \ref{noise1}, \ref{lrupper2}) & $O(n^{-1-r})$ & $O(n^{-1-r})+M_0 n^{-1}$ & $O(n^{-r})$ & $O(n^{-r})+M_0 $ \\ \hline 
BBIS (Ass. \ref{CSA}, \ref{noise1}, \ref{lrupper}) & $O(n^{-1})$ & $O(n^{-1})$ & $O(n^{-1})$ & $O(n^{-1})$ \\ \hline 
BBIS (Ass. \ref{CSA}, \ref{noise1}, \ref{lrupper}, \ref{lrupper4}) & $o(n^{-1})$ & $o(n^{-1})+M_0n^{-1}$ & $o(n^{-1})$ & $o(n^{-1})+M_0n^{-1}$\\ \hline 
DRSK (Ass. \ref{CSA}, \ref{noise1}, \ref{lrupper}) & $O(n^{-\frac{1}{2}-r})$ &  $O(n^{-\frac{1}{2}-r})+M_0 n^{-1}$ & $O(n^{-\frac{1}{2}-r})$ & $O(n^{-\frac{1}{2}-r})+M_0 n^{-1}$ \\ \hline 
DRSK (Ass. \ref{CSA}, \ref{noise1}, \ref{lrupper}, \ref{lrupper4}) & $o(n^{-\frac{1}{2}-r})$ &  $o(n^{-\frac{1}{2}-r})+M_0 n^{-1}$ & $o(n^{-\frac{1}{2}-r})$ & $o(n^{-\frac{1}{2}-r})+M_0 n^{-1}$ \\ \hline 
\end{tabular}

\end{center}
\end{table}

Table \ref{Pro} conveys the following:
\begin{enumerate}[leftmargin=*]
\item Except that the noise part remains at the canonical rate, CF has a supercanonical rate in the ``Standard" and ``Partial" cases, but subcanonical in the ``Biased" and ``Both" cases when $r<1$. In the following, the supercanonical and subcanonical rates are referred to as the property on the ``dominating" factor $X$ with the convention that the noise $Y$ is at the canonical rate.

\item BBIS in all cases has the canonical rate under a weak assumption and a supercanonical rate under a strong assumption (Assumption \ref{lrupper4}).

\item DRSK always has a supercanonical rate either when $r>\frac{1}{2}$ or under a strong assumption (Assumption \ref{lrupper4}). 

\item Suppose $r>\frac{1}{2}$. Then DRSK is strictly faster than CF and BBIS in the ``Biased" and ``Both" cases, under both weak and strong assumptions. Moreover, DRSK is strictly faster than BBIS in any case, under both weak and strong assumptions.
\end{enumerate}

CF can handle extra noise quite well (in the ``Standard" and ``Partial" cases) since it takes advantage of the functional approximation of $f$, but only partially reduce bias. %Note that CF relies mostly on functional approximation of $f$, but less on the distribution of $X$.
%CF performs very well in the ``Standard" and ``Partial" cases. %We will demonstrate that the CF estimator is effective for variance reduction but not very effective for biased reduction. 
In the ``Biased" and ``Both" cases, %CF exhibits a similar rate as a single component $f_m(X,Y)$ in \eqref{equ:CF}. 
a single component $f_m(X,Y)$ (with finite $m$) in CF is generally a biased estimator of $\theta$.
The uniform weight $\frac{1}{n-m}$ in the final step of constructing CF cannot reduce the bias in $f_m-\theta$ effectively, %and does not contribute to reducing the MSE rate effectively. 
which leads to underperformance in the ``Biased" and ``Both" cases. Therefore, the simplified CF estimator (Algorithm \ref{SimCFalgorithm}) could be a better alternative by omitting the final step, as recommended by \citet{CF1}. On the other hand, our results also show that a single $f_m(X,Y)$ is an asymptotically unbiased estimator of $\theta$ (at the rate of $O(m^{-r})$ when $m\to \infty$), indicating the bias reduction perspective of CF.
% On the other hand, Theorem \ref{mainCF0-3} indicates that $f_m$ works for biased generating distribution so the issue lies in the final step of CF construction.
%  (Theorem \ref{mainIS0-1})

BBIS performs efficiently for bias reduction in general, although no higher order than $o(n^{-1})$ is guaranteed theoretically.
%is effective for biased reduction but not very effective for variance reduction.
The validity of reducing the MSE in the BBIS estimator is entirely due to the black-box importance weights, ignoring the information of the function $f$ and the output data $f(x_i,y_i)$. This implies that a more ``regular" function $f$ in the RKHS may not be able to improve the rate in BBIS as CF does. Besides, %BBIS does not guarantee better performance in the unbiased case than in the biased case while CF can, and 
the original BBIS estimator faces an additional challenge of not controlling the noise term (though the latter is not shown in Table \ref{Pro}).

Therefore, CF and BBIS both encounter difficulties when applying in the ``Both" case. %One way of improvement is to combine these two techniques. As a result, we propose our DRSK estimator is introduced, which can be viewed as a joint approach of control functionals and importance sampling.
%As a result, we propose our DRSK estimator as a joint approach of CF and BBIS.  
Our DRSK estimator improves CF and BBIS by taking advantage of both estimators. %constructing suitable weights and simultaneously using the information of $f$ in a CF.
The weighting part in DRSK utilizes knowledge of $\pi_X$ to diminish the bias as in BBIS, and the control functional part in DRSK utilizes the information of $f(x_j,y_j)$ to learn a more ``concentrated" function than $f(x)$ as in CF. 
As shown in Table \ref{Pro}, it can reduce the overall variance and bias efficiently. %than either of the existing methods alone. 
%Moreover, DRSK features a supercanonical rate in all cases. 
% To the best of our knowledge, this is the first estimator with such feature in our general problem setting. 

%\subsection{Main Theorems}\label{sec:maintheorem}

%This section presents main theorems rigorously. 
\subsection{Assumptions} \label{sec:assumptions}
To present our main theorems rigorously, we will employ the following assumptions.
\begin{assumption} [Covariate shift assumption]
$\pi_{Y|X}(y|x)=q_{Y|X}(y|x)$.\label{CSA}
\end{assumption}
%This assumption is valid when, for instance, $X$ is independent of $Y$. 
Here we do not require $\pi_{Y|X}(y|x)$ to be known or to be a continuous probability distribution. Covariate shift assumption holds, for instance, 1) when $X$ is independent of $Y$, or 2) in stochastic simulation problems where the aleatory noise in the simulation output $Y$, conditional on the input parameter $X$, is not affected by the epistemic noise incurred by the input parameter estimation (e.g., $X$ is estimated via historical data independent of the simulation model). Example \ref{example} in Section \ref{sec:intro} and the computer communication network in Section \ref{sec:multipledistributions} provide concrete examples where covariate shift assumption holds.  Though not directly relevant to our work, we note that the assumption is standard in transfer learning or covariate shift problems \citep{gretton2009covariate,YS,KM,lam2019robust}. 
We remark that under this assumption, the ground-truth regression function $f_\pi(x):=\mathbb{E}_\pi[f(X,Y)|X=x]$ and $f_q(x):=\mathbb{E}_q[f(X,Y)|X=x]$ are the same so we denote it as $\bar{f}$. We will use crucially the decomposition 
$$f(X,Y)=\bar f(X)+\epsilon(X,Y)$$
where $\bar{f}(X)$ can be viewed as the contribution of the fluctuation on $f$ from $X$, and $\epsilon(X,Y)=f(X,Y)-\bar{f}(X)$ is the error term. Note that, by definition, we have 
$$\mathbb{E}[\epsilon(X,Y)]=0, \quad \mathbb{E}[\epsilon(X,Y)|X]=0, \quad \mathbb{E}[\epsilon(X,Y)\bar f(X)]=0$$
where the expectation can be taken with respect to $\pi$ or $q$ under Assumption \ref{CSA}. Next, we introduce a basic assumption on the error term.

\begin{assumption}
$\mathbb{E}_q[\epsilon(X,Y)^2]\le M_0<\infty$.
\label{noise1}
\end{assumption}
%\begin{assumption}
%$\mathbb{E}_q[\epsilon(X,Y)^2|X=x]\le M_0<\infty, \quad \forall x\in \Omega.$
%\label{noise2}
%\end{assumption}
%Note that Assumption \ref{noise1} implies Assumption \ref{noise1} and
%$M_0\le M_0$. With Assumption \ref{CSA}, Assumption \ref{noise1} is equivalent to $$\mathbb{E}_\pi[\epsilon(X,Y)^2|X=x]=\mathbb{E}_q[\epsilon(X,Y)^2|X=x]\le M_0.$$

%\begin{assumption}  $\bar{f} \in (\mathcal{H}_+)_1^q$.\label{A8replace}
%\end{assumption}

The following assumptions are considered in \citet{LL} for the biased input samples. 
%\begin{assumption}
%$\mathbb{E}_\pi[\pi_X / q_X]<\infty$\label{lrupper2}.
%\end{assumption}

%Note that Assumption \ref{lrupper} implies Assumption
%\ref{lrupper2}.

\begin{assumption}
%$\frac{\pi_X(x)}{q_X(x)} > 0$ for $\forall x \in \Omega$ and 
$\mathbb{E}_{x\sim q_X}[(\frac{\pi_X(x)}{q_X(x)})^2]=\mathbb{E}_{x\sim \pi_X}[\frac{\pi_X(x)}{q_X(x)}]<\infty$.
\label{lrupper3} \label{lrupper2}
\end{assumption}

\begin{assumption}
$\frac{\pi_X(x)}{q_X(x)}\le B<\infty \quad \forall x\in \Omega$. \label{lrupper}
\end{assumption}

Note that this assumption implies that $$\mathbb{E}_{x\sim q_X}\left[(\frac{\pi_X(x)}{q_X(x)})^2 k_0(x,x)\right]<\infty,$$  
$$\mathbb{E}_{x,x'\sim q_X}\left[\left(\frac{\pi_X(x)}{q_X(x)} \frac{\pi_X(x')}{q_X(x')}k_0(x,x')\right)^2\right]<\infty.$$
The reason is that we assume $\sup_{x\in \Omega} k_0(x,x) < \infty$ in our construction of $\mathcal{H}_0$ and we note the fact that $|k_0(x,x')| \le (k_0(x,x)k_0(x',x'))^{\frac{1}{2}}.$
Therefore, Assumption \ref{lrupper3} is the same as Assumption B.1 in \citet{LL}. Moreover, Assumption \ref{lrupper} implies Assumption \ref{lrupper3}.

\begin{assumption}
%Let $\{\phi_l\}_{l=1}^{\infty}$ be the set of orthogonal eigenfunctions of $k_0(x,x')$ with respect to $\pi_X(x)$ with positive eigenvalues $\{\lambda_l\}_{l=1}^{\infty}$:
Suppose $k_0(x,x')$ has the following eigen-decomposition
$$k_0(x,x')=\sum_{l=1}^{\infty} \lambda_l \phi_l(x) \phi_l(x'),$$
where $\{\lambda_l\}_{l=1}^{\infty}$ are the positive eigenvalues sorted in non-increasing order, and $\{\phi_l\}_{l=1}^{\infty}$ are the eigenfunctions orthonormal w.r.t. the distribution $\pi_X(x)$, i.e., $
\mathbb{E}_{x\sim \pi_X}[\phi_l(x) \phi_{l'}(x)]=\textbf{1}_{l=l'}$. We assume that
$\tr(k_0(x,x'))=\sum_{l=1}^{\infty}\lambda_l<\infty$ and $\sup_{x\in \Omega,l}|\phi_l(x)|<\infty$.
%(3) $\text{var}_{x\sim q}[(\frac{\pi_X(x)}{q_X(x)})^2\phi_l(x)\phi_{l'}(x)]\le M_2$ for all $l$ and $l'$;\\ $\sup_{x\in \Omega,l}|\phi_l(x)|<\infty$ for all $l$ and $x\in\Omega$.
\label{lrupper4}
\end{assumption}
Assumption \ref{lrupper} plus Assumption \ref{lrupper4} is the same as Assumption B.4 in \citet{LL}. In particular, we notice that 
$$\text{var}_{x\sim q_X}\Big[(\frac{\pi_X(x)}{q_X(x)})^2\phi_l(x)\phi_{l'}(x)\Big]\le \Big(\sup_{x\in\Omega} \frac{\pi_X(x)}{q_X(x)}\Big)^4  \Big(\sup_{x\in \Omega,l}|\phi_l(x)|\Big)^4. $$
To simplify notations, we denote $M_2$ as the upper bounds in Assumptions \ref{lrupper} and \ref{lrupper4}, i.e., 
% (I DON'T GET WHAT WE TRY TO DO HERE.. DO WE NEED $M_2$ WHICH IS NOT IN OUR ASSUMPTION, AND ASSUMPTIONS 5 AND 6 ARE OUR ASSUMPTIONS OR LIU'S ASSUMPTIONS?) {\color{blue} In fact, Assumption 4, 5, 6 are all Liu's assumptions. 4 is their assumption B.1, and 5+6 is their assumption B.2. I separate 5 and 6 to develop new results because of the noise Y. $M_2$ is just the maximum value of all the upper bounds in assumption 4-6, a uniform value introduced only for the convenience of the proof}, i.e.,
$$\sup_{x\in\Omega} \frac{\pi_X(x)}{q_X(x)}\le M_2, \quad \sup_{x\in \Omega,l}|\phi_l(x)|\le M_2, \quad \text{var}_{x\sim q_X}\Big[(\frac{\pi_X(x)}{q_X(x)})^2\phi_l(x)\phi_{l'}(x)\Big]\le M_2$$
where the single value $M_2$ is introduced only for the convenience of our proof.

Finally, the integral operator $L_q:L^2(q_X)\to L^2(q_X)$ is defined as follows:
\begin{equation} \label{equ:integraloperator}
(L_q g)(x) :=\int_{\Omega} k_+(x,x')g(x')q_X(x')dx', \ x\in\Omega,\ g\in L^2(q_X).    
\end{equation}

This operator can be viewed as a positive self-adjoint operator on $L^2(q_X)$. We can define $L_\pi:L^2(\pi_X)\to L^2(\pi_X)$ in a similar way. Note that the power function of $L_q$, $L_q^r$, is well-defined as a positive self-adjoint operator as $L_q$ is a positive self-adjoint operator. Denote $\text{Range}(L_q^r)$ the range of $L_q^r$ on the domain $L^2(q_X)$. Note that a larger $r\ge \frac{1}{2}$ in $\text{Range}(L_q^r)$ corresponds to a more regular (and smaller) subspace of $L^2(q_X)$: $\text{Range}(L_q^{r_1}) \subset \text{Range}(L_q^{r_2})$ whenever $r_1 \ge r_2$. Conventionally, we write $L_q^{-r}g\in L^2(q_X)$ if (1) $g\in \text{Range}(L_q^r)$, (2) $L_q^{-r}g$ is an element in the preimage set of $g$ under the operator $L_q^{r}$ on the domain $L^2(q_X)$. Further details about the operator $L_q^r$ can be found in Section \ref{sec:RLS}.

\subsection{Convergence of Doubly Robust Stein-Kernelized Estimator} \label{sec:DRSK2}
We are now ready to present the main theorems for our DRSK estimator (Algorithm \ref{DRSKalgorithm}). All the theorems are understood in the following way: If the noise term $Y$ does not exist, then we drop Assumptions \ref{CSA}-\ref{noise1} (since they are automatically true) and set $M_0=0$ in the results; If $\pi=q$, then we drop Assumptions \ref{lrupper2}-\ref{lrupper} (since they are automatically true) with $B=1$.

%Further developments of the main theorem can be found in Section.

% For CF estimators:
\begin{theorem} [DRSK in all cases under weak assumptions] \label{mainDRSK0}
Suppose Assumptions \ref{CSA}, \ref{noise1}, and \ref{lrupper} hold. Take an RLS estimate with $\lambda = m^{-\frac{1}{2}}$ and $B_0=2B$ in (\ref{COP3}). Let $m=\alpha n$ where $0< \alpha< 1$. %The DRSK estimator $\hat{\theta}_{DRSK}$ satisfies the following bound.\\
If $\bar{f} \in \text{Range}(L_q^{r})$ ($\frac{1}{2}\le r \le 1$), then $\mathbb{E}_q [(\hat{\theta}_{DRSK}-\theta)^2] \le C_1(C_{f} n^{-\frac{1}{2}-r}+M_0 n^{-1})$ where $C_{f}=\|L_q^{-r}\bar{f}\|^2_{L^2(q_X)}$ (which is a constant indicating the regularity of $\bar{f}$ in $\mathcal{H}_+$), $C_1$ only depends on $\alpha, \kappa, B$.
%(b) If $\bar{f}\in \text{Range}(L_q^{\frac{1}{2}})$, then for any given $\varepsilon>0$, $\mathbb{E}_q [(\hat{\theta}_{DRSK}-\theta)^2] \le C_1(C_\varepsilon n^{-\frac{3}{2}}+M_0 n^{-1}+\varepsilon n^{-1})$ where $C_\varepsilon=\|L_q^{-1}g\|_{L^2(q_X)}$ and $g\in \text{Range}(L_q)$ is an $\varepsilon$-approximation of $\bar{f}$ in $\mathcal{H}_+$, $C_1$ only depends on $\alpha, \kappa, B$.\\
%In particular, if the noise term $Y$ does not exist, then we set $M_0=0$ in the result.
\end{theorem}

Next we can obtain a better result with a stronger assumption. %Therefore, the DRSK estimator is guaranteed to be better than BBIS estimator in terms of the MSE.

\begin{theorem} [DRSK in all cases under strong assumptions] \label{mainDRSK0-2}
Suppose Assumptions \ref{CSA}, \ref{noise1}, \ref{lrupper}, and \ref{lrupper4} hold. Take an RLS estimate with $\lambda = m^{-\frac{1}{2}}$ and $B_0=4B$ in (\ref{COP3}). Let $m=\alpha n$ where $0< \alpha< 1$. %The DRSK estimator $\hat{\theta}_{DRSK}$ satisfies the following bound.\\
If $\bar{f} \in \text{Range}(L_q^{r})$ ($\frac{1}{2}\le r \le 1$), then $\mathbb{E}_q [(\hat{\theta}_{DRSK}-\theta)^2] \le C_1(C_{f,n} n^{-\frac{1}{2}-r}+M_0 n^{-1})$ where $C_{f,n}=\|L_q^{-r}\bar{f}\|^2_{L^2(q_X)} \cdot o(1)$ as $n\to \infty$, $C_1$ only depends on $\alpha, \kappa, B$.
%(b) If $\bar{f}\in \text{Range}(L_q^{\frac{1}{2}})$, then for any given $\varepsilon>0$, $\mathbb{E}_q [(\hat{\theta}_{DRSK}-\theta)^2] \le C_1(C_{\varepsilon,n} n^{-\frac{3}{2}}+M_0 n^{-1}+\varepsilon n^{-1})$ where $C_{\varepsilon,n}=\|L_q^{-1}g\|_{L^2(q_X)} \cdot o(1)$  as $n\to \infty$ and $g\in \text{Range}(L_q)$ is an $\varepsilon$-approximation of $\bar{f}$ in $\mathcal{H}_+$, $C_1$ only depends on $\alpha, \kappa, B$. \\
%In particular, if the noise term $Y$ does not exist, then we set $M_0=0$ in the result.
\end{theorem}

%Further remarks can be found in Section \ref{sec:further}. 

%Consider the case where $M_0$ is a relatively small number compared with $\|L_q^{-r}\bar{f}\|_{L^2(q_X)}$ which is generally intractable in practice. Then the bound in Theorem \ref{mainDRSK0-2} becomes approximately
%$$\mathbb{E}_q[(\hat{\theta}_{DRSK}-\theta)^2] = o(n^{-\frac{1}{2}-r}).$$ 
%and so 
Theorem \ref{mainDRSK0-2} shows that except that the noise part remains at the canonical rate, DRSK achieves effectively a supercanonical rate in all cases.

We make a remark on the conditions $\bar{f} \in \text{Range}(L_q^{r})$ ($\frac{1}{2}\le r \le 1$) appearing in Theorems \ref{mainDRSK0} and \ref{mainDRSK0-2} (as well as the subsequent theorems). The requirement of $r= \frac{1}{2}$ is essentially equivalent to saying that $\bar{f}$ (potentially with a difference on a set of measure zero with respect to the measure $q_X$) is in the RKHS $\mathcal{H}_+$; See Section \ref{sec:RLS} for technical details. A larger $r\ge \frac{1}{2}$ corresponds to a more restrictive assumption that $\bar{f}$ is in a more regular (and smaller) subspace of $\mathcal{H}_+$, and leads to a better MSE rate in our theorems (which is consistent with intuition). Therefore, as long as $\bar{f} \in \mathcal{H}_+$, we can assert $r\ge \frac{1}{2}$ and apply Theorems \ref{mainDRSK0} and \ref{mainDRSK0-2}. 

In addition, we pinpoint that $\bar{f} \in \mathcal{H}_+$ is a common and necessary assumption in kernel-based CF and IS \citep{CF1,CF2,LL}. The KRR theory \citep{SW} indicates that the approximation and estimation error of KRR can be as large as $O(1)$ if $\bar{f}$ is merely in a large space like $L^2(q_X)$. Therefore, we can foresee that the supercanonical rate can be achieved only when the ground-truth regression function is in a small regular space like the Stein-Kernelized RKHS. It is generally not easy to check $\bar{f} \in \text{Range}(L_q^{r})$ or $\bar{f} \in \mathcal{H}_+$ in practice as the ground-truth regression function $\bar{f}$ is typically unknown. Nevertheless, our algorithms can still be applied and the experimental results in Sections \ref{sec:illustration} and \ref{sec:multipledistributions} show that our performance could still be superior despite the challenge in assumption verification.
% those requirements are needed only for theoretical guarantees.

Another version of Theorem \ref{mainDRSK0} (and similarly, Theorem \ref{mainDRSK0-2}) to replace the requirement $\bar{f} \in \text{Range}(L_q^{r})$ is to assume that there exists a $\varepsilon>0$ and $g \in \text{Range}(L_q)$, such that $\|\bar{f}-g\|^2_{\mathcal{H}_+}\le \varepsilon$. This assumption holds, for instance, if $\bar{f} \in \text{Range}(L_q^{\frac{1}{2}})$ (Proposition \ref{ease2} in Section \ref{sec:RLS}).
Then under this assumption, $\mathbb{E}_q [(\hat{\theta}_{DRSK}-\theta)^2] \le C_1(\|L_q^{-1}g\|^2_{L^2(q_X)} n^{-\frac{3}{2}}+M_0 n^{-1}+\varepsilon n^{-1})$ where $C_1$ only depends on $\alpha, \kappa, B$. 

A proof outline of Theorems \ref{mainDRSK0} and \ref{mainDRSK0-2} is in Section \ref{sec:outline}. Detailed proofs are given in Section \ref{sec:DRSK}.

%As a contrast, the CF and BBIS estimators have the following less favorable bounds in the ``Both" case.

\subsection{Convergence of Control Functional}
We present our main results for CF (Algorithm \ref{CFalgorithm}) including the results for SimCF (Algorithm \ref{SimCFalgorithm}), which provide comparisons to our results for DRSK.

\begin{theorem}[CF in the ``Standard" case]\label{mainCF0-2}
Take an RLS estimate with $\lambda = m^{-\frac{1}{2}}$. Let $m=\alpha n$ where $0< \alpha< 1$. %The CF estimator $\hat{\theta}_{CF}$ is an unbiased estimator of $\theta$ that satisfies the following bounds:\\
If $f \in \text{Range}(L_\pi^{r})$ ($0\le r \le 1$), then $\mathbb{E}_\pi [(\hat{\theta}_{CF}-\theta)^2] \le C_1C_{f} n^{-1-r}$ where $C_{f}=\|L_\pi^{-r}f\|^2_{L^2(\pi_X)}$ (which is a constant indicating the regularity of $f$ in $\mathcal{H}_+$), and $C_1$ only depends on $\alpha, \kappa$. In particular, $\mathbb{E}_\pi [(\mu_X(s_m)-\theta)^2] \le C_1C_{f} m^{-r}$.
%(b) If $f \in \overline{\mathcal{H}_+}^\pi$, then for any given $\varepsilon>0$, $\mathbb{E}_\pi [(\hat{\theta}_{CF}-\theta)^2] \le C_1(C_\varepsilon n^{-2}+\varepsilon n^{-1})$ where $C_\varepsilon=\|L_\pi^{-1}g\|_{L^2(\pi_X)}$ and $g\in \text{Range}(L_\pi)$ is an $\varepsilon$-approximation of $f$ in $L^2(\pi_X)$, and $C_1$ only depends on $\alpha, \kappa$. 
\end{theorem}

\begin{theorem}[CF in the ``Partial" case]\label{mainCF0-1}
Suppose Assumption \ref{noise1} holds and take an RLS estimate with $\lambda = m^{-\frac{1}{2}}$. Let $m=\alpha n$ where $0< \alpha< 1$. %The CF estimator $\hat{\theta}_{CF}$ is an unbiased estimator of $\theta$ that satisfies the following bounds:\\
If $\bar{f} \in \text{Range}(L_\pi^{r})$ ($0\le r \le 1$), then 
$\mathbb{E}_\pi [(\hat{\theta}_{CF}-\theta)^2] \le C_1(C_{f} n^{-1-r}+M_0 n^{-1})$
where $C_{f}=\|L_\pi^{-r}\bar{f}\|^2_{L^2(\pi_X)}$ (which is a constant indicating the regularity of $\bar{f}$ in $\mathcal{H}_+$), and $C_1$ only depends on $\alpha, \kappa$. In particular, $\mathbb{E}_\pi [(\mu_X(s_m)-\theta)^2] \le C_1(C_{f} m^{-r}+M_0)$.
%(b) If $\bar{f} \in \overline{\mathcal{H}_+}^\pi$, then for any given $\varepsilon>0$, $\mathbb{E}_\pi [(\hat{\theta}_{CF}-\theta)^2] \le C_1(C_\varepsilon n^{-2}+M_0 n^{-1}+\varepsilon n^{-1})$ where $C_\varepsilon=\|L_\pi^{-1}g\|_{L^2(\pi_X)}$ and $g\in \text{Range}(L_\pi)$ is an $\varepsilon$-approximation of $\bar{f}$ in $L^2(\pi_X)$, and $C_1$ only depends on $\alpha, \kappa$. 

%	Then the estimator $\hat{\mu}$ is an unbiased estimator of $\mu$ with
%$$\mathbb{E}_\pi [(\hat{\theta}_{CF}-\theta)^2] = O(C_\kappa(C_fn^{-2}+ M_0n^{-1}))$$
%where $C_f$ is a constant free of $m$ (and $n$), $C_\kappa=2\kappa^4+2$ and the outside $O$ only depends on the ratio $m/n$.
\end{theorem}

%Suppose $M_0$ is a relatively small number compared with $C$ (which contains information about $\bar{f}$ and is possibly large in practice). Then the bound in Theorem \ref{mainCF0-1} becomes approximately
%$$\mathbb{E}_\pi[(\hat{\theta}_{CF}-\theta)^2] =O(n^{-1-r}).$$ 
%and thus CF achieves an effective supercanonical rate.
% especially when a limited number of samples are available to us (which is usually true in practice).
Theorem \ref{mainCF0-1} shows that in the ``Partial" case, even if the noise $Y$ is fully unknown, CF applied on only $X$ still improves the Monte Carlo rate except that the noise part remains at the canonical rate. %Theorem \ref{mainCF0} part (b) provides a reasonable explanation that the CF estimator still performs well even if a weak assumption is imposed (NOT SURE ABOUT THIS; SEE THE SIMILAR COMMENT BELOW..).
The same choice of $\lambda$ is also suggested by Theorem 2 in \citet{CF1}. \footnotemark\footnotetext{In general, $\lambda>0$ is required to prevent overfitting and stabilize the inverse numerically by bounding the smallest eigenvalues away from zero \citep{hastie2009elements,welling2013kernel}.}  Our refined study in Section \ref{sec:RLS} contributes to obtaining a better rate in Theorem \ref{mainCF0-1} compared with \citet{CF1}. 
%in Theorem \ref{mainCF0-2}

\begin{theorem}[CF in the ``Biased" case]\label{mainCF0-3}
Suppose Assumption \ref{lrupper2} holds and take an RLS estimate with $\lambda = m^{-\frac{1}{2}}$. %Then the CF estimator $\hat{\theta}_{CF}$ satisfies the following bounds:\\
If $f \in \text{Range}(L_q^{r})$ ($0\le r \le 1$), then $\mathbb{E}_q [(\hat{\theta}_{CF}-\theta)^2] \le C_1C_{f} m^{-r}$ where $C_{f}=\|L_q^{-r}f\|^2_{L^2(q_X)}$ (which is a constant indicating the regularity of $f$ in $\mathcal{H}_+$), and $C_1$ only depends on $\kappa$, $\mathbb{E}_{x\sim \pi_X}[\frac{\pi_X(x)}{q_X(x)}]$. In particular, $\mathbb{E}_q [(\mu_X(s_m)-\theta)^2] \le C_1C_{f} m^{-r}$.
%(b) If $f \in \overline{\mathcal{H}_+}^q$, then for any given $\varepsilon>0$, $\mathbb{E}_q [(\hat{\theta}_{CF}-\theta)^2] \le C_1(C_\varepsilon m^{-1}+\varepsilon)$ where $C_\varepsilon=\|L_q^{-1}g\|_{L^2(q_X)}$ and $g\in \text{Range}(L_q)$ is an $\varepsilon$-approximation of $f$ in $L^2(q_X)$, and $C_1$ only depends on $\kappa$ and $\mathbb{E}_\pi[\pi_X / q_X]$. 
\end{theorem}

Theorem \ref{mainCF0-3} implies that the CF estimator retains consistency regardless of the generating distribution of $X$, as long as this distribution is not too ``far from" the target distribution in the sense of a controllable likelihood ratio. This shows that CF can partially reduce the bias in addition to variance reduction, yet it may be less favorable since a supercanonical rate is not guaranteed theoretically.
Note that the above bound has nothing to do with
$n-m$ since in this case, a single $f_m(X,Y)$ is not necessarily an unbiased estimator of $\theta$ and hence taking the average of $f_m(x_j,y_j)$ may not improve the rate. Therefore it is reasonable to take $m=n$ to minimize the upper bound and use the simplified CF estimator (Algorithm \ref{SimCFalgorithm}) in this case.
%Theorem \ref{mainCF0-3} also implies the effectiveness of the simplified CF estimator (Algorithm \ref{SimCFalgorithm}) in the “Biased” case.

%However, the next theorem indicates that within a given tolerance, the estimator can make the canonical rate even with a weak assumption.

%It is not hard to see that the CF estimator has the similar less favorable property in the ``Both" case.

%\begin{theorem}\label{mainCF0}
%Suppose Assumptions \ref{CSA}, \ref{noise1}, \ref{lrupper2} hold. Take an RLS estimate with $\lambda = m^{-\frac{1}{2}}$. Then the CF estimator $\hat{\theta}_{CF}$ satisfies the following bound.\\
%(a) If $\bar{f} \in \text{Range}(L_q^{r})$ ($0\le r \le 1$), then $\mathbb{E}_q [(\hat{\theta}_{CF}-\theta)^2] \le C_1(C m^{-r}+M_0)$ where $C=\|L_q^{-r}\bar{f}\|_{L^2(q_X)}$ (which is a constant indicating the regularity of $\bar{f}$ in $\mathcal{H}_+$), $C_1$ only depends on $\kappa$ and $\mathbb{E}_\pi[\pi_X / q_X]$.\\
%(b) If $\bar{f} \in \overline{\mathcal{H}_+}^q$ (the closure of $\mathcal{H}_+$ in the space $L^2(q_X)$), then for any given $\varepsilon>0$, $\mathbb{E}_q [(\hat{\theta}_{CF}-\theta)^2] \le C_1(C_\varepsilon m^{-1}+\varepsilon+M_0)$ where $C_\varepsilon=\|L_q^{-1}g\|_{L^2(q_X)}$ and $g\in \text{Range}(L_q)$ is an $\varepsilon$-approximation of $\bar{f}$ in $L^2(q_X)$, $C_1$ only depends on $\kappa$ and $\mathbb{E}_\pi[\pi_X / q_X]$. 
%\end{theorem}

\begin{theorem}[CF in the ``Both" case]\label{mainCF0}
Suppose Assumptions \ref{CSA}, \ref{noise1}, and \ref{lrupper2} hold. Take an RLS estimate with $\lambda = m^{-\frac{1}{2}}$. Let $m=\alpha n$ where $0< \alpha< 1$. %Then the CF estimator $\hat{\theta}_{CF}$ satisfies the following bounds:\\
If $\bar{f} \in \text{Range}(L_q^{r})$ ($0\le r \le 1$), then $\mathbb{E}_q [(\hat{\theta}_{CF}-\theta)^2] \le C_1(C_{f} n^{-r}+M_0)$ where $C_{f}=\|L_q^{-r}\bar{f}\|^2_{L^2(q_X)}$ (which is a constant indicating the regularity of $\bar{f}$ in $\mathcal{H}_+$), and $C_1$ only depends on $\alpha$, $\kappa$, $\mathbb{E}_{x\sim \pi_X}[\frac{\pi_X(x)}{q_X(x)}]$. In particular, $\mathbb{E}_q [(\mu_X(s_m)-\theta)^2] \le C_1(C_{f} m^{-r}+M_0)$.\\
%(b) If $\bar{f} \in \overline{\mathcal{H}_+}^q$ (the closure of $\mathcal{H}_+$ in the space $L^2(q_X)$), then for any given $\varepsilon>0$, $\mathbb{E}_q [(\hat{\theta}_{CF}-\theta)^2] \le C_1(C_\varepsilon m^{-1}+\varepsilon+M_0)$ where $C_\varepsilon=\|L_q^{-1}g\|_{L^2(q_X)}$ and $g\in \text{Range}(L_q)$ is an $\varepsilon$-approximation of $\bar{f}$ in $L^2(q_X)$, and $C_1$ only depends on $\kappa$ and $\mathbb{E}_\pi[\pi_X / q_X]$.  
\end{theorem}

% (NOT SURE IF IT EXPLAINS CLEARLY WHY OUR RESULT IS DIFFERENT FROM OATES..)

Another version of Theorem \ref{mainCF0} to replace the requirement $\bar{f} \in \text{Range}(L_q^{r})$ is to assume that there exists a $\varepsilon>0$ and $g \in \text{Range}(L_q)$, such that $\|\bar{f}-g\|^2_{L^2(q_X)}\le \varepsilon$. This assumption holds, for instance, if $\bar{f} \in \overline{\mathcal{H}_+}^q$ (the closure of $\mathcal{H}_+$ in the space $L^2(q_X)$) (Proposition \ref{ease} in Section \ref{sec:RLS}).
Then under this assumption, $\mathbb{E}_q [(\hat{\theta}_{CF}-\theta)^2] \le C_1(\|L_q^{-1}g\|^2_{L^2(q_X)} n^{-1} +M_0 +\varepsilon)$ where $C_1$ only depends on $\kappa$, $\mathbb{E}_{x\sim \pi_X}[\frac{\pi_X(x)}{q_X(x)}]$. 

Detailed proofs of the above theorems can be found in Section \ref{sec:CF}.

% For BBIS estimators:
%\begin{theorem}\label{mainIS0}
%Suppose Assumption \ref{CSA}, \ref{noise1}, \ref{lrupper} hold and $\bar{f}\in \mathcal{H}_+$. The BBIS estimator $\hat{\theta}_{IS}$ satisfies the following bound.\\
%(a) Take $B_0=2B$ in (\ref{COP3}). Then $\mathbb{E}_q[(\hat{\theta}_{IS}-\theta)^2] =O(n^{-1})$.\\
%(b) Assume Assumption \ref{lrupper4} holds in addition. Take $B_0=4B$ in (\ref{COP3}).  Then $\mathbb{E}_q[(\hat{\theta}_{IS}-\theta)^2] \le o(n^{-1})+2M_0 B_0^2n^{-1}$.
%\end{theorem}

\subsection{Convergence of Black-Box Importance Sampling} \label{sec:BBIS1}
We present the main theorems for BBIS (Algorithm \ref{ISalgorithm})  as follows. The first theorem in the ``Standard" and ``Biased" cases (where the noise $Y$ does not exist) is proved by \citet{LL}.
\begin{theorem}[BBIS in ``Standard" \& ``Biased" cases] \label{mainIS0-1}
Suppose $f\in \mathcal{H}_+$. BBIS $\hat{\theta}_{IS}$ satisfies the following bounds (with $B_0=+\infty$):\\
(a) Suppose Assumption \ref{lrupper3} holds. Then $\mathbb{E}_q[(\hat{\theta}_{IS}-\theta)^2] =O(n^{-1})$.\\
(b) Suppose Assumptions \ref{lrupper} and \ref{lrupper4} hold. Then $\mathbb{E}_q[(\hat{\theta}_{IS}-\theta)^2] =o(n^{-1})$.
\end{theorem}

In the ``Partial" case and ``Both" case, we have an extra noise term $Y$. Note that the weights constructed in the BBIS only depends on the $X$ factor, free of $Y$. Therefore the noise cannot be controlled by the weights. To address this issue, we impose an upper bound on each weight in (\ref{COP0}) to ensure that the noise will not blow up.

\begin{theorem}[BBIS in ``Partial" \& ``Both" cases]\label{mainIS0}
Suppose Assumptions \ref{CSA}, \ref{noise1}, and \ref{lrupper} hold and $\bar{f}\in \mathcal{H}_+$. BBIS $\hat{\theta}_{IS}$ satisfies the following bounds:\\
(a) Take $B_0=2B$ in (\ref{COP0}). Then $\mathbb{E}_q[(\hat{\theta}_{IS}-\theta)^2] =O(n^{-1})$.\\
(b) Suppose Assumption \ref{lrupper4} holds in addition. Take $B_0=4B$ in (\ref{COP0}).  Then $\mathbb{E}_q[(\hat{\theta}_{IS}-\theta)^2] \le o(n^{-1})+2M_0 B_0^2n^{-1}$.
\end{theorem}

Note that $\bar{f}\in \mathcal{H}_+$ is essentially the same as $\bar{f}\in \text{Range}(L_q^{\frac{1}{2}})$ in our setting (see Section \ref{sec:RLS} for details). Assuming $\bar{f}$ in a more ``regular" space such as $ \text{Range}(L_q^{r})$ ($r>\frac{1}{2}$) does not improve the above rate since the construction of BBIS weights is independent of this function. Consider a case where $M_0$ is a relatively small number. Then the bound in Theorem \ref{mainIS0-1} part (b) essentially gives us a supercanonical rate
$$\mathbb{E}_q[(\hat{\theta}_{IS}-\theta)^2] =o(n^{-1}).$$

There is a difference regarding the construction of weights between our Theorem \ref{mainIS0} and the original BBIS work, on the imposition of the upper bound $\frac{B_0}{n}$ on the weights in the quadratic program. This is to guarantee that the error term $\epsilon$ is controlled to induce at least the canonical rate, otherwise the error may blow up. %If the error term $Y$ does not show up, then we do not require such a bound $B_0$. 
This modification leads us to redevelop results for BBIS. The detailed proof of Theorem \ref{mainIS0} can be found in Section \ref{sec:IS}.

\subsection{Proof Outline} \label{sec:outline}
We close this section by briefly outlining our proofs for the three estimators in the ``Both" case. 
Write $\bar{f}_m(x_j):=\bar{f}(x_j)-s_m(x_j)+\mu_X(s_m)$. To see Theorems \ref{mainDRSK0} and \ref{mainDRSK0-2} (and obtain Theorem \ref{mainIS0} along the way), we first express $(\hat{\theta}_{DRSK}-\theta)^2$ as

\begin{align*}
(\hat{\theta}_{DRSK}-\theta)^2&=\left(\sum_{j=m+1}^{n} \hat{w}_j (\bar{f}_m(x_j)+\epsilon(x_j,y_j)-\theta)\right)^2\\
&\le 2\left(\left(\sum_{j=m+1}^{n} \hat{w}_j (\bar{f}_m(x_j)-\theta)\right)^2+\left(\sum_{j=m+1}^{n} \hat{w}_j \epsilon(x_j,y_j)\right)^2\right)\\
&\le 2\left(\|\bar{f}_m-\theta\|^2_{\mathcal{H}_0} \cdot \mathbb{S}(\{\hat{w}_j,x_j\},\pi_X)+\left(\sum_{j=m+1}^{n} \hat{w}_j \epsilon(x_j,y_j)\right)^2\right)
\end{align*}
where we have used the Cauchy-Schwarz inequality in both inequalities and additionally the reproducing kernel property in the last inequality. By the construction of the RKHS, we can readily see that
$$\|\bar{f}_m-\theta\|^2_{\mathcal{H}_0}\le \|\bar{f}-s_m\|^2_{\mathcal{H}_+}.$$
By the construction of $\hat{w}$ and Assumption \ref{noise1}, we can prove that
\begin{align}
\mathbb{E}_q[(\hat{\theta}_{DRSK}-\theta)^2]&\le 2\left(\mathbb{E}_q[\|\bar{f}_m-\theta\|^2_{\mathcal{H}_0}] \cdot \mathbb{E}_q[\mathbb{S}(\{\hat{w}_j,x_j\},\pi_X)]+M_0B_0^2(n-m)^{-1}\right).\nonumber\\
&\le 2\left(\mathbb{E}_q[\|\bar{f}-s_m\|^2_{\mathcal{H}_+}] \cdot \mathbb{E}_q[\mathbb{S}(\{\hat{w}_j,x_j\},\pi_X)]+M_0B_0^2(n-m)^{-1}\right).\label{equ:outline1}
\end{align} 
Therefore, the main task is to analyze two terms:
$$\mathbb{E}_q[\|\bar{f}-s_m\|^2_{\mathcal{H}_+}] \text{\ \ and\ \ } \mathbb{E}_q[\mathbb{S}(\{\hat{w}_j,x_j\},\pi_X)].$$
Note that the first term measures the learning error between the true regression function and the KRR function under the $\mathcal{H}_+$ norm. This term can be analyzed using the theory of KRR in Section \ref{sec:RLS}. The theory indicates that
\begin{equation} \label{equ:outline2}
\mathbb{E}_q[\|\bar{f}-s_m\|^2_{\mathcal{H}_+}]=O(m^{-r+\frac{1}{2}}).    
\end{equation}
The second term is about the performance guarantee of black-box importance weights, and could be analyzed by applying similar techniques from \citet{LL}. However, since we have modified the original BBIS algorithm by adding an additional upper bound on each weight, we need to re-establish new results. This term will be analyzed in Section \ref{sec:IS}. Note that analyzing this term provides us with the result for BBIS, Theorem \ref{mainIS0}, at the same time. We will show that
\begin{equation} \label{equ:outline3}
\mathbb{E}_q[\mathbb{S}(\{\hat{w}_j,x_j\},\pi_X)] = O((n-m)^{-1}) \ (\text{Ass. } \ref{lrupper2})  \quad \text{or}  \quad o((n-m)^{-1}) \ (\text{Ass. } \ref{lrupper},\ref{lrupper4}).
\end{equation}
Plugging \eqref{equ:outline2} and \eqref{equ:outline3} into \eqref{equ:outline1}, we obtain Theorems \ref{mainDRSK0} and \ref{mainDRSK0-2}.

Next, to see Theorem \ref{mainCF0}, we look at one item $f_m(x_j,y_j)-\theta$ in the summation and separate the error term $\epsilon$:
$$\mathbb{E}_q[\nu((f_m - \theta)^2)] \le 3\left(\mathbb{E}_q[\nu_X((\bar{f} - s_m)^2)]+\mathbb{E}_q [(\mu_X(s_m)-\theta)^2])+\mathbb{E}_q[\epsilon^2]\right)$$
For the second term, by applying Cauchy–Schwarz inequality, we have that 
$$\mathbb{E}_q [(\mu_X(s_m)-\theta)^2]\le  \mathbb{E}_q[\nu_X((f - s_m)^2)] \mathbb{E}_{q}\left[\left(\frac{\pi_X}{q_X}\right)^2\right].$$ 
Hence, in order to analyze $f_m(x_j,y_j)-\theta$, we only need to estimate
$$
\mathbb{E}_q[\nu_X((\bar{f} - s_m)^2)].$$
This measures the learning error between the true regression function and the KRR function under the $L^2$ norm, which can be analyzed using the theory in Section \ref{sec:RLS}. The theory indicates that
\begin{equation} \label{equ:outline4}
\mathbb{E}_q[\nu_X((\bar{f} - s_m)^2)] = O(m^{-r}).
\end{equation}
showing that a single $f_m(X,Y)$ is an asymptotically unbiased estimator of $\theta$ (at the rate of $O(m^{-r})$ when $m\to \infty$). Nevertheless, %given $D_0$ (and thus given $s_m$), 
$f_m(X,Y)$ with finite $m$ is in general a biased estimator of $\theta$ under the biased generating distribution $q$ so it is not necessary that taking the average in the final step of
$\hat{\theta}_{CF}$ can enhance a single $f_m(X,Y)$. In this case, Theorem \ref{mainCF0} follows from \eqref{equ:outline4}.

%This explains why splitting data in CF does not work in the ``Biased" case.

\section{Numerical Experiments}\label{sec:numerics}
We conduct extensive numerical experiments to demonstrate the effectiveness of our method. %We compare three methods, %We reveal results, especially when $n$ is small, about the (empirical) MSE of three estimators: 
In addition to the CF, modified BBIS (which includes the original BBIS by setting $B_0=+\infty$, and which we refer to simply as BBIS in this section), and DRSK estimators described in Algorithms \ref{CFalgorithm}, \ref{ISalgorithm}, \ref{DRSKalgorithm} respectively, we consider two more estimators:
\begin{enumerate}
\item The DRSK-Reuse (DRSK-R) estimator: This estimator is similar to the DRSK estimator except that it reuses the entire dataset $D$ (not only $D_0$ or $D_1$ in DRSK) to construct the CF-adjusted sample $f_n(x,y)=f(x,y)-s_n(x)+\mu_X(s_n)$ and the importance weights $\hat{w}_j$ ($j=1,\cdots,n$), so that the final DRSK-R estimator is given by
$$\hat{\theta}_{DRSK-R}:=\sum_{j=1}^{n} \hat{w}_j f_n(x_j,y_j).$$
The reuse of $D$ will cause some dependency between the CF-adjusted sample and the weights. \footnotemark\footnotetext{For DRSK-R, although it is hard to theoretically analyze the mean squared error $\mathbb{E}_q[(\hat{\theta}_{DRSK-R}-\theta)^2]$ due to the dependency, it is indeed feasible to analyze the mean absolute error $\mathbb{E}_q[|\hat{\theta}_{DRSK-R}-\theta|]$. In fact, similarly as in \eqref{equ:outline1}, we can show that $$
(\mathbb{E}_q[|\hat{\theta}_{DRSK-R}-\theta|])^2
\le 2\left(\mathbb{E}_q[\|\bar{f}_n-\theta\|^2_{\mathcal{H}_0}] \cdot \mathbb{E}_q[\mathbb{S}(\{\hat{w}_j,x_j\},\pi_X)]+M_0B_0^2n^{-1}\right)$$
where we use the fact that $(\mathbb{E}_q[\|\bar{f}_n-\theta\|_{\mathcal{H}_0}\cdot \sqrt{\mathbb{S}(\{\hat{w}_j,x_j\},\pi_X)}])^2 \le \mathbb{E}_q[\|\bar{f}_n-\theta\|^2_{\mathcal{H}_0}] \cdot \mathbb{E}_q[\mathbb{S}(\{\hat{w}_j,x_j\},\pi_X)]$ by Cauchy-Schwarz inequality regardless of the dependency in $\bar{f}_n$ and $\mathbb{S}(\{\hat{w}_j,x_j\},\pi_X)$. Therefore, the upper bounds in Theorems \ref{mainDRSK0} and \ref{mainDRSK0-2} apply to $(\mathbb{E}_q[|\hat{\theta}_{DRSK-R}-\theta|])^2$.} Nevertheless, we will observe in experiments the high effectiveness of the DRSK-R estimator in terms of reducing MSE. 

\item The Simplified CF (SimCF) estimator, described in Algorithm \ref{SimCFalgorithm}.
\end{enumerate}

%We observe that DRSK works in particular well for multivariate distributions.

%In practice, the height $h_2$ is heuristically chosen to be the largest $f(x_i,y_i)$ of data $D_1$, and the bandwidth $h_1$ is heuristically chosen to be larger than the sum of the pairwise square distance of data $D_1$. the expression of $\hat{\theta}^2$ is very complicated (and even unavailable in BBIS and DRSK)  and, in some of our examples, the need to estimate $\theta$ as well using a large number of simulation runs

Note that the ground-truth population MSE, i.e.,
$$\text{MSE}:=\mathbb{E}[(\hat{\theta}-\theta)^2]$$
where $\hat{\theta}$ is the estimator, cannot be computed in closed form due to the sophisticated expression of $\hat\theta$. Therefore, we use the following alternative to compute the MSE. For each data distribution and each size $n$, we simulate the whole procedure 50 times: At each repetition $j=1,\cdots,50$, we generate a new dataset of size $n$ drawn from the data distribution, and derive estimators $\hat{\theta}_{i,j}$ of the target parameter $\theta$ based on this dataset where $i=1,\cdots,5$ indicates the five considered approaches. Thus $(\hat{\theta}_{i,j}-\theta)^2$ represents the squared error in the $j$-th repetition. Then we regard the average of all squared errors (empirical MSE) as the proxy for the population MSE, i.e.,
$$\hat{\text{MSE}}_i:=\frac{1}{50} \sum_{j=1}^{50} (\hat{\theta}_{i,j}-\theta)^2, \qquad i=1,\cdots,5.$$
%where $\hat{\theta}_j$ is the value of the estimator of the $j$-th simulation. 

\textbf{Kernel Selection.} Throughout this section, the reproducing kernel of the ``primary" RKHS (i.e., the $\mathcal{H}$ in Section \ref{sec:Stein}) is chosen to be the widely used radial basis function (RBF) kernel (also known as the Gaussian kernel): 
$$k(x,x') = \exp\left(-\frac{1}{h_1}\|x-x'\|^2_2\right).$$
This kernel satisfies the conditions in Section \ref{sec:Stein} for any continuously differentiable densities supported on $\mathbb{R}^d$ whose score function $\bm{u}(x)$ is of polynomial growth rate. Therefore, it is an ideal kernel to be used.

\textbf{Hyperparameter Selection.} As suggested by our theorems, we select the following hyperparameters:
\begin{enumerate}
    \item $m=0.5n$. Theorems \ref{mainDRSK0}-\ref{mainCF0} suggest that $m$ should be taken as $\alpha n$ where $0<\alpha<1$. $\alpha=0.5$ is a simple middle-ground choice and has also been used in \citet{CF1}.
    \item $\lambda=0.01m^{-\frac{1}{2}}$. Theorems \ref{mainDRSK0}-\ref{mainCF0} suggest that $\alpha$ should be taken as $\Theta(m^{-\frac{1}{2}})$.  We choose a small multiplier $0.01$ since a small regularization term under the premise of stabilizing the inverse numerically is preferrable in practice \citep{CF1,CF2}. Other larger or smaller choices of $\lambda$ such as $\lambda=m^{-\frac{1}{2}}$ or $\lambda=0$ may have less satisfactory performance. We will illustrate this in Section \ref{sec:illustration}.
    \item $B_0=50$. $B_0$ is explicitly provided by Theorems \ref{mainDRSK0},\ref{mainDRSK0-2},\ref{mainIS0}. A large value of $B_0$ obviously satisfies the conditions therein. In our experiments, we observe that the performance of all estimators is robust to the choice of $B_0$ (including the ``$+\infty$" in the original BBIS) as long as $B_0$ is relatively large. We will illustrate this in Section \ref{sec:illustration}. %We show in Section \ref{sec:illustration} that different choices of $B_0$ (including the one in the original BBIS) produce very similar results. %So we only report the results of our version of BBIS.
    \item The bandwidth $h_1$ in the kernel is typically chosen to be the median of the pairwise square distance of the input data, as suggested by \citet{LL,gretton2012kernel}. We follow this approach in our experiments. %This may not be a good choice when the biased generated distribution and the target distribution differ a lot, as we will show later. %In our experiment, we observe that a large value of $h_1$ produce similar good performance while a small value of $h_1$ performs poorly. Therefore we choose a relatively large $h_1$ for conservativeness.

%Another potential way is to select $h_1$ via validation by minimizing the variance from multiple constructions of CF based on different subset division. This approach is slightly different from the ``MSE validation" in \citet{CF1} because the biased generated distribution does not allow accurate estimation of $\theta$ on the validation data. However, this approach may fail if the range of candidate $h_1$ is too wide. Therefore, we use a combined approach: $h_1=a_1*PSD$ where $a_1$ is around $1$, such as $[0.1, 10]$, and tuned by validation.
\end{enumerate} 
We emphasize that the same hyperparameters are used in all approaches for a fair comparison.

Experimental results are displayed using plots of log MSE against the sample size $n$. Log MSE is used as it allows easy observation on the polynomial decay. In the following, we conduct experiments on a wide range of problem settings.

\subsection{Basic Illustration} \label{sec:illustration}
In this section, we consider a synthetic
problem setting borrowed from \citet{CF1}. Our goal is to estimate the expectation of $f(X,Y) = \sin( \frac{\pi}{d} \sum_{i=1}^d X_i)+Y$ under the target distribution $\pi$ where $\pi_X=\mathcal{N}(0,I_d)$ is a $d$-dimensional standard Gaussian distribution, and $\pi_{Y}$ is a zero-mean distribution. By symmetry, the ground-true expectation is $\mathbb{E}_\pi[f(X,Y)] = 0$.
We consider the dimension $d = 4$.

\textbf{Illustration of Different Scenarios.}
We consider 9 different scenarios as introduced in Section \ref{sec:discussion}, using 3 noise settings and 3 biased distribution settings as described below:
\begin{enumerate}
    \item Noise settings: (1) $\pi_{Y|X}=0$ (no noise), (2) $\pi_{Y|X}=\mathcal{N}(0,0.1^2)+\sum_{i=1}^d X_i$, %(small noise)\footnotemark\footnotetext{Rigorously, we mean $\mathcal{N}(0,0.1^2)$ is a relatively small noise compared with $\sum_{i=1}^d X_i+\sin( \frac{\pi}{d} \sum_{i=1}^d X_i)$.},
    (3) $\pi_{Y|X}=\mathcal{N}(0,0.1^2)$ .
    \item Biased distribution settings: (A) $q_X=\pi_X$ (no bias), (B) $q_X=  \mathcal{N}(0.5,1)$, (C) $q_X=\mathcal{N}(1,1)$.
\end{enumerate}

\begin{figure*}
 %\label{exp3p}
\begin{center}
\makebox[\textwidth][c]{\includegraphics[width=1.5\textwidth]{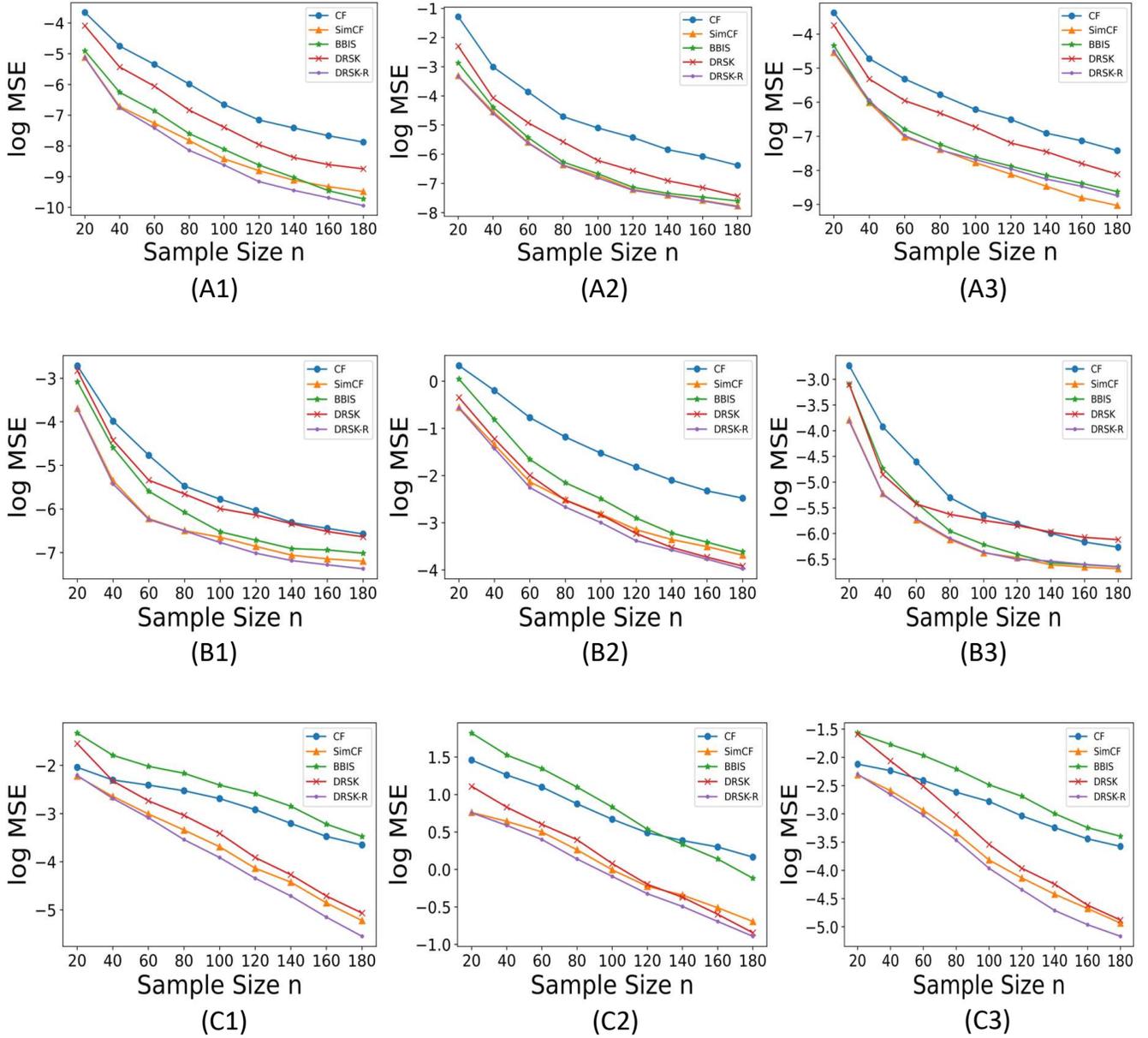}}
\caption{Illustration of different scenarios: The title of the Plot (A1) means the combination of noise setting (1) and biased distribution setting (A). The rest is similar. (1) $\pi_{Y|X}=0$, (2) $\pi_{Y|X}=\mathcal{N}(0,0.1^2)+\sum_{i=1}^d X_i$, (3) $\pi_{Y|X}=\mathcal{N}(0,0.1^2)$.
(A) $q_X=\pi_X$, (B) $q_X=  \mathcal{N}(0.5,1)$, (C) $q_X=\mathcal{N}(1,1)$. }\label{illustration1}
\end{center}
\end{figure*}

The results are shown in Figure \ref{illustration1}. We observe that

\begin{enumerate}
\item DRSK-R and SimCF are the top two approaches in most scenarios. Empirically, DRSK-R appears to be a better alternative to DRSK, and SimCF appears to be a better alternative to CF. 
%\item BBIS is effective in Senerios.
\item DRSK-R is the best when the sampling distribution is biased (e.g., biased distribution setting (C)) and the noise is small (e.g., noise setting (1)); See Plots (A1), (B1), (B2), (C1), (C2), (C3). In these scenarios, it can outperform SimCF (the second-best approach) by up to 25 percent.
\item The superior performance of DRSK-R decreases when the sampling distribution is exact or the noise becomes larger; See Plots (A2), (A3), (B3). In (A2) and (B3), DRSK-R and SimCF perform similarly. (A3) is the only scenario where DRSK-R is less effective than SimCF.
\end{enumerate}

These observations coincide with our theory. DRSK (DRSK-R) performs well, especially when the noise is small and the sampling distribution is not exact, consistent with our Table \ref{Pro}.
As discussed in Section \ref{sec:discussion}, CF (SimCF) is resistant to noise since it takes advantage of the functional approximation of $f$ while the importance weight in BBIS ignores $f$. Thus, increasing noise typically hurts the performance of BBIS, but DRSK (DRSK-R) and CF (SimCF) can maintain some similarly good performance. On the other hand, CF is not resistant to bias because the uniform weight in the final step of constructing CF cannot reduce the bias effectively. Thus SimCF by omitting the final step is a better alternative to CF. Note that this observation holds in almost all scenarios regardless of the bias (including the subsequent experiments). In fact, \citet{CF1} indicates a similar empirical result: In the ``standard" case, although the CF-adjusted sample $f_m(x_j)$ is constructed for the unbiasedness, the actual bias in $s_m(x_j)$ can be negligible for practical purposes, and thus they recommended SimCF for use in applications. In the ``Biased" case, even the unbiasedness of $f_m(x_i)$ is no longer valid so SimCF stays preferable.

\textbf{Illustration of Different Hyperparameters.} We study the effect of different choices of hyperparameters $h_1$, $\lambda$ and $B_0$ in the scenario (C3): $\pi_{Y|X}=\mathcal{N}(0,0.1^2)$ and $q_X=\mathcal{N}(1,1)$. Let $PSD$ denote the median of the pairwise square distance of the input data. Let $\lambda=a_1\times m^{-\frac{1}{2}}$, $h_1=a_2\times PSD$.

\begin{figure*}
 %\label{exp3p}
\begin{center}
\makebox[\textwidth][c]{\includegraphics[width=1.15\textwidth]{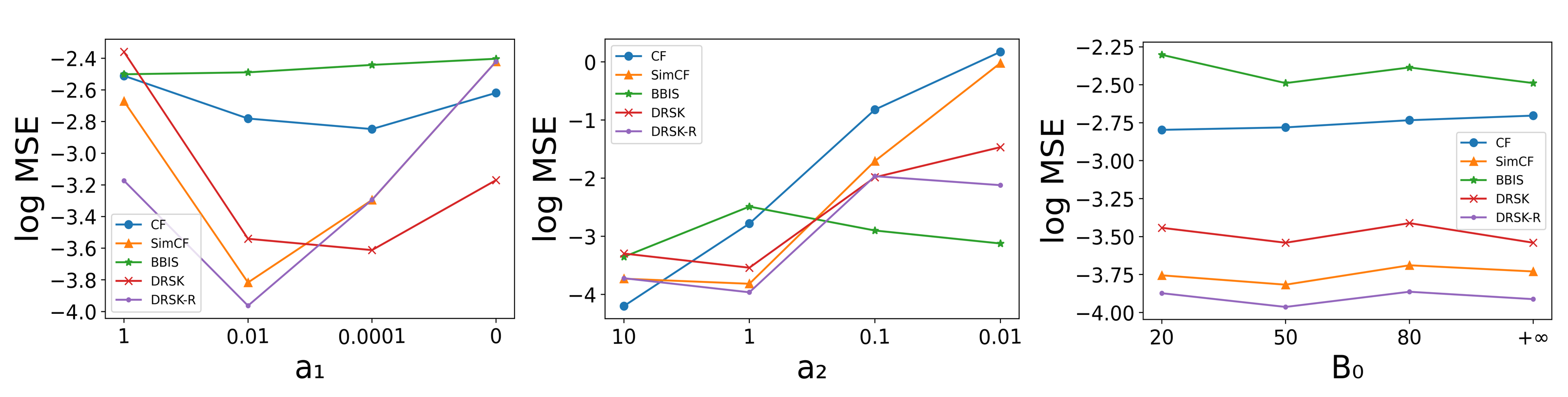}}
\caption{Illustration of different hyperparameters: $\lambda=a_1\times m^{-\frac{1}{2}}$, $h_1=a_2\times PSD$. }\label{illustration2}
\end{center}
\end{figure*}

The results are shown in Figure \ref{illustration2}. We observe that

\begin{enumerate}
\item $\lambda=0.01m^{-\frac{1}{2}}$ with $a_1=0.01$ appears to be a reasonable choice.
Larger or smaller choices of $\lambda$ such as $\lambda=m^{-\frac{1}{2}}$ or $\lambda=0$ produce less satisfactory performance. In particular, $\lambda > 0$ is recommended to prevent overfitting and stabilize the inverse numerically.
\item $h_1=a_2 \times PSD$ with $a_2=1$ is the choice recommended by previous studies \citep{LL,gretton2012kernel}. We recognize that other choices of $h_1$ may be better. However, tunning $h_1$ is not easy in practice. Note that unlike \citet{CF1}, $h_1$ cannot be selected via cross validation in our problem setting because the biased generated distribution does not allow accurate estimation of $\theta$ on the validation data. We have tried some alternative validation approaches such as minimizing variance from different subset divisions but did not see any consistent gain from validation. Therefore, we follow the choice from these previous studies \citep{LL,gretton2012kernel}.
\item $B_0$ is an insensitive hyperparameter. Different choices of $B_0$, including the ``$+\infty$" in the original
BBIS, produce very similar results as long as $B_0$ is relatively large.
\end{enumerate}

\subsection{Results on a Range of Problems} \label{sec:multipledistributions}
In this section, we conduct experiments on %All these methods are evaluated on datasets drawn from 
a variety of data distributions ranging from light-tailed (e.g., the mixture of Gaussian distribution) to heavy-tailed distributions (e.g., t-distribution). Moreover, we consider different dimensions of the input data to show the dimensionality effect on each estimator: $d=1,2,4$.

Throughout the following experiments, the ground-truth $\theta$ is obtained by drawing $10^6$ i.i.d. samples $(x_i,y_i)$ from the target distribution $\pi$ and calculating the sample average $\theta = \frac{1}{10^6}\sum_{i=1}^{10^6}f(x_i,y_i)$. The hyperparemeters are chosen based on the discussion at the beginning of Section \ref{sec:numerics}. \\  %on these synthetic distributions so as to accurately evaluate the performance. \\

\noindent\textbf{Mixture of Gaussian Distributions:} Suppose $\pi_X$ is the pdf of $(0.7\times \mathcal{N}(2,1)+0.3\times\mathcal{N}(1,1))^{\otimes d}$ ($^{\otimes d}$ represents d-dimensional independent component), $q_X$ is the pdf of $(\mathcal{N}(1,1))^{\otimes d}$.  $\pi_{Y|X}=q_{Y|X}$ is the distribution of $\mathcal{N}(0,0.001^3)$. Our goal is to compute the expectation of $f(X,Y) = \sin( \frac{\pi}{d} \sum_{i=1}^d X_i)+Y$ under the distribution of $\pi$.
The results are displayed in Figure \ref{exp:mixture}.
\\
%In Figure \ref{exp5p}, we notice that DRSK has the fastest rate compared with CF and BBIS. 

%\begin{figure*}
%\caption{Comparisons of different methods for Mixture of Gaussian distribution} \label{exp5p}
%\begin{center}
%\includegraphics[width=0.5\textwidth]{download5.png}
%\end{center}
%\end{figure*}

\begin{figure*}
\begin{center}
\makebox[\textwidth][c]{\includegraphics[width=1.15\textwidth]{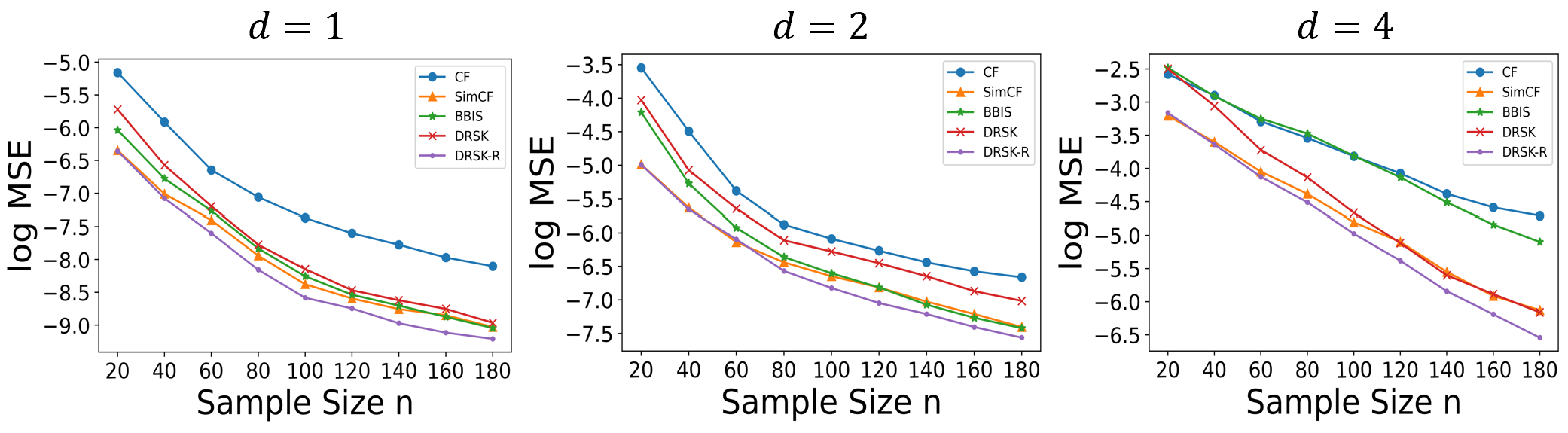}}
\end{center}
\caption{Comparisons of different methods on mixture of Gaussian distributions.} \label{exp:mixture}
\end{figure*}

\noindent\textbf{Generalized Student's T-distributions:} Suppose $\pi_X$ is the pdf of $(t_3(1,1))^{\otimes d}$ ($t_3(1,1)=1+1\times t_3$ where $t_3$ is the standard t-distribution with $3$ degrees of freedom) and $q_X$ is the pdf of $\mathcal{N}(\textbf{1},I_d)$.  $\pi_{Y|X}=q_{Y|X}$ is the distribution of $\mathcal{N}(0,0.001^3)$. Our goal is to compute the expectation of $f(X,Y)=\cos( \frac{\pi}{d} \sum_{i=1}^d X_i)+Y$ with respect to the distribution of $\pi$. %Note that the sample mean is a biased estimator in this setting. 
The results are displayed in Figure \ref{exp:t}. 
\\

\begin{figure*}
\begin{center}
\makebox[\textwidth][c]{\includegraphics[width=1.15\textwidth]{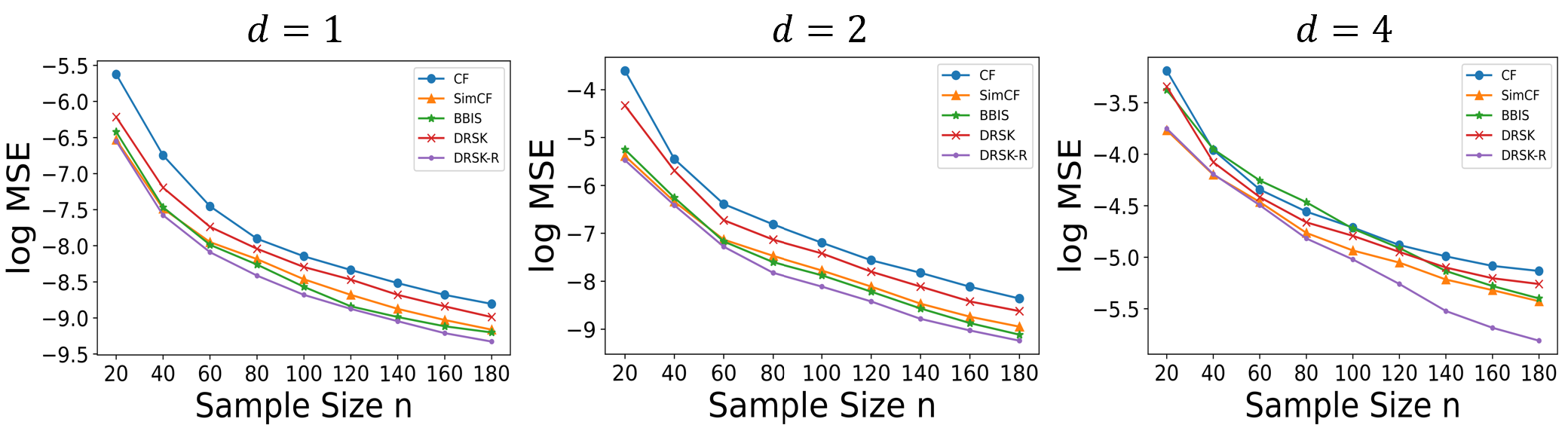}}
\end{center}
\caption{Comparisons of different methods on generalized Student’s t-distributions.} \label{exp:t}
\end{figure*}

Next, we consider a couple of Bayesian problems. Here, for each problem, we are interested in computing an expectation under the posterior distribution of an unknown input parameter. Typically in a realistic Bayesian setting, we do not have the exact form of the posterior distribution %because of the difficulty in exact integration. To mimic this setting, in our experiment we would use a sampling distribution different from the exact posterior. 
but can obtain it up to a normalizing constant. To simulate the posterior, we draw samples by running the parallel Metropolis–Hastings MCMC algorithm \citep{LL,rosenthal2000parallel}. To be specific, suppose we want to draw $n$ samples $x_i$ ($i\in[n]$): 
\begin{enumerate}
\item Draw $n$ i.i.d. samples $x^0_i$ ($i\in[n]$) from the prior distribution;
\item For each $i$ ($i\in[n]$), run Metropolis-Hastings MCMC algorithm with a symmetric Gaussian proposal kernel starting from $x^0_i$, which produces $n$ independent Markov chains;
\item For each $i$ ($i\in[n]$), take the endpoint after a finite number of iterations (e.g., 50 iterations) from the $i$-th Markov chains as $x_i$.
\end{enumerate}
Obviously, this algorithm will end up with a biased distribution $x_i \sim q_X$ instead of $\pi_X$. Moreover, $x_1,\cdots,x_n$ are independent because they are from different Markov chains. 

In the following experiments, to be able to validate our estimators against the ground truth, we consider conjugate distributions so that the posterior distribution can be derived explicitly. Note that this information is only used to obtain the ground-truth $\theta$ for evaluation, but not used for constructing the estimators. \\

% This framework falls into the range of our problem setting.

\noindent\textbf{Gamma Conjugate Distributions:} Suppose $x$ is an unknown input parameter. Let $x$ have a prior distribution $p_0(x)=(\text{Gamma}(2,2))^{\otimes d}$ (a Gamma distribution with a shape parameter $2$ and a rate parameter $2$). To estimate $x$, collect a set of data
$\Xi=\{\xi^{(l)}=(\xi_{1}^{(l)},\cdots,\xi_{d}^{(l)}): l=1,\cdots,L\}$ which are i.i.d. drawn from the likelihood $p(\Xi|x)= (\text{Gamma}(4,x))^{\otimes d}$ (a Gamma distribution with a shape parameter $4$ and a rate parameter $x$). The distribution of interest is the posterior distribution, which is $\text{Gamma}(4L+2,\sum_{l=1}^L\xi_1^{(l)}+2)\otimes \cdots \otimes \text{Gamma}(4L+2,\sum_{l=1}^L\xi_d^{(l)}+2)$ as it is a conjugate distribution. To mimic $\pi_X$, $q_X$ is obtained by running the parallel Metropolis–Hastings MCMC algorithm as described in this Section.
%To simulate $\pi_X$, we draw samples by running independent parallel Metropolis–Hastings MCMC algorithm. Precisely, suppose we want to draw $n$ samples $x_i$ ($i\in[n]$). First, we draw $n$ i.i.d. samples from the prior distribution $p_0(x)$
%and taking the endpoint after $100$ iterations from each independent Markov chains. This algorithm will output a biased distribution $q_X$ instead of $\pi_X$.  %After a finite number of iterations, the algorithm will output a distribution $q_X$ which is close to $\pi_X(x)$ but not exactly the same. 
%We let our sampling distribution $q_X$ be the prior distribution for simplicity. 
$\pi_{Y|X}=q_{Y|X}$ is the distribution of $\mathcal{N}(0,0.001^3)$. Suppose $L=12$ and $\sum_{l=1}^L \xi_i^{(l)}=3+5i$. Our goal is to compute the expectation of $f(X,Y) = \sin( \frac{\pi}{d} \sum_{i=1}^d X_i)+Y$ under the distribution $\pi$. The results are displayed in Figure \ref{exp:gamma}.  %with respect to the joint posterior distribution $\pi$ (equivalently, the posterior predictive distribution of $Y$, which is a generalized Beta prime distribution), which has a ground-truth value $\frac{4}{49} \sum_{i=1}^d (19+i)$. The results are displayed in Figure \ref{exp:gamma}. 
\\
%We discard the first 1000 ``burn-in" samples in the Metropolis–Hastings algorithm.

\begin{figure*}
\begin{center}
\makebox[\textwidth][c]{\includegraphics[width=1.15\textwidth]{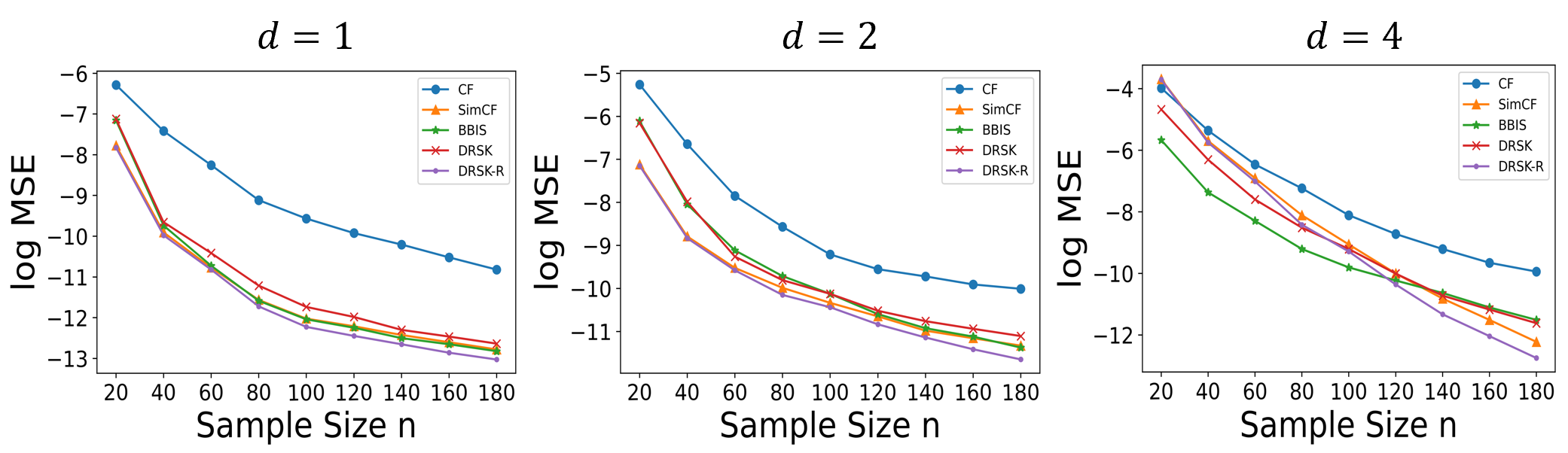}}
\end{center}
\caption{Comparisons of different methods on Gamma conjugate distributions.} \label{exp:gamma}
\end{figure*}

%In Figure \ref{exp4p}, we notice that DRSK is comparable to BBIS, and better than CF. 

%\begin{figure*}
%\caption{Comparisons of different methods for Gamma distributions} \label{exp4p}
%\begin{center}
%\includegraphics[width=0.5\textwidth]{download4.png}
%\end{center}
%\end{figure*}

\noindent\textbf{Beta Conjugate Distribution:} Suppose $x$ is an unknown input parameter. Let $x$ have a prior distribution $p_0(x)=(\text{Beta}(1,1))^{\otimes d}$ (a Beta distribution with two shape parameters $1$ and $1$). To estimate $x$, collect a set of data $\Xi=\{\xi^{(l)}=(\xi_{1}^{(l)},\cdots,\xi_{d}^{(l)}): l=1,\cdots,L\}$ which are i.i.d. drawn from the likelihood $p(\Xi|x)=\text{Bernoulli}(x)$ where the ``success" parameter is $x$. The distribution of interest is the posterior distribution $\pi_X(x) \propto p(\Xi|x)p_0(x)$, which is $\text{Beta}(\sum_{l=1}^Ly_1^{(l)}+1,L-\sum_{l=1}^Ly_1^{(l)}+1)\otimes \cdots \otimes\text{Beta}(\sum_{l=1}^Ly_d^{(l)}+1,L-\sum_{l=1}^Ly_d^{(l)}+1)$ as it is a conjugate distribution. To mimic $\pi_X$, $q_X$ is obtained by running the parallel Metropolis–Hastings MCMC algorithm as described in this Section. $\pi_{Y|X}=q_{Y|X}$ is the distribution of $\mathcal{N}(0,0.001^3)$. Suppose $L=11$ and $\sum_{l=1}^Ly_i^{(l)}=1+i$. Our goal is to compute the expectation of $f(X,Y) = \cos( \frac{\pi}{d} \sum_{i=1}^d X_i)+Y$ under the distribution $\pi$. The results are displayed in Figure \ref{exp:beta}. %with respect to the joint posterior distribution $\pi$ (equivalently, the posterior predictive distribution of $Y$, which is a Bernoulli distribution), which has a ground-truth value $\frac{1}{34} \sum_{i=1}^d (5+i)$. Let $h_1=10^4$. 
\\

\begin{figure*}
\begin{center}
\makebox[\textwidth][c]{\includegraphics[width=1.15\textwidth]{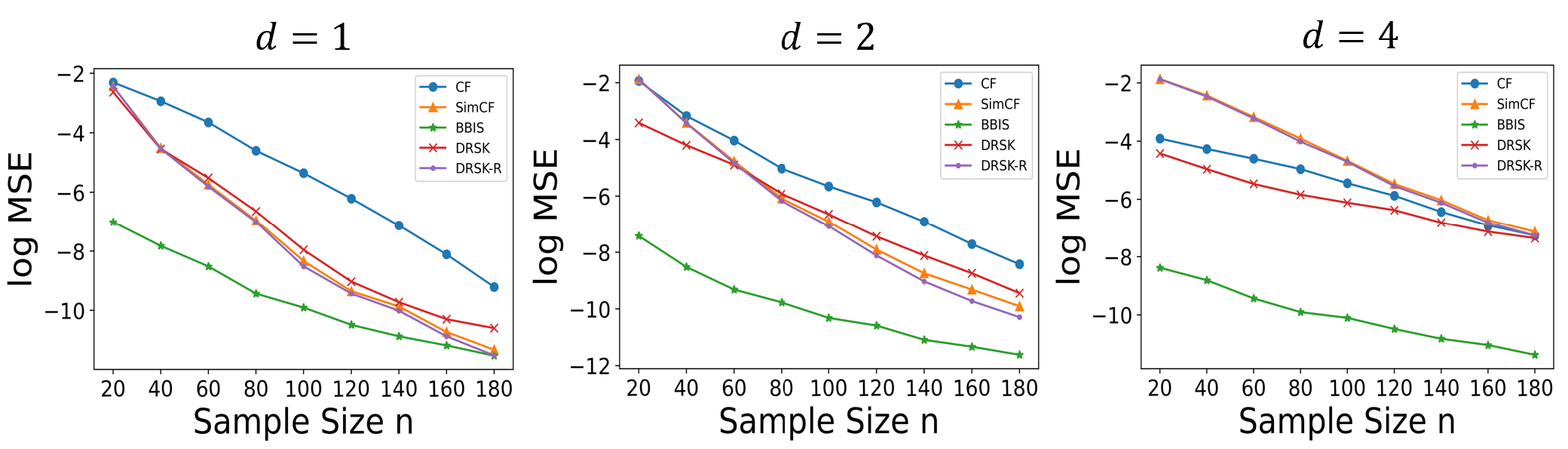}}
\end{center}
\caption{Comparisons of different methods on Beta conjugate distributions.} \label{exp:beta}
\end{figure*}

%\noindent\textbf{Inverse-Gamma Conjugate Distribution:} Let $\mathcal{Y}=\{y^{(l)}=(y_{1}^{(l)},\cdots,y_{d}^{(l)}): l=1,\cdots,L\}$ be a set of observed data which are i.i.d. drawn from the likelihood $(\mathcal{N}(0,x))^{\otimes d}$ where the variance parameter $x$ has a prior distribution $p_0(x)=(\text{InvGamma}(2,2))^{\otimes d}$ (an inverse-gamma distribution with a shape parameter $2$ and a scale parameter $2$). The distribution of interest is the posterior distribution $\pi_X(x) \propto p(\mathcal{Y}|x)p_0(x)$. Note that $\pi_X(x)$ is the pdf of $\text{InvGamma}(\frac{L}{2}+2,\frac{1}{2}\sum_{l=1}^L(y_1^{(l)})^2+2)\otimes \cdots \otimes\text{InvGamma}(\frac{L}{2}+2,\frac{1}{2}\sum_{l=1}^L(y_1^{(l)})^2+2)$ as it is a conjugate distribution. We let $q_X$ be the prior distribution as our sampling distribution for simplicity. $\pi_{Y|X}=q_{Y|X}$ is the pdf of $(\mathcal{N}(0,x))^{\otimes d}$. $f(x,y)=y$. Our goal is to compute the expectation of $f(X,Y)=Y_1+\cdots+Y_d$ with respect to the posterior distribution $\pi$ (equivalently, the posterior predictive distribution of $Y$, which is a generalized Student's t-distribution).In this experiment, we choose $L=20$, $\sum_{l=1}^L(y_i^{(l)})^2=2+\frac{i}{2}$, $h_1=10^5$, $m=0.5n$, $\lambda=0.1*m^{-\frac{1}{2}}$ and $B_0=50$. 

%In Figure \ref{exp3p},

Finally, we consider a stochastic simulation experiment, under input uncertainty where the modeler needs to provide performance measure estimate that accounts for both the aleatory noise from the simulation model and the epistemic noise from the input parameter that is handled via Bayesian posterior.\footnotemark\footnotetext{Note that here we focus on the estimation of performance measures when the inputs parameters are informed directly via input-level data. This is in contrast to Bayesian inverse problems \citep{stuart2010inverse} that involve possibly data at the output level. It is unclear, though possible, that our estimator can apply to this latter setting, in which case it would warrant a separate future work.}
% Despite the presence of both epistemic and aleatory uncertainty like in our setting, it's unclear yet if our approach can apply to the \textit{ but In the latter, it's probably true that we would need to simulate the posterior that involves both uncertainties in some way. It's possible that our approach can apply, but it requires some thoughts to understand if the problem structure there can sustain the use of our estimator
\\

\noindent\textbf{Computer Communication Network:} 
% \subsection{A Real-World Example: Computer Communication Network} \label{sec:real}
We conduct experiments on a computer communication network example borrowed from \citet{lam2021subsampling,lin2015single}.  Figure \ref{fig1} illustrates the structure of this computer communication network. It consists of 4 message-processing units (nodes) which are connected by 4 transport channels (edges). There are external messages that will enter the network. We assume that the lengths of the external messages are i.i.d. and follow an exponential distribution with rate $\frac{1}{300}$. For every pair $(i, j)$ of nodes ($i \ne j$), %the external messages that are to be transmitted from unit $i$ to unit $j$
external messages arriving at node $i$ that are to be transmitted to node $j$ follow a Poisson arrival process with an unknown rate $\lambda_{i,j}$ where their transmission path is fixed and known. %The values for $\lambda_{i,j}$’s are given in Table \ref{arrivalrates} and assumed to be known. 
Suppose that each node takes a constant time of $0.001$ seconds to process a message with unlimited storage capacity, and each edge has a capacity of 275000 bits. All messages transmit through the edges with a constant velocity of 150000 miles per second, and the $i$-th edge has
a length of $100 i$ miles for $i = 1, 2, 3, 4$. Therefore, the total time that a message of length $l$ bits occupies the $i$-th edge is
$\frac{l}{275000} +
\frac{100i}{150000}$ seconds. Suppose that the computer network is empty at time zero. 

\begin{figure}
    \centering
    \includegraphics[width=0.7\textwidth]{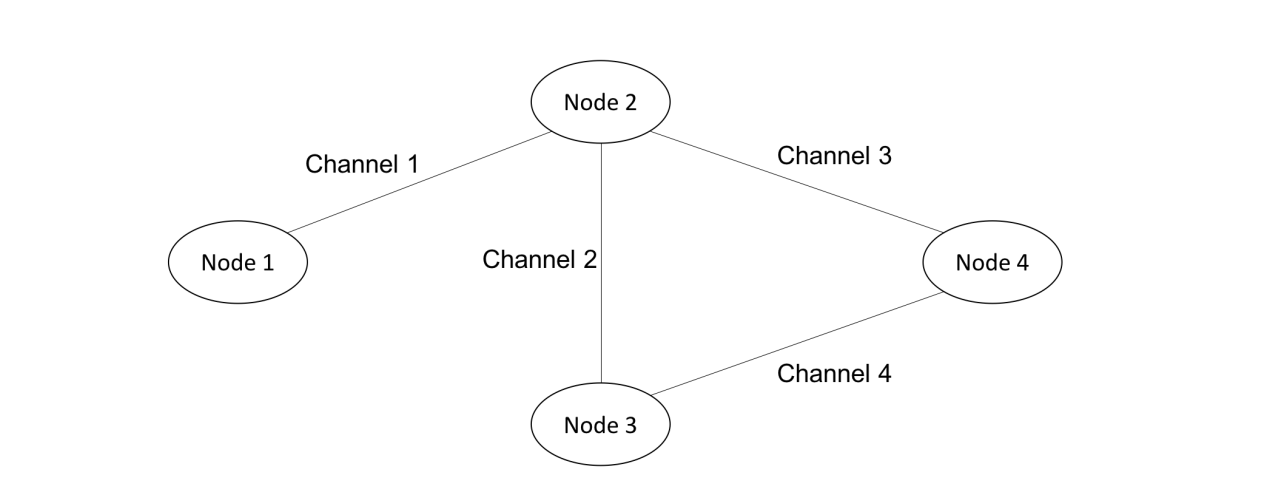}
    \caption{A computer communication network with four nodes and four channels.}
    \label{fig1}
\end{figure}

\begin{table*}[t]
\centering
\makebox[\textwidth][c]{{\small
\begin{tabular}{|c|cccc|}
\multicolumn{5}{c}{Setting (A)}\\
\hline
\diagbox{node $i$}{node $j$} & 1 & 2 & 3 & 4 \\
 \hline
 1 & n.a. & 0.5 & 0.7 & 0.6 \\
 2 & 0.4 & n.a. & 0.4 & 1.2 \\
 3 & 0.3 & 1.2 & n.a. & 1.0 \\
 4 & 0.8 & 0.7 & 0.5 & n.a. \\
 \hline
\end{tabular}
\
\begin{tabular}{|c|cccc|}
\multicolumn{5}{c}{Setting (B)}\\
\hline
\diagbox{node $i$}{node $j$} & 1 & 2 & 3 & 4 \\
 \hline
 1 & n.a. & 0.3 & 0.5 & 0.4 \\
 2 & 0.2 & n.a. & 0.3 & 1.1 \\
 3 & 0.2 & 1.0 & n.a. & 0.9 \\
 4 & 0.5 & 0.4 & 0.3 & n.a. \\
 \hline
\end{tabular}
\
\begin{tabular}{|c|cccc|}
\multicolumn{5}{c}{Setting (C)}\\
\hline
\diagbox{node $i$}{node $j$} & 1 & 2 & 3 & 4 \\
 \hline
 1 & n.a. & 0.4 & 0.6 & 0.5 \\
 2 & 0.3 & n.a. & 0.4 & 1.0 \\
 3 & 0.3 & 1.2 & n.a. & 0.7 \\
 4 & 0.6 & 0.5 & 0.4 & n.a. \\
 \hline
\end{tabular}}}
\caption{Interarrival times for estimating the arrival rates $\lambda_{i,j}$. These tables report the cumulative value of $0.1+\sum_{l=1}^L\xi^{(l)}_{i,j}$ where $L=10$ and $i,j=1,\cdots,4$ in three different settings of ground-truth $\lambda$.}
\label{arrivalrates}
\end{table*}
\begin{figure*}[t]
\begin{center}
\makebox[\textwidth][c]{\includegraphics[width=1.15\textwidth]{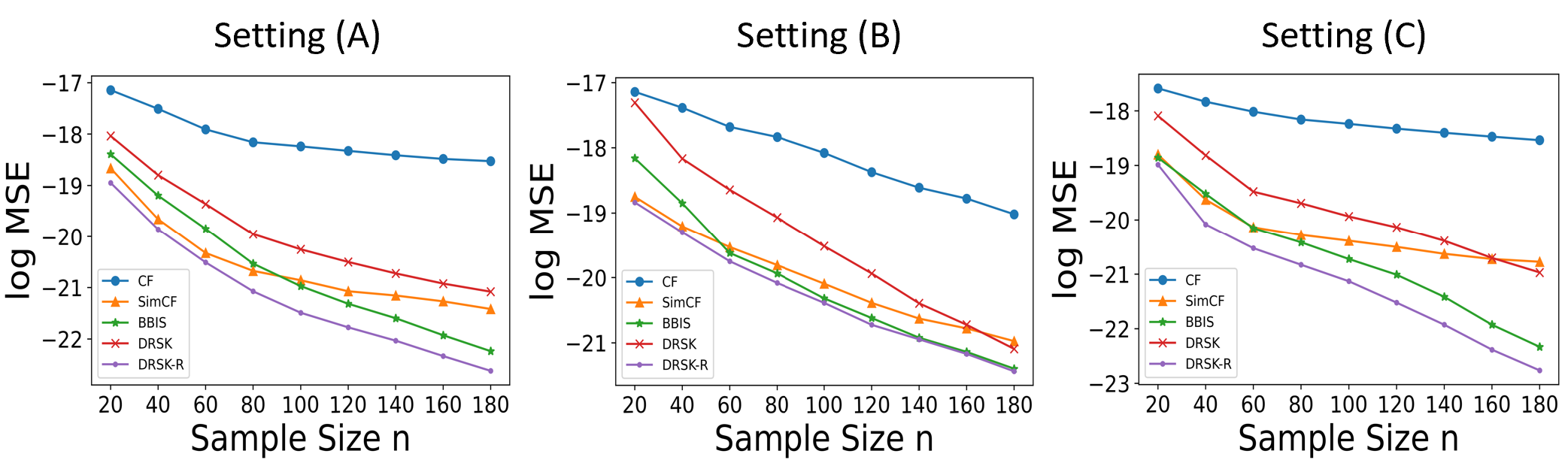}}
\end{center}
\caption{Comparisons of different methods on the computer communication network. Settings (A),(B),(C) correspond to the different data in Table \ref{arrivalrates} with different settings of ground-truth $\lambda$.} \label{exp:real}
\end{figure*}

The value of
interest is the delay of the first 30 messages that arrive in the system, or $\frac{1}{30}\sum_{i=1}^{30} T_k$, where $T_k$ is the time that the $k$-th message takes to be transmitted from its entering node to destination node. To estimate this value, for each input rate vector $\lambda=(\lambda_{i,j})_{i,j=1,2,3,4; i\ne j}\in \mathbb{R}_+^{12}$, we simulate the system $100$ times and take the average $f(\lambda, Y):= \frac{1}{100}\sum_{j=1}^{100} \frac{1}{30}\sum_{i=1}^{30} T^{(j)}_k(\lambda)$ as our simulation output. Here, $T^{(j)}_k(\lambda)$ is the delay of the $k$-th message in the $j$-th simulation when the input rate vector is given by $\lambda$. Obviously, $f(\lambda, Y)$ depends on not only $\lambda$, but also the stochasticity in the simulation model, denoted as $Y$ (since $T^{(j)}_k$ is a random output from the simulation model). Let $\pi_{Y|\lambda}$ denote the distribution of $Y$ given $\lambda$.

We apply Bayesian inference to estimate the unknown arrival rate vector $\lambda$. %Similar to the Bayesian setting in Section \ref{sec:multipledistributions}, 
First we endow $\lambda$ with a prior distribution $(\text{Gamma}(10,0.1))^{\otimes d}$. By the definition of $\lambda$, the interarrival time of the external messages arriving at node $i$ that are to be transmitted to node $j$, denoted as $\xi_{i,j}$, follow an exponential distribution with rate $\lambda_{i,j}$. To infer $\lambda$, we collect some historical data of the interarrival times $\xi^{(1)}_{i,j},\cdots,\xi^{(L)}_{i,j}$ and use them to update the posterior distribution of $\lambda_{i,j}$. Suppose $L=10$. These data are summarized in Table \ref{arrivalrates}, each obtained from a different setting of ground-truth $\lambda$.

The exact posterior distribution of $\lambda$, denoted as $\pi_\lambda$, is known to be $\text{Gamma}(10+L,\sum_{l=1}^L\xi_{1,2}^{(l)}+0.1)\otimes \cdots \otimes \text{Gamma}(10+L,\sum_{l=1}^L\xi_{4,3}^{(l)}+0.1)$. Let $\pi:=\pi_{\lambda}\pi_{Y|\lambda}$. However, this information is only used to obtain the ground-truth $\theta:=\mathbb{E}_\pi[f(\lambda,Y)]$ (by drawing $10^6$ i.i.d. samples as we discussed at the beginning of Section \ref{sec:multipledistributions}), but not used for constructing the estimators. To construct the estimators, we run the parallel Metropolis–Hastings MCMC algorithm as described in this Section to mimic the posterior, obtain i.i.d. samples $\lambda^{(1)},\cdots,\lambda^{(n)}\sim q_\lambda$ with sample size $n$, and plug into $f(\lambda^{(i)},Y)$ ($i\in[n]$). The results are displayed in Figure \ref{exp:real}. 

We conclude our findings from the numerical results. 
\begin{enumerate}
\item Overall, our experiments illustrate the high effectiveness of our estimators DRSK-R (DRSK) in terms of reducing the MSE in Monte Carlo computation. DRSK-R appears to be a better alternative to DRSK, outperforming DRSK in almost all experiments. %SimCF appears to be a better alternative to CF.

\item The relative performance of DRSK-R compared with other methods is stable. DRSK-R maintains superior performance regardless of the dimension of the input data, up to 12 in the real-world computer communication network example. In some cases (e.g., Figures \ref{exp:mixture} and \ref{exp:t}), its superiority becomes more distinct when the dimension increases. In addition, DRSK-R achieves better results than CF and SimCF in almost all experiments. DRSK-R outperforms BBIS in most experiments while the relative performance of BBIS tends to be variable among different experiments.

\item In the realistic computer communication network example consisting of 12 estimated input parameters and a sophisticated black-box simulation model, the experiments indicate that our DRSK-R achieves the best results consistently across multiple model settings and also shows its better relative performance against CF (SimCF).

\item DRSK-R may be less effective when the sampling distribution is exact (e.g., Plot (A3) in Figure \ref{illustration1}), or for some specific target distributions (e.g., Beta distributions in Figure \ref{exp:beta}). It may sometimes perform similarly to BBIS (e.g., Setting (A2) in Figure \ref{exp:real}) or SimCF (e.g., Setting (B3) in Figure \ref{illustration1}). Nevertheless, by taking advantage of both CF and BBIS, it achieves the smallest MSE in most experiments. 

\end{enumerate}

\section{Theory for Regularized Least Square Regression}  \label{sec:RLS}

In this section, we first review basic facts about RLS regression (aka KRR), whose framework follows \citet{SZ2} and \citet{SW}. Next, we develop some theoretical results for the purpose of our analysis in Section \ref{sec:analysis}. %The notations for this section are valid only in this section.

Suppose we have i.i.d. samples $\{(x_j,f(x_j,y_j)):j=1,\cdots,m\}$ where $(x_i,y_i)$ are i.i.d. drawn from \textit{the sampling distribution} $q$. %For this section, $\pi$ is the data distribution and it may not be not the same as in Section \ref{sec:mainres}. (For example, if the samples are drawn from another distribution $q$, then $\pi$ here should be understood as $q$.)  
We denote $\hat{z}=(f(x_1,y_1),\cdots, f(x_m,y_m))^T$. 
We call $f_q$ the \textit{(ground-truth) regression function} defined by
$$f_q(x)=\mathbb{E}_q [f(X,Y)|X=x].$$%which is $\bar f$ in the setting in Section \ref{sec:setup}.
	
%     We will use the method of regularized least square regression.  (For now, $\mathcal{H}$ is a general RKHS and we do not assume it coincides with any RKHSs before.) 
The goal of KRR is to learn the regression function by constructing a ``good" approximating function $s_m$ from the data. Let $\mathcal{H}$ be a \textit{generic} RKHS with the associated kernel $k(x,y)$\footnotemark\footnotetext{In this Section, $\mathcal{H}$ is \textit{generic}, not necessarily associated with the ones introduced in Section \ref{sec:mainres}. We will specify $\mathcal{H}$ in Section \ref{sec:analysis} for the proofs of our theorems in Section \ref{sec:DRSK main}.}. Let $\|\cdot\|_{\mathcal{H}}$ denote the norm on $\mathcal{H}$. Note that $k_x=k(x,\cdot)$ is a function in $\mathcal{H}$.

For this section, we only impose the following assumption:
\begin{assumption}
$\kappa:=\sup_{x\in \Omega} \sqrt{k(x,x)}<\infty$, $M_0:=\mathbb{E}_q [(z-f_q(x))^2]<\infty$. For any  $g\in \mathcal{H}$, $g(x)$ is $q$-measurable. $f_q\in L^2(q_X)$. \label{A0}
\end{assumption}

It follows that for any  $g\in \mathcal{H}$,
$$\sup_{t\in \Omega}|g(t)|^2=\sup_{t\in \Omega}|\langle g,k_t\rangle|^2\le \sup_{t\in \Omega} \|g\|^2_\mathcal{H} \|k_t\|^2_\mathcal{H} \le \kappa^2 \|g\|^2_\mathcal{H}. $$
So under Assumption \ref{A0}, any $g\in \mathcal{H}$ is a bounded function. We point out the following inequality that we will use frequently:
\begin{equation} \label{equ:metric}
\|g\|_{L^p(q_X)}\le \kappa \|g\|_\mathcal{H}, \quad \forall 1\le p\le \infty.     
\end{equation}

The RLS problem is given by
$$s_m(x):= \mathop{\arg\min}_{g\in \mathcal{H}} \left\{ \frac{1}{m} \sum_{j=1}^{m} (f(x_j,y_j)-g(x_j))^2 +\lambda \|g\|^2_{\mathcal{H}} \right\}$$
where $\lambda>0$ is a regularization parameter which may depend on the cardinality of the training data set. % given by
% $$s_m(x):= \mathop{\arg\min}_{g\in \mathcal{H}_+} \left\{ \frac{1}{m} \sum_{j=1}^{m} (f(x_j,y_j)-g(x_j)) +\lambda \|g\|_{\mathcal{H}_+} \right\}$$
Note that the data needed in this problem are only the values $\{x_j\}_{j=1,\cdots,m}$ and $\{z_j:=f(x_j,y_j)\}_{j=1,\cdots,m}$, so $y_j$ and $f$ can be in a black box. A nice property of RLS is that there is an closed-form formula for its solution, as stated below.
\begin{lemma}\label{RLS1}
	Let $K=(k(x_i,x_j))_{m\times m}\in \mathbb{R}^{m\times m}$ be the kernel Gram matrix and
	$$\hat{k}(x)=(k(x_1,x), \cdots,  k(x_m,x))^T.$$ 
	Then the RLS solution is given as $s_m(x)= \beta^T \hat{k}(x)$ where $\beta =(K +\lambda m I)^{-1} \hat{z}$.
\end{lemma}

Lemma \ref{RLS1} shows a direct way to calculate $s_m$. %Nevertheless, computing the KRR estimator can be prohibitively expensive for large datasets. Thus, the design of scalable methods for KRR (and other kernel based methods) has been the focus of intensive research in recent years. One of the most popular approaches to scaling up kernel based methods is random Fourier features sampling,  originally proposed by \citet{RR}.
From a theoretical point of view, to derive the property of $s_m$, we need to use some tools from functional analysis. To begin, we give an equivalent form of $s_m$ in terms of linear operators. Define the sampling operator $S_x : \mathcal{H} \to \mathbb{R}^m$ associated with a discrete subset $\{x_i\}^{m}_{i=1}$ by
$$S_x(g)=(g(x_i))^{m}_{i=1} , \ g\in \mathcal{H}.$$
The adjoint of the sampling operator,  $S^T_x : \mathbb{R}^m \to \mathcal{H}$, is given by
$$S^T_x(c)=\sum_{i=1}^{m} c_i k_{x_i} , c\in \mathbb{R}^m.$$

Note that the compound mapping $S^T_x S_x$ is a positive self-adjoint operator on $\mathcal{H}$.
Let $I$ denote the identity mapping on $\mathcal{H}$. We have:

\begin{lemma}\label{RLS2}
	The RLS solution can be written as follows:
	$$s_m= \left(\frac{1}{m} S^T_x S_x +\lambda I\right)^{-1} \frac{1}{m} S^T_x(\hat{z}).$$\end{lemma}

A proof can be found in \citet{SZ2}. It is also easy to derive this result directly from Lemma \ref{RLS1}.

Denote 
$$\mathcal{H}_0^q:=\{g\in\mathcal{H}: g=0 \ \text{ a.e. with respect to } q\} \ \text{and}$$ $$\mathcal{H}_1^{q}:=(\mathcal{H}_0^{q})^\perp \text{, the orthogonal complement of } \mathcal{H}_0^{q} \text{ in } \mathcal{H}.$$
If $q$ is clear from the context, we write $\mathcal{H}_0=\mathcal{H}_0^q$,  $\mathcal{H}_1=\mathcal{H}_1^q$ for simplicity. Note that both $\mathcal{H}_0$ and $\mathcal{H}_1$ are closed subspaces in $\mathcal{H}$ with respect to the norm $\|\cdot\|_{\mathcal{H}}$. It is well-known that $\mathcal{H}/\mathcal{H}_0$ is isometrically isomorphic to $\mathcal{H}_1$. So  $\mathcal{H}_1$ is essentially the quotient space of $\mathcal{H}$ induced by the equivalence relation ``a.e. with respect to $q$", the same equivalence relation in $L^2(q_X)$. If $f\in \mathcal{H}$, we may replace the original $f\in \mathcal{H}$ with a equivalent $\Tilde{f}\in \mathcal{H}_1$ ($f=\Tilde{f}$ a.e. with respect to $q$) since this will not effect the estimation of the parameter. For this purpose, we may treat $\mathcal{H}_1$ as $\mathcal{H}$. (But they are substantially different in some way, see \citet{SW}.) Let $\overline{\mathcal{H}_1}^q$ (which is the same as $\overline{\mathcal{H}}^q$) be the closure  of $\mathcal{H}_1$ in $L^2(q_X)$. 

Next, we introduce a standard result from functional analysis. This result can be found in Theorem 2.4 and Proposition 2.10 in \citet{SO}.

\begin{theorem} \label{CFC}
Let $A$ be a bounded self-adjoint linear operator. Let $C(\sigma(A))$ be the set of real-valued continuous functions defined on the spectrum of $A$. Then for any $g\in C(\sigma(A))$, $g(A)$ is self-adjoint and $\|g(A)\|=\sup_{x\in \sigma(A)}|g(x)|.$
\end{theorem}

Define $L_q:L^2(q_X)\to L^2(q_X)$ as the integral operator
$$(L_q g)(x) :=\int_{\Omega} k(x,x')g(x')q(x')dx', \ x\in\Omega,\ g\in L^2(q_X).$$
This operator can be viewed as a linear operator on $L^2(q_X)$ or on $\mathcal{H}$. Unless specified otherwise, we always assume the domain of $L_q$ is $L^2(q_X)$. \citet{SW} shows that $L_q$ is a compact and positive self-adjoint operator on $L^2(q_X)$. Theorem \ref{CFC} shows that $L_q^{r}$ is a well-defined self-adjoint operator on $L^2(q_X)$ for $0\le r\le 1$. Denote $\text{Range}(L_q^r)$ the range of $L_q^r$ on the domain $L^2(q_X)$. When we write $L_q^{-r}g\in L^2(q_X)$, it should be understood that (1) $g\in \text{Range}(L_q^r)$, (2) $L_q^{-r}g$ is an element in the preimage set of $g$ under the operator $L_q^{r}$ on the domain $L^2(q_X)$.

We will frequently use the following lemma, which indicates a useful property of the integral operator $L_q$ (a proof can be found in \citet{SW}).
\begin{lemma} \label{isomorphism}
$L_q^{\frac{1}{2}} f \in \mathcal{H}_1$ for any $f\in L^2(q_X)$, and $L_q^{\frac{1}{2}}$ is an isometric isomorphism from $(\overline{\mathcal{H}_1}^q, \|\cdot\|_{L^2(q_X)})$ onto $(\mathcal{H}_1, \|\cdot\|_{\mathcal{H}})$.
\end{lemma}
Note that Lemma \ref{isomorphism} implies that $\text{Range}(L_q^{\frac{1}{2}})=\text{Range}(L_q^{\frac{1}{2}}|_{\overline{\mathcal{H}_1}^{q}})=\mathcal{H}_1$.

Next, consider an oracle or a data-free limit of $s_m$ as
$$f_\lambda :=\mathop{\arg\min}_{g\in \mathcal{H}} \left\{ \|g-f_q\|^2_{L^2(q_X)} + \lambda \|g\|^2_{\mathcal{H}} \right\}.$$
We have the following explicit expression (a proof can be found in \citet{CS}):
\begin{lemma}\label{RLS3}
	The solution of $f_\lambda$ is given as $f_\lambda= (L_q+\lambda I)^{-1} L_qf_q$.
\end{lemma}

To show that $s_m-f_q$ is small, we split it into two parts
\begin{equation}
s_m-f_q=(s_m-f_\lambda)+(f_\lambda-f_q).\label{decomposition}
\end{equation}
The first part in \eqref{decomposition} comes from the statistical noise in the RLS regression, whereas the second part can be viewed as the bias of the functional approximation. In terms of terminology in machine learning, the first term is called the estimation error (or sample error) and the second term called the approximation error. We study the asymptotic error of each part in the next set of results.
\begin{proposition}\label{second}
	Suppose that $L_q^{-r} f_q \in L^2(q_X)$ where $0 \le r\le 1$. Then\\
	(1) $\|f_\lambda-f_q\|_{L^2(q_X)}\le \lambda^{r} \|L_q^{-r} f_q\|_{L^2(q_X)}$.\\
	(2) $\|f_\lambda-f_q\|_{\mathcal{H}} \le \lambda^{r-\frac{1}{2}} \|L_q^{-r} f_q\|_{L^2(q_X)}$ for $\frac{1}{2}\le r\le 1$ only.
\end{proposition}

\begin{proof}
    We remark that $ \|L_q^{-r} f_q\|_{L^2(q_X)}$ measures a complexity of the regression function \citep{SZ3}. Using Lemma \ref{RLS3}, we write
	$$f_\lambda-f_q=(L_q+\lambda I)^{-1} L_qf_q -f_q=-\lambda(L_q+\lambda I)^{-1} f_q.$$
	(1) For the first part, we write
	$$\|f_\lambda-f_q\|_{L^2(q_X)}=\lambda \|(L_q+\lambda I)^{-1}L_q^{r} L_q^{-r} f_q\|_{_{L^2(q_X)}}\le \lambda \|(L_q+\lambda I)^{-1}L_q^{r}\| \|L_q^{-r} f_q\|_{_{L^2(q_X)}}.$$
	Here we regard $L_q$ as a bounded positive self-adjoint operator on $L^2(q_X)$. Then Theorem \ref{CFC} states that $\|(L_q+\lambda I)^{-1}L_q^r\|\le\|h\|_\infty$ 	where $h(x)=\frac{x^r}{x+\lambda}$ is defined on $x\in [0,\infty)$ ($L_q$ is positive so $\sigma(L_q) \subset [0,\infty)$). Note that $\|h\|_\infty\le \lambda^{r-1}$. So
   	$$ \|f_\lambda-f_q\|_{L^2(q_X)}\le \lambda \lambda^{r-1} \|L_q^{-r} f_q\|_{_{L^2(q_X)}}^2= \lambda^{r} \|L_q^{-r} f_q\|_{L^2(q_X)}.$$
   	(2) For the second part, we exhibit an intuitive proof here. We write
	$$\|f_\lambda-f_q\|_{\mathcal{H}} =\lambda \|L_q^{\frac{1}{2}} L_q^{-\frac{1}{2}}(L_q+\lambda I)^{-1} f_q\|_{\mathcal{H}}=\lambda \| L_q^{-\frac{1}{2}}(L_q+\lambda I)^{-1} f_q\|_{L^2(q_X)}.$$
	The last equality is intuitively correct due to Lemma \ref{isomorphism}. (However, this statement is not rigorous because $ L_q^{-\frac{1}{2}}(L_q+\lambda I)^{-1} f_q$ is not necessarily in $\overline{\mathcal{H}_1}^q$.  A rigorous argument can be found in \citet{SW}.) Next we notice that
	\begin{align*}
	\| L_q^{-\frac{1}{2}}(L_q+\lambda I)^{-1} f_q\|_{L^2(q_X)}&=\| L_q^{-\frac{1}{2}}(L_q+\lambda I)^{-1} L_q^{r} L_q^{-r} f_q\|_{L^2(q_X)}\\
	&\le \| L_q^{-\frac{1}{2}}(L_q+\lambda I)^{-1} L_q^{r}\| \|L_q^{-r} f_q\|_{L^2(q_X)} .   
	\end{align*}
	Theorem \ref{CFC} states that $\| L_q^{-\frac{1}{2}}(L_q+\lambda I)^{-1} L_q^{r}\|\le\|h\|_\infty$ 	where $h(x)=\frac{x^{r-\frac{1}{2}}}{x+\lambda}$ is defined on $x\in [0,\infty)$ ($L_q$ is positive so $\sigma(L_q) \subset [0,\infty)$). Note that $\|h\|_\infty\le \lambda^{r-\frac{3}{2}}$. So
   	$$\|f_\lambda-f_q\|_{\mathcal{H}}\le \lambda \lambda^{r-\frac{3}{2}} \|L_q^{-r} f_q\|_{L^2(q_X)}= \lambda^{r-\frac{1}{2}} \|L_q^{-r} f_q\|_{L^2(q_X)}.$$
\end{proof}

If we want to obtain a better bound for $f_\lambda-f_q$ by using this proposition, we may want $r$ to be as large as possible, but meanwhile $L_q^{-r} f_q \in L^2(q_X)$ becomes a more restrictive constraint. However, we have the following proposition that can bypass this tradeoff.

% this issue. We develop this new proposition for the usage of an interesting approximation technique proposed in the next section.

\begin{proposition}\label{ease}
	The range of $L_q$ satisfies
    $$\overline{\text{Range}(L_q)}^{q}=\overline{\text{Range}(L_q|_{\overline{\mathcal{H}_1}^{q}})}^{q}= \overline{\mathcal{H}_1}^{q}(=\overline{\mathcal{H}}^{q}). $$
\end{proposition}

\begin{proof}
    Take any $f_1\in \overline{\mathcal{H}_1}^{q}$. For any $\epsilon>0$, there exists $f_2\in \mathcal{H}_1$ such that $\|f_1-f_2\|_{L^2(q_X)}\le \epsilon$. It follows from Lemma \ref{isomorphism} that there exists $g_1\in \overline{\mathcal{H}_1}^{q}$ such that $L_q^{\frac{1}{2}} g_1=f_2$. There exists $g_2\in \mathcal{H}_1$ such that $\|g_1-g_2\|_{L^2(q_X)}\le \frac{\epsilon}{\kappa}$. Again, it follows from Lemma \ref{isomorphism} that there exists $h_1\in \overline{\mathcal{H}_1}^{q}$ such that $L_q^{\frac{1}{2}} h_1=g_2$. Then we have
    $$\|L_q h_1-f_2\|_{L^2(q_X)}\le \kappa\|L_q h_1-f_2\|_{\mathcal{H}}=\kappa\|L_q^{\frac{1}{2}} g_2-L_q^{\frac{1}{2}}g_1\|_{\mathcal{H}}=\kappa\|g_2-g_1\|_{L^2(q_X)}\le \epsilon$$
    and
    $$\|L_q h_1-f_1\|_{L^2(q_X)}\le \|L_q h_1-f_2\|_{L^2(q_X)}+ \|f_2-f_1\|_{L^2(q_X)}\le 2\epsilon.$$
    This implies that
    $f_1\in \overline{\text{Range}(L_q|_{\overline{\mathcal{H}_1}^{q}})}^{q}$ so
    $$\overline{\text{Range}(L_q|_{\overline{\mathcal{H}_1}^{q}})}^{q}\supset \overline{\mathcal{H}_1}^{q}.$$
    On the other hand, Lemma \ref{isomorphism} indicates that
    $$\text{Range}(L_q|_{\overline{\mathcal{H}_1}^{q}})\subset \text{Range}(L_q) \subset \text{Range}(L_q^{\frac{1}{2}})= \mathcal{H}_1.$$
    Hence we have
    $$\overline{\text{Range}(L_q)}^{q}=\overline{\text{Range}(L_q|_{\overline{\mathcal{H}_1}^{q}})}^{q}= \overline{\mathcal{H}_1}^{q}. $$
\end{proof}

Denote $\overline{\text{Range}(L_q)}^{\mathcal{H}}$ the closure of $\text{Range}(L_q)$ in $\mathcal{H}$ with respect to the norm of $\mathcal{H}$.
We have the following similar proposition in terms of the norm in $\mathcal{H}$.
\begin{proposition}\label{ease2}
	The range of $L_q$ satisfies
    $$\overline{\text{Range}(L_q)}^{\mathcal{H}}=\overline{\text{Range}(L_q|_{\overline{\mathcal{H}_1}^{q}})}^{\mathcal{H}} = \mathcal{H}_1. $$
\end{proposition}

\begin{proof}
    Take any $f_1\in \mathcal{H}_1$. It follows from Lemma \ref{isomorphism} that there exists $g_1\in \overline{\mathcal{H}_1}^{q}$ such that $L_q^{\frac{1}{2}} g_1=f_1$. For any $\epsilon>0$, there exists $g_2\in \mathcal{H}_1$ such that $\|g_1-g_2\|_{L^2(q_X)}\le \epsilon$. Again, it follows from Lemma \ref{isomorphism} that there exists $h_1\in \overline{\mathcal{H}_1}^{q}$ such that $L_q^{\frac{1}{2}} h_1=g_2$ and we have
    $$\|L_q h_1-f_1\|_{\mathcal{H}}= \|L_q h_1-L_q^{\frac{1}{2}}g_1\|_{\mathcal{H}}= \|L_q^{\frac{1}{2}}h_1-g_1\|_{L^2(q_X)}=\|g_2-g_1\|_{L^2(q_X)}\le\epsilon.$$
    This implies that
    $f_1\in \overline{\text{Range}(L_q|_{\overline{\mathcal{H}_1}^{q}})}^{\mathcal{H}}$ so
    $$\overline{\text{Range}(L_q|_{\overline{\mathcal{H}_1}^{q}})}^{\mathcal{H}} \supset \mathcal{H}_1. $$
    On the other hand, Lemma \ref{isomorphism} indicates that
    $$\text{Range}(L_q|_{\overline{\mathcal{H}_1}^{q}})\subset \text{Range}(L_q) \subset \text{Range}(L_q^{\frac{1}{2}})= \mathcal{H}_1.$$
    Note that $\mathcal{H}_1$ is a closed subspace in $\mathcal{H}$. Hence we have
    $$\overline{\text{Range}(L_q)}^{\mathcal{H}}=\overline{\text{Range}(L_q|_{\overline{\mathcal{H}_1}^{q}})}^{\mathcal{H}} = \mathcal{H}_1. $$
\end{proof}

Propositions \ref{ease} and \ref{ease2} give a theoretical explanation that if we have a result in the space $\text{Range}(L_q)$, we may anticipate that it is also (approximately) valid in $\overline{\mathcal{H}}^{q}$ or $\mathcal{H}_1$.   

\begin{proposition}\label{first1}
	We have
    $$\|s_m-f_\lambda\|_{\mathcal{H}}\le \frac{1}{\lambda} \|\Delta\|_{\mathcal{H}}$$
	where
	$$\Delta:=\frac{1}{m} \sum_{i=1}^{m} (z_i-f_\lambda(x_i))k_{x_i}-L_q(f_q-f_\lambda).$$
\end{proposition}

\begin{proof}
	By definition,
	$$s_m-f_\lambda= \left(\frac{1}{m} S^T_x S_x +\lambda I\right)^{-1} \left(\frac{1}{m} S^T_x(z)- \frac{1}{m} S^T_x S_x f_\lambda-\lambda f_\lambda\right).$$
	Direct computation leads to
	$$\frac{1}{m} S^T_x(z)- \frac{1}{m} S^T_x S_x f_\lambda=\frac{1}{m} \sum_{i=1}^{m} (z_i-f_\lambda(x_i))k_{x_i},\ \text{and} \  \lambda f_\lambda=L_q(f_q-f_\lambda)$$
	so
	$$s_m-f_\lambda= \left(\frac{1}{m} S^T_x S_x +\lambda I\right)^{-1} \Delta.$$
	View $S^T_x S_x$ as a positive self-adjoint operator on $\mathcal{H}$. By Theorem \ref{CFC}, we have
	$$\left\|\left(\frac{1}{m} S^T_x S_x +\lambda I\right)^{-1}\right\|_{\mathcal{H}\to \mathcal{H}}\le \frac{1}{\lambda}.$$
	Hence we obtain
	$$\|s_m-f_\lambda\|_{\mathcal{H}}\le \frac{1}{\lambda} \|\Delta\|_{\mathcal{H}}$$
	as desired.
\end{proof}

Next we have the following:
\begin{proposition}\label{first2}
We have
	$$\mathbb{E}_q[\|\Delta\|^2_{\mathcal{H}}]\le
	\frac{1}{m} \kappa^2 (\nu_X((f_q-f_\lambda)^2)+M_0). $$
\end{proposition}

\begin{proof}
	Consider
    $$\|\Delta\|^2_{\mathcal{H}}=\langle\frac{1}{m} \sum_{i=1}^{m} ((z_i-f_\lambda(x_i))k_{x_i} -L_q(f_q-f_\lambda)), \frac{1}{m} \sum_{i=1}^{m} ((z_i-f_\lambda(x_i))k_{x_i}-L_q(f_q-f_\lambda))\rangle.$$
	Direct computation shows that, for the first term,
	$$\langle (z_i-f_\lambda(x_i))k_{x_i}, (z_j-f_\lambda(x_j))k_{x_j}\rangle= (z_i-f_\lambda(x_i))(z_j-f_\lambda(x_j))k(x_i,x_j)$$
	For $i\ne j$, we have
	\begin{align*}
&\mathbb{E}_q \langle (z_i-f_\lambda(x_i))k_{x_i}, (z_j-f_\lambda(x_j))k_{x_j}\rangle\\
=& \int_{\Omega}\int_{\Omega} (f_q(x_i)-f_\lambda(x_i)) (f_q(x_j)-f_\lambda(x_j)) k(x_i,x_j)q(x_i)q(x_j)dx_idx_j
	\end{align*}
	For $i= j$, we have
	$$\mathbb{E}_q \langle (z_i-f_\lambda(x_i))k_{x_i}, (z_j-f_\lambda(x_j))k_{x_j}\rangle= \mathbb{E}_q [(z_i-f_\lambda(x_i))^2  k(x_i,x_i)]$$
	For the cross item, we have
	$$\langle (z_i-f_\lambda(x_i))k_{x_i}, L_q(f_q-f_\lambda)\rangle= \int_{\Omega}k(x_i,x_j)(z_i-f_\lambda(x_i))(f_q(x_j)-f_\lambda(x_j))q(x_j)dx_j,$$
	so
	\begin{align*}
	&\mathbb{E}_q \langle (z_i-f_\lambda(x_i))k_{x_i}, L_q(f_q-f_\lambda)\rangle \\
	=& \int_{\Omega}\int_{\Omega} (f_q(x_i)-f_\lambda(x_i)) (f_q(x_j)-f_\lambda(x_j)) k(x_i,x_j)q(x_i)q(x_j)dx_idx_j.    
	\end{align*}
	For the last item, we have
	$$\langle L_q(f_q-f_\lambda), L_q(f_q-f_\lambda)\rangle= \int_{\Omega}\int_{\Omega} (f_q(x_i)-f_\lambda(x_i)) (f_q(x_j)-f_\lambda(x_j)) k(x_i,x_j)q(x_i)q(x_j)dx_idx_j,$$
	so
	\begin{align*}
	&\mathbb{E}_q \langle L_q(f_q-f_\lambda), L_q(f_q-f_\lambda)\rangle \\ 
	=& \int_{\Omega}\int_{\Omega} (f_q(x_i)-f_\lambda(x_i)) (f_q(x_j)-f_\lambda(x_j)) k(x_i,x_j)q(x_i)q(x_j)dx_idx_j.
	\end{align*}
	We observe that for $i\ne j$,
	$$\mathbb{E}_q \langle (z_i-f_\lambda(x_i))k_{x_i} -L_q(f_q-f_\lambda),  (z_j-f_\lambda(x_j))k_{x_j}-L_q(f_q-f_\lambda)\rangle \rangle=0.$$
	For $i= j$, we have
	\begin{align*}
	&\mathbb{E}_q \langle (z_i-f_\lambda(x_i))k_{x_i} -L_q(f_q-f_\lambda),  (z_j-f_\lambda(x_j))k_{x_j}-L_q(f_q-f_\lambda)\rangle \rangle\\
	=& \mathbb{E}_q [(z_i-f_\lambda(x_i))^2  k(x_i,x_i)] - \mathbb{E}_q \langle L_q(f_q-f_\lambda), L_q(f_q-f_\lambda)\rangle\\
	\le& \mathbb{E}_q [(z_i-f_\lambda(x_i))^2  k(x_i,x_i)]\\
	\le& \kappa^2 \mathbb{E}_q [(z_i-f_\lambda(x_i))^2]\\
	=& \kappa^2 (\mathbb{E}_q [(f_\lambda(x_i)-f_q(x_i))^2]+ \mathbb{E}_q [(z_i-f_q(x_i))^2])\\
	=& \kappa^2 (\nu_X((f_q-f_\lambda)^2)+M_0).
	\end{align*}
	Therefore
	$$\mathbb{E}_q[\|\Delta\|^2_{\mathcal{H}}]\le
	\frac{1}{m} \kappa^2 (\nu_X((f_q-f_\lambda)^2)+M_0).$$
\end{proof}

With this, we have the following estimate:
\begin{corollary}\label{cor1}
We have	
$$\mathbb{E}_q[\|s_m-f_\lambda\|^2_{\mathcal{H}}] \le \frac{\kappa^2 (\nu_X((f_q-f_\lambda)^2)+M_0)}{\lambda^2 m},$$
$$\mathbb{E}_q[\nu_X((s_m-f_\lambda)^2)]\le \frac{\kappa^4 (\nu_X((f_q-f_\lambda)^2)+M_0)}{\lambda^2 m}.$$
\end{corollary}

\begin{proof}
	Combining Proposition \ref{first1} and Proposition \ref{first2}, we obtain
	$$\mathbb{E}_q[\|s_m-f_\lambda\|^2_{\mathcal{H}}] \le \frac{\kappa^2 (\nu_X((f_q-f_\lambda)^2)+M_0)}{\lambda^2 m},$$
	and we also note that
	$$\nu_X((s_m-f_\lambda)^2)\le \kappa^2 \|s_m-f_\lambda\|^2_{\mathcal{H}}.$$
%     Therefore
% 	$$\mathbb{E}_q[\nu_X((s_m-f_\lambda)^2)]\le \frac{\kappa^4 \nu_X((f_q-f_\lambda)^2)}{\lambda^2 m} $$
\end{proof}

Corollary \ref{cor1} shows that to estimate $s_m-f_q=(s_m-f_\lambda)+(f_\lambda-f_q)$, we only need to handle $f_\lambda-f_q$.
Finally, putting everything together, we establish the following two corollaries that will be frequently used in our analysis:
\begin{corollary}\label{main_lemma}
	Suppose that $L_q^{-r} f_q \in L^2(q_X)$ where $0\le r\le 1$. Then 
	$$\mathbb{E}_q[\nu_X((f_q - s_m)^2)] \le \left(\frac{2 \kappa^4}{\lambda^{2-2r} m}+2\lambda^{2r}\right) \nu_X(( L_q^{-r} f_q)^2)+\frac{2 \kappa^4 M_0}{\lambda^2 m}. $$
	In particular, taking $\lambda = m^{-\frac{1}{2}}$, we have
	$$\mathbb{E}_q[\nu_X((f_q - s_m)^2)] \le C_\kappa m^{-r} \nu_X(( L_q^{-r} f_q)^2)+ 2 \kappa^4 M_0 $$
	where $C_\kappa=2\kappa^4+2$ only depends on $\kappa$.\\
	Taking $\lambda = m^{-\frac{1}{2+2r}}$, we have
	$$\mathbb{E}_q[\nu_X((f_q - s_m)^2)] \le C_1 m^{-\frac{r}{1+r}} (\nu_X(( L_q^{-r} f_q)^2)+M_0)$$
	where $C_1$ only depends on $\kappa$.
\end{corollary}

\begin{proof}
	Proposition \ref{second} shows that
    if $L_q^{-r} f_q \in L^2(q_X)$, then
    $$\nu_X((f_\lambda-f_q)^2)\le \lambda^{2r} \nu_X(( L_q^{-r} f_q)^2).$$
    We note that
	$$\nu_X((f_q - s_m)^2)\le 2(\nu_X((f_q - f_\lambda)^2)+\nu_X((f_\lambda - s_m)^2)).$$
	So taking the expectation and using Corollary \ref{cor1}, we have
	\begin{align*}
	\mathbb{E}_q[\nu_X((f_q - s_m)^2)] &\le \left(\frac{2 \kappa^4}{\lambda^2 m}+2\right) \nu_X((f_q-f_\lambda)^2)+\frac{2 \kappa^4 M_0}{\lambda^2 m}\\
	&\le \left(\frac{2 \kappa^4}{\lambda^{2-2r} m}+2\lambda^{2r}\right) \nu_X(( L_q^{-r} f_q)^2)+\frac{2 \kappa^4 M_0}{\lambda^2 m}.
	\end{align*}
    \end{proof}
 
 \begin{corollary}\label{main_lemma2}
	Suppose that $L_q^{-r} f_q \in L^2(q_X)$ where $\frac{1}{2}\le r\le 1$. Then
	$$\mathbb{E}_q[\|f_q-s_m\|^2_{\mathcal{H}}] \le \left(\frac{2 \kappa^2}{\lambda^{2-2r} m}+2\lambda^{2r-1}\right) \nu_X(( L_q^{-r} f_q)^2)+\frac{2 \kappa^2 M_0}{\lambda^2 m}.$$
	In particular, taking $\lambda = m^{-\frac{1}{2}}$, we have
	$$\mathbb{E}_q[\|f_q-s_m\|^2_{\mathcal{H}}]  \le C_\kappa m^{-r+\frac{1}{2}} \nu_X(( L_q^{-r} f_q)^2)+ 2 \kappa^2 M_0 $$
	where $C_\kappa=2\kappa^2+2$ only depends on $\kappa$.\\
	Taking $\lambda = m^{-\frac{1}{1+2r}}$, we have
	$$\mathbb{E}_q[\|f_q-s_m\|^2_{\mathcal{H}}]  \le C_1 m^{-\frac{2r-1}{2r+1}} (\nu_X(( L_q^{-r} f_q)^2)+M_0)$$
	where $C_1$ only depends on $\kappa$.
\end{corollary}

\begin{proof}
    Proposition \ref{second} shows that
    if $L_q^{-r} f_q \in L^2(q_X)$, then
    $$\|f_\lambda-f_q\|^2_{\mathcal{H}}\le \lambda^{2r-1} \nu_X(( L_q^{-r} f_q)^2),$$
    and
    $$\nu_X((f_\lambda-f_q)^2)\le \lambda^{2r} \nu_X(( L_q^{-r} f_q)^2).$$
    We note that
	$$\|s_m-f_q\|^2_{\mathcal{H}}\le 2(\|f_\lambda-s_m\|^2_{\mathcal{H}}+ \|f_\lambda-f_q\|^2_{\mathcal{H}}).$$
	So taking the expectation and using Corollary \ref{cor1}, we have
	\begin{align*}
	\mathbb{E}_q[\|f_q-s_m\|^2_{\mathcal{H}}] &\le \frac{2 \kappa^2}{\lambda^2 m} \nu_X((f_\lambda-f_q)^2) +\frac{2 \kappa^2 M_0}{\lambda^2 m}+2\|f_\lambda-f_q\|^2_{\mathcal{H}}\\
	&\le \left(\frac{2 \kappa^2}{\lambda^{2-2r} m}+2\lambda^{2r-1}\right) \nu_X(( L_q^{-r} f_q)^2)+\frac{2 \kappa^2 M_0}{\lambda^2 m}.
	\end{align*}
\end{proof}

Corollaries 20 and 21 show that $s_m$ computed through RLS approximates $f_q$ closely, measured by the expected $L^2$ norm under $q$ and by the expected distance in $\mathcal{H}$ respectively. Note that the latter metric is stronger than the former metric by (\ref{equ:metric}). The error bounds are related to the sample size $m$ and the regularization parameter $\lambda$. Corollary \ref{main_lemma} will be employed for CF and Corollary \ref{main_lemma2} for DRSK. We pinpoint that both corollaries are more refined and elaborate than the theory cited by \citet{CF1}. Accordingly, we will obtain better convergence results in this paper.

\section{Proofs of Theorems in Section \ref{sec:DRSK main}} \label{sec:analysis}

\subsection{Control Functionals} \label{sec:CF}
This section presents the properties of $\hat{\theta}_{CF}$. We first justify
%We have seen that $s_m$ can be decomposed into a constant $c$ and a function $\psi$ with zero mean under the distribution $\pi$ so finding its mean $\mu_X(s_m)$ is straightforward and equal $c$. The effectiveness of the CF estimator lies on the approximation quality of $s_m$ for $\bar f$. Section \ref{sec:RLS} presents our RLS analysis, which obtains $s_m$ in a different path from \citet{CF1} who does not consider the extra component $Y$ and uses a more simplified machinery.
the closed-form formula of $f_m(x,y)$ in Algorithm \ref{CFalgorithm} by applying the theory from Section \ref{sec:RLS}. Then we prove the theorems in Section \ref{sec:CF0}.

We check that the setting here accords with the conditions in Section \ref{sec:RLS}. Recall that the set of samples in Section \ref{sec:RLS} corresponds to $\{(x_j,z_j=f(x_j,y_j))\}_{j=1,\cdots,m}$ in $D_0$, and the $f_q(x)$ there corresponds to $\bar{f}(x)$ here under Assumption \ref{CSA}. Let the generic $\mathcal{H}$ in Section \ref{sec:RLS} be the $\mathcal{H}_+$ in Section \ref{sec:PS}. It follows from $\sup_{x\in \Omega} k_0(x,x)<\infty$ specified in Section \ref{sec:Stein} and $k_+(x,x') = 1 + k_0(x,x')$ that
$$\kappa:=\sup_{x\in \Omega} \sqrt{k_+(x,x)}<\infty.$$
Besides, $M_0$ in Assumption \ref{A0} is exactly $M_0$ in Assumption \ref{noise1} since by the definition, we have
$$M_0=\mathbb{E}_q [(z-f_q(x))^2]=\mathbb{E}_q [(f(x,y)-\bar{f}(x))^2]=\mathbb{E}_q [\epsilon(x,y)^2]<\infty.$$

\begin{lemma}\label{H+RLS1}
	Let
	$$\hat{z}=(f(x_1,y_1),\cdots, f(x_m,y_m))^T,$$
	$$K_+=(k_+(x_i,x_j))_{m\times m},$$
	$$\hat{k}_+(x)=(k_+(x_1,x),\cdots, k_+(x_m,x))^T.$$
	Then the RLS solution is given as  $s_m(x)= \beta^T \hat{k}_+(x)$  where  $\beta =(K_+ +\lambda m I)^{-1} \hat{z}$. Moreover, $\mu_X(s_m)=\beta^T \mathbf{1}$.
\end{lemma}

\begin{proof}
The first part of the expression of $s_m$ is a direct consequence of Lemma \ref{RLS1}.
By the definition of two reproducing kernel Hilbert spaces,
   $$\hat{k}_+(x)=\hat{k}_0(x)+\mathbf{1}$$
    so
     $$\mu_X(s_m(x))= \beta^T \mu_X(\hat{k}_0(x))+\beta^T \mathbf{1}.$$
     Note that $\mu_X(k_0(x_i,\cdot))=0$ for any given $x_i$. Therefore we conclude that
     $$\mu_X(s_m)=\beta^T \mathbf{1}.$$
 \end{proof}
 
For a given underlying function $f$, a more precise notation is to write $s^f_m$ as the solution to the RLS problem, but we will simply use $s_m$ for short if no confusion arises. We observe a fact that $s_m$ is a linear combination of $\hat{z}$ and thus a linear functional of $f$, that is, $s_m^{f_1}+s_m^{f_2}=s_m^{f_1+f_2}$ for any two functions $f_1, f_2$. We will leverage this fact several times later in the paper. Moreover, these linear coefficients only depend on the RKHS $\mathcal{H}_0$, free of the function of interest $f$. 

We explain some connections between $\pi$, $q$ and $\mathcal{H}_+$ when constructing the CF estimator in the ``Biased" case. Note that the theoretical results on RLS that we developed in Section \ref{sec:RLS} does not require any connection between $\mathcal{H}_+$ (which can be a general RKHS) and $q$ (the underlying distribution of the samples). In CF, we specify the choice of $\mathcal{H}_+$ which is constructed from the original distribution $\pi_X$. Meanwhile, $s_m$ is learned from the data drawn from $q$ (note that $s_m$ only depends on $\mathcal{H}_+$ and the data). Therefore the formula for $\mu_X(s_m)$ in Lemma \ref{H+RLS1} is also valid in the ``Biased" case.

% \subsection{X is sampled from $\pi$}
%To describe the error of $\hat\mu$, we first state the following observation of \citet{CF1} that translates the error of $s_m$ into the error of the two-phase estimator $\hat\mu$ .

We first establish the following lemma for CF when the extra component $Y$ appears.

\begin{lemma} \label{split}
Suppose $g\in L^2(\pi_X)$ and $\mathbb{E}_\pi[g(X,Y)]=\theta$. For any constant $a$, we have 
$$\mathbb{E}_\pi[(\bar{g}(X)-a)^2]=\mathbb{E}_\pi[(g(X,Y)-\theta)^2]+(\theta-a)^2-\mathbb{E}_\pi[\epsilon_g(X,Y)^2]$$
where $\bar{g}(X)=\mathbb{E}_\pi[g(X,Y)|X]$ and $\epsilon_g(X,Y)=g(X,Y)-\bar{g}(X)$. In particular,
$$\mathbb{E}_\pi[(\bar{g}(X)-a)^2]+\mathbb{E}_\pi[\epsilon_g(X,Y)^2]\ge \mathbb{E}_\pi[(g(X,Y)-\theta)^2].$$
\end{lemma}

\begin{proof} Note that, by definition, we have 
$$\mathbb{E}_\pi[\epsilon_g(X,Y)]=0, \quad \mathbb{E}_\pi[\epsilon_g(X,Y)\bar g(X)]=0.$$
Hence
\begin{align*}
&\mathbb{E}_\pi[(\bar{g}(X)-a)^2]\\
=&\mathbb{E}_\pi[(g(X,Y)-\theta)+(\theta-a-\epsilon_g(X,Y))^2]\\
=&\mathbb{E}_\pi[(g(X,Y)-\theta)^2]+\mathbb{E}_\pi[(\theta-a-\epsilon_g(X,Y))^2]
+2\mathbb{E}_\pi[(g(X,Y)-\theta)(\theta-a-\epsilon_g(X,Y))]\\
=&\mathbb{E}_\pi[(g(X,Y)-\theta)^2]+(\theta-a)^2+\mathbb{E}_\pi[\epsilon_g(X,Y)^2]
-2\mathbb{E}_\pi[\epsilon_g(X,Y)(g(X,Y)-\theta)]\\
%=&\mathbb{E}_\pi[(g(X,Y)-\theta)^2]+(\theta-a)^2+\mathbb{E}_\pi[\epsilon_g(X,Y)^2]\\
%&-2\mathbb{E}_\pi[\epsilon_g(X,Y)(\epsilon_g(X,Y)+\bar{g}(X)-\theta)]\\
=&\mathbb{E}_\pi[(g(X,Y)-\theta)^2]+(\theta-a)^2+\mathbb{E}_\pi[\epsilon_g(X,Y)^2]
-2\mathbb{E}_\pi[\epsilon_g(X,Y)^2]\\
=&\mathbb{E}_\pi[(g(X,Y)-\theta)^2]+(\theta-a)^2-\mathbb{E}_\pi[\epsilon_g(X,Y)^2].
\end{align*}
\end{proof}

\iffalse
\begin{proposition}\label{prop1}
	Assume
	$$\mathbb{E}_\pi[\mu((f -s_m)^2)] = I_1$$ Then the mean squared error of $\hat\mu$ is given by
	$$\mathbb{E}_\pi [(\hat{\theta}_{CF}-\theta)^2] = \mathbb{E}_\pi[\mathbb{E}_\pi [(\hat{\theta}_{CF}-\theta)^2|D_0]] \le \frac{I_1}{n-m}$$
\end{proposition}

\begin{proof}
	For $i=m+1,\cdots,n$, we have
	$\mathbb{E}_\pi[f_m(x_i,y_i)-\mu|D_0]=0$.
	By the independence of $D$,
	$$\mathbb{E}_\pi[(\hat{\theta}_{CF}-\theta)^2|D_0] = \left(\frac{1}{n-m}\right)^2 \sum_{i=m+1}^{n} \mathbb{E}_\pi [(f_m(x_i,y_i)-\mu)^2|D_0]$$
	It is well-known that $\mathbb{E}[(X-a)^2]$ is minimized when $a=\mathbb{E}(X)$. This implies that
	$$\mathbb{E}_\pi [(f_m(x_i,y_i)-\mu)^2|D_0]\le \mathbb{E}_\pi [(f_m(x_i,y_i)-\mu_X(s_m))^2|D_0]$$
	The right-hand side is exactly $\mathbb{E}_\pi[(f -s_m)^2|D_0]$. Therefore
	$$\mathbb{E}_\pi [(\hat{\theta}_{CF}-\theta)^2|D_0] \le \left(\frac{1}{n-m}\right)^2 \sum_{i=m+1}^{n} \mu((f -s_m)^2)= \frac{1}{n-m} \mu((f -s_m)^2)$$
\end{proof}
\fi

Considering CF estimator applied to the ``Partial" case introduced in Section \ref{sec:discussion}, we have the following result:

\begin{theorem}[CF in the ``Partial" case]\label{mainCF1}
Suppose Assumption \ref{noise1} holds and take an RLS estimate with $\lambda = m^{-\frac{1}{2}}$. Let $m=\alpha n$ where $0< \alpha< 1$. The CF estimator $\hat{\theta}_{CF}$ is an unbiased estimator of $\theta$ that satisfies the following bound.\\
(a) If $\bar{f} \in \text{Range}(L_\pi^{r})$ ($0\le r \le 1$), then $\mathbb{E}_\pi [(\hat{\theta}_{CF}-\theta)^2] \le C_1(C_{f} n^{-1-r}+M_0 n^{-1})$ where $C_{f}=\|L_\pi^{-r}\bar{f}\|^2_{L^2(\pi_X)}$ (which is a constant indicating the regularity of $\bar{f}$ in $\mathcal{H}_+$), $C_1$ only depends on $\alpha, \kappa$. In particular, $\mathbb{E}_\pi [(\mu_X(s_m)-\theta)^2] \le C_1(C_{f} m^{-r}+M_0)$.\\
(b) %If $\bar{f} \in \overline{\mathcal{H}_+}^\pi$, then for any given $\varepsilon>0$, $\mathbb{E}_\pi [(\hat{\theta}_{CF}-\theta)^2] \le C_1(C_\varepsilon n^{-2}+M_0 n^{-1}+\varepsilon n^{-1})$ where $C_\varepsilon=\|L_\pi^{-1}g\|_{L^2(\pi_X)}$ and $g\in \text{Range}(L_\pi)$ is an $\varepsilon$-approximation of $\bar{f}$ in $L^2(\pi_X)$, $C_1$ only depends on $\alpha, \kappa$. 
Suppose that there exists a $\varepsilon>0$ and $g \in \text{Range}(L_\pi)$, such that $\|\bar{f}-g\|^2_{L^2(\pi_X)}\le \varepsilon$. This assumption holds, for instance, if $\bar{f} \in \overline{\mathcal{H}_+}^\pi$ (the closure of $\mathcal{H}_+$ in the space $L^2(\pi_X)$).
Then $\mathbb{E}_\pi [(\hat{\theta}_{CF}-\theta)^2] \le C_1(\|L_\pi^{-1}g\|^2_{L^2(\pi_X)} n^{-2}+M_0 n^{-1}+\varepsilon n^{-1})$ where $C_1$ only depends on $\alpha, \kappa$.
%	Then the estimator $\hat{\mu}$ is an unbiased estimator of $\mu$ with
%$$\mathbb{E}_\pi [(\hat{\theta}_{CF}-\theta)^2] = O(C_\kappa(C_fn^{-2}+ M_0n^{-1}))$$
%where $C_f$ is a constant free of $m$ (and $n$), $C_\kappa=2\kappa^4+2$ and the outside $O$ only depends on the ratio $m/n$.
\end{theorem}

\begin{proof}
%Recall that Lemma \ref{isomorphism} shows that $\|L^{-\frac{1}{2}}h\|_{L^2(\pi_X)}=\|h\|_{\mathcal{H}_+}$ for all $h \in (\mathcal{H}_+)_1$.
(a) Suppose $\bar{f} \in \text{Range}(L_\pi^{r})$ ($0\le r \le 1$). We apply Corollary \ref{main_lemma} (with $r$) to the samples $\{(x_i, f(x_i,y_i))\}$ to obtain
$$\mathbb{E}_\pi[\mu_X((\bar{f} - s_m)^2)]\le C_\kappa m^{-r} \mu_X((L_\pi^{-r}\bar{f})^2)+ (C_\kappa-2) M_0\le C_\kappa C_{f}m^{-r}+(C_\kappa-2) M_0$$ where $C_\kappa=2\kappa^4+2$. Note that
\begin{align*}
(\mu_X(s_m)-\theta)^2=& \left|\int_{\Omega} (s_m(t)-\bar{f}(t))\pi_X(t)dt\right|^2\\
\le&\left(\int_{\Omega} |s_m(t)-\bar{f}(t)|\pi_X(t)dt\right)^2\\
\le &\int_{\Omega} |s_m(t)-\bar{f}(t)|^2\pi_X(t)dt   \quad \text{by Cauchy-Schwarz inequality}\\
= & \mu_X((\bar{f} - s_m)^2)
\end{align*}
Thus we obtain a bound for $\mu_X(s_m)$:
$$\mathbb{E}_\pi[(\mu_X(s_m)-\theta)^2]\le \mathbb{E}_\pi[\mu_X((\bar{f} - s_m)^2)]\le C_\kappa C_{f}m^{-r}+(C_\kappa-2) M_0.$$

In Lemma \ref{split}, we choose $g=f_m$ and $a=\mu_X(s_m)$. In this case, $\bar{g}=\bar{f}_m=\bar{f}-s_m+\mu_X(s_m)$ and $\epsilon_g=\epsilon$. So it follows from Lemma \ref{split} that
\begin{align*}
\mu((f_m-\theta)^2)&\le \mu_X((\bar{f}_m-\mu_X(s_m))^2)+\mathbb{E}_\pi[\epsilon(X,Y)^2]\\
& = \mu_X((\bar{f}-s_m)^2)+M_0
\end{align*}
and thus $\mathbb{E}_\pi[\mu((f_m-\theta)^2)] \le C_\kappa (C_{f}m^{-r}+ M_0)$. Furthermore, note that given $D_0$, $f_m(x_j,y_j)$ is an unbiased estimator of $\theta$ so
\begin{align*}
\mathbb{E}_{\pi}\left[\left(\frac{1}{n-m}\sum_{j=m+1}^{n}f_m(x_j,y_j)-\theta\right)^2\Bigg|D_0\right]&=\frac{1}{(n-m)^2}\sum_{j=m+1}^{n}\mathbb{E}_{\pi}\left[(f_m(x_j,y_j)-\theta)^2|D_0\right]\\
 &=\frac{1}{n-m}\mathbb{E}_\pi[\mu((f_m-\theta)^2)]
\end{align*}
Therefore,
$$\mathbb{E}_{\pi}\left[(\hat{\theta}_{CF}-\theta)^2\right]\le\frac{C_\kappa (C_{f}m^{-r}+ M_0)}{n-m},$$
which implies that
$$\mathbb{E}_\pi [(\hat{\theta}_{CF}-\theta)^2]\le C_1(C_{f} n^{-1-r}+M_0 n^{-1})$$
since $m=\alpha n$.

(b) We first note that Proposition \ref{ease} shows that $\overline{\mathcal{H}_+}^\pi = \overline{\text{Range}(L_\pi)}^\pi$. Therefore, the assumption holds if $\bar{f} \in \overline{\mathcal{H}_+}^\pi$. Note that $\bar{f} \in \overline{\mathcal{H}_+}^\pi$ is a mild assumption (which is weaker than the assumptions in \citet{CF1} and \citet{LL}) and in some cases, $\overline{\mathcal{H}_+}^\pi=L^2(\pi_X)$ (Lemma 4 in \citet{CF2}). This result shows that we can establish similar results even with a very weak assumption.

%Consider the following approximation approach: Fix $\epsilon>0$. There exists a function $g\in\text{Range}(L_\pi)$ such that $\|\bar{f}-g\|^2_{L^2(\pi_X)}\le \epsilon$. 
Let $h=f-g$ so $\bar{h}=\bar{f}-g$.  Let $s^{h}_m$, $s^{g}_m$ be the RLS functional approximation of $h$, $g$ respectively. As we point out after Lemma \ref{H+RLS1}, $s^{h}_m$ is a linear functional of $h$, so we have
$$h - s^{h}_m=(f - s_m)-(g - s^{g}_m).$$
Next we apply Corollary \ref{main_lemma} (with $r=1$) to the samples $\{(x_i, g(x_i))\}$: since $g\in\text{Range}(L_\pi)$ and $g$ is a function of $x$ only, $M^g_0=0$ and
$$\mathbb{E}_\pi[\mu_X((g - s^{g}_m)^2)] \le C_\kappa m^{-1} \mu_X((L_\pi^{-1}g)^2).$$
Again, we apply Corollary \ref{main_lemma} (with $r=0$) to the samples $\{(x_i, h(x_i,y_i))\}$: we note that $\bar{h}\in L^2(\pi_X)=\text{Range}(L_\pi^{0})$ and $\bar{g}=g$ so 
$$M^h_0:=\mathbb{E}_\pi [(h(x,y)-\bar{h}(x))^2]=\mathbb{E}_\pi [(f(x,y)-\bar{f}(x))^2]$$
which is the same as $M_0$ and thus
$$\mathbb{E}_\pi[\mu_X((\bar{h} - s^{h}_m)^2)]\le C_\kappa \mu_X(\bar{h}^2)+ (C_\kappa-2) M_0\le C_\kappa\varepsilon +(C_\kappa-2) M_0.$$
Combining the above inequalities, we obtain by 
Cauchy–Schwarz inequality that
\begin{align*}
\mathbb{E}_\pi[\mu_X((\bar{f} - s_m)^2)] &\le 2(\mathbb{E}_\pi[\mu_X((g - s^{g}_m)^2)]+\mathbb{E}_\pi[\mu_X((\bar{h} - s^{h}_m)^2)]) \\
&\le 2 C_\kappa (\varepsilon + M_0+ m^{-1}\mu_X((L_\pi^{-1}g)^2))\\
&= 2 C_\kappa (\varepsilon + M_0+ m^{-1} \|L_\pi^{-1}g\|^2_{L^2(\pi_X)}).
\end{align*}
The rest is similar to part (a). We finally conclude
$$\mathbb{E}_\pi [(\hat{\theta}_{CF}-\theta)^2] \le C_1(\|L_\pi^{-1}g\|^2_{L^2(\pi_X)} n^{-2}+M_0 n^{-1}+\varepsilon n^{-1}).$$
\end{proof}

%Consider a case where $M_0$ is a relatively small number compared with $C$ (which contains information about $\bar{f}$ and is untraceable, possibly large in practice). Then the bound in Theorem \ref{mainCF1} part (a) essentially becomes
%$$\mathbb{E}_\pi[(\hat{\theta}_{CF}-\theta)^2] =O(n^{-1-r}).$$ 
%So the CF estimator achieves a supercanonical rate especially when a limited number of samples are available to us (which is usually true in practice). It is reasonable to say that even if we know very little about $Y$, the CF method applied on $X$ only still improves the Monte Carlo rate. Theorem \ref{mainCF1} part (b) provides a reasonable explanation that the CF estimator still performs well even if a weak assumption is imposed.

%One could diminish the effect of $M_0$ in the RLS regression by selecting another suitable $\lambda$ as illustrated in Corollary \ref{main_lemma} but this will offer a worse rate. Note that we will always include $M_0$ in the final result because of the error term $\epsilon$. Hence $\lambda = m^{-\frac{1}{2}}$ is the best choice of $\lambda$ leading to the best rate.

Theorem \ref{mainCF0-1} in Section \ref{sec:DRSK main} is part (a) of Theorem \ref{mainCF1}. In addition, we remark that $\lambda=\Theta(m^{-\frac{1}{2}})$ is the best choice of $\lambda$ leading to the best rate in theory due to the following reasons: %although $\lambda=0$ sometimes give rise to satisfactory performance in practice.
\begin{enumerate}
\item On the one hand, if there exists the noise $Y$, one might wish to diminish the effect of $M_0$ in the RLS regression by selecting another $\lambda$ such as $\lambda=m^{-\frac{1}{2+2r}}$ in Corollary \ref{main_lemma}. However, doing so will offer a worse rate than $O(m^{-r})$ for the term $\mathbb{E}_\pi[\mu_X((\bar{f} - s_m)^2)]$ and at the same time, the MSE bound of CF still contains the term $M_0n^{-1}$ because the effect of the error $\epsilon$ cannot be eliminated in the MSE. So $\lambda=\Theta(m^{-\frac{1}{2}})$ is preferred. 

\item On the other hand, if there does not exist the noise $Y$ (implying $M_0=0$), %there is no need to choose $\lambda$ rather than $\lambda = m^{-\frac{1}{2}}$ since it provides the best bound in 
Corollary \ref{main_lemma} with $M_0=0$ shows that the best upper bound is achieved by setting $\lambda=\Theta(m^{-\frac{1}{2}})$.
\end{enumerate}

%in Theorem \ref{mainCF0-2}

As a special case, we immediately get the following result when considering the CF estimator applied to the ``Standard" case:

\begin{theorem}[CF in the ``Standard" case]\label{mainCF2}
Take an RLS estimate with $\lambda = m^{-\frac{1}{2}}$. Let $m=\alpha n$ where $0< \alpha< 1$. The CF estimator $\hat{\theta}_{CF}$ is an unbiased estimator of $\theta$ that satisfies the following bound.\\
(a) If $f \in \text{Range}(L_\pi^{r})$ ($0\le r \le 1$), then $\mathbb{E}_\pi [(\hat{\theta}_{CF}-\theta)^2] \le C_1C_{f} n^{-1-r}$ where $C_{f}=\|L_\pi^{-r}f\|^2_{L^2(\pi_X)}$ (which is a constant indicating the regularity of $f$ in $\mathcal{H}_+$), $C_1$ only depends on $\alpha, \kappa$. In particular, $\mathbb{E}_\pi [(\mu_X(s_m)-\theta)^2] \le C_1C_{f} m^{-r}$.\\
(b) Suppose that there exists a $\varepsilon>0$ and $g \in \text{Range}(L_\pi)$, such that $\|f-g\|^2_{L^2(\pi_X)}\le \varepsilon$. This assumption holds, for instance, if $f \in \overline{\mathcal{H}_+}^\pi$ (the closure of $\mathcal{H}_+$ in the space $L^2(\pi_X)$).
Then $\mathbb{E}_\pi [(\hat{\theta}_{CF}-\theta)^2] \le C_1(\|L_\pi^{-1}g\|^2_{L^2(\pi_X)} n^{-2} +\varepsilon n^{-1})$ where $C_1$ only depends on $\alpha, \kappa$.
\end{theorem}

Theorem \ref{mainCF0-2} in Section \ref{sec:DRSK main} is part (a) of Theorem \ref{mainCF2}.

\begin{theorem}[CF in the ``Biased" case]\label{mainCF3}
Suppose Assumption \ref{lrupper2} holds and take an RLS estimate with $\lambda = m^{-\frac{1}{2}}$. Then the CF estimator $\hat{\theta}_{CF}$ satisfies the following bound.\\
(a) If $f \in \text{Range}(L_q^{r})$ ($0\le r \le 1$), then $\mathbb{E}_q [(\hat{\theta}_{CF}-\theta)^2] \le C_1C_{f} m^{-r}$ where $C_{f}=\|L_q^{-r}f\|^2_{L^2(q_X)}$ (which is a constant indicating the regularity of $f$ in $\mathcal{H}_+$), $C_1$ only depends on $\kappa$, $\mathbb{E}_{x\sim \pi_X}[\frac{\pi_X(x)}{q_X(x)}]$. In particular, $\mathbb{E}_q [(\mu_X(s_m)-\theta)^2] \le C_1C_{f} m^{-r}$.\\
(b) %If $f \in \overline{\mathcal{H}_+}^q$, then for any given $\varepsilon>0$, $\mathbb{E}_q [(\hat{\theta}_{CF}-\theta)^2] \le C_1(C_\varepsilon m^{-1}+\varepsilon)$ where $C_\varepsilon=\|L_q^{-1}g\|_{L^2(q_X)}$ and $g\in \text{Range}(L_q)$ is an $\varepsilon$-approximation of $f$ in $L^2(q_X)$, $C_1$ only depends on $\kappa$, $\mathbb{E}_{x\sim \pi_X}[\frac{\pi_X(x)}{q_X(x)}]$.
Suppose that there exists a $\varepsilon>0$ and $g \in \text{Range}(L_q)$, such that $\|f-g\|^2_{L^2(q_X)}\le \varepsilon$. This assumption holds, for instance, if $f \in \overline{\mathcal{H}_+}^q$ (the closure of $\mathcal{H}_+$ in the space $L^2(q_X)$).
Then $\mathbb{E}_q [(\hat{\theta}_{CF}-\theta)^2] \le C_1(\|L_q^{-1}g\|^2_{L^2(q_X)} m^{-1} +\varepsilon)$ where $C_1$ only depends on $\kappa$, $\mathbb{E}_{x\sim \pi_X}[\frac{\pi_X(x)}{q_X(x)}]$.
\end{theorem}

\begin{proof}
Note that in this theorem, only $m$ rather than $n$ appears in the upper bound. In the ``Biased" case, the second subset of data $D_1$ may have no contribution to the MSE bound of CF because of the bias.

(a) Note that
\begin{align*}
&(\mu_X(s_m)-\theta)^2\\
=& \left|\int_{\Omega} (s_m(t)-f(t))\pi_X(t)dt\right|^2\\
\le&\left(\int_{\Omega} |s_m(t)-f(t)|\pi_X(t)dt\right)^2\\
\le & \left(\int_{\Omega} |s_m(t)-f(t)|^2 q_X(t)dt\right)\left( \int_{\Omega} \frac{(\pi_X(t))^2}{q_X(t)} dt\right) \quad \text{by Cauchy-Schwarz inequality}\\
= & \nu_X((f - s_m)^2) \mathbb{E}_{x\sim \pi_X}\left[\frac{\pi_X(x)}{q_X(x)}\right].
\end{align*}
%It follows from Cauchy-Schwarz inequality that
%\begin{align*}
%\left(\int_{\Omega} |s_m(t)-f(t)|\pi_X(t)dt\right)^2 &\le \left(\int_{\Omega} |s_m(t)-f(t)|^2 q_X(t)dt\right)\left( \int_{\Omega} \frac{(\pi_X(t))^2}{q_X(t)} dt\right) \\
%&=\nu_X((f - s_m)^2) \mathbb{E}_{x\sim \pi_X}[\frac{\pi_X(x)}{q_X(x)}].
%\end{align*}
Therefore we obtain
$$\mathbb{E}_q [(\mu_X(s_m)-\theta)^2]\le  \mathbb{E}_q[\nu_X((f - s_m)^2)] \mathbb{E}_{x\sim \pi_X}\left[\frac{\pi_X(x)}{q_X(x)}\right].$$
Moreover, since we have $f_q=f$ and $M_0=0$ in the ``Biased" case.
%$$f\in (\mathcal{H}_+)_1^q\subset \text{Range}(L^{\frac{1}{2}}_q)$$
It follows from Corollary \ref{main_lemma} that
\begin{equation} \label{mainCF3_1}
\mathbb{E}_q[\nu_X((f - s_m)^2)] \le C_\kappa m^{-r} \nu_X((L_q^{-r}f)^2)
\end{equation}%=C_\kappa m^{-r} \|f\|^2_{\mathcal{H}_+}$$
where $C_\kappa=2\kappa^4+2$. Combining the above inequalities, we obtain by 
Cauchy–Schwarz inequality that
\begin{align*}
\mathbb{E}_q[\nu_X((f_m - \theta)^2)] &\le 2\left(\mathbb{E}_q[\nu_X((f - s_m)^2)]+\mathbb{E}_q [(\mu_X(s_m)-\theta)^2]\right) \\
&\le 2\left(\mathbb{E}_{x\sim \pi_X}\left[\frac{\pi_X(x)}{q_X(x)}\right]+1\right)  \mathbb{E}_q[\nu_X((f - s_m)^2)]\\
&\le 2\left(\mathbb{E}_{x\sim \pi_X}\left[\frac{\pi_X(x)}{q_X(x)}\right]+1\right) C_\kappa m^{-r} \nu_X((L_q^{-r}f)^2) \\
&\le C_1 C_{f} m^{-r}.
\end{align*}
Nevertheless, given $D_0$, $f_m(x_j)$ is not necessarily an unbiased estimator of $\theta$ so we can only assert that
%\begin{align*}
%&\nu_X\left(\left(\sum_{j=m+1}^{n}\frac{1}{n-m}f_m(x_j)-\theta\right)^2\right)\\
%=&\frac{\nu_X\left((f_m(x_{m+1})-\theta)^2\right)}{n-m} + \frac{(n-m-1)\nu_X\left((f_m(x_{m+1})-\theta)(f_m(x_{m+2})-\theta)\right)}{n-m}\\
%\le& \nu_X\left((f_m(x_{m+1})-\theta)^2\right).
%\end{align*}
$$\mathbb{E}_q\left[(\hat{\theta}_{CF}-\theta)^2\right]=\mathbb{E}_q\left[\left(\sum_{j=m+1}^{n}\frac{1}{n-m}f_m(x_j)-\theta\right)^2\right]\le \mathbb{E}_q\left[\nu\left((f_m - \theta)^2\right)\right]$$
by Cauchy–Schwarz inequality. (Note that the cross terms may not vanish.)
Therefore,
$$\mathbb{E}_q\left[(\hat{\theta}_{CF}-\theta)^2\right]\le C_1 C_{f} m^{-r}.$$

(b) We only need to replace inequality (\ref{mainCF3_1}). The rest of the proof is the same as part (a).

We note that Proposition \ref{ease} shows that $\overline{\mathcal{H}_+}^q = \overline{\text{Range}(L_q)}^q$. Therefore, the assumption holds if $f \in \overline{\mathcal{H}_+}^q$.
%We note that $f \in \overline{\mathcal{H}_+}^q = \overline{\text{Range}(L_q)}^q$ by Proposition \ref{ease}. Consider the following approximation approach. For the given $\varepsilon>0$, there exists a $g\in\text{Range}(L_q)$ such that $\|f-g\|^2_{L^2(q_X)}\le \epsilon$.
Let $h=f-g$.  Let $s^{h}_m$, $s^{g}_m$ be the RLS functional approximation of $h$, $g$ respectively. $s^{h}_m$ is a linear functional of $h$, so we write
$$h - s^{h}_m=(f - s_m)-(g - s^{g}_m)$$
Next we apply Corollary \ref{main_lemma} (with $r=1$) to the samples $\{(x_i, g(x_i))\}$: Since $g\in\text{Range}(L_q)$ and $g$ is a function of $x$ only, then $M^g_0=0$ and
$$\mathbb{E}_q[\nu_X((g - s^{g}_m)^2)] \le C_\kappa m^{-1} \nu_X((L_q^{-1}g)^2)$$
Again, we apply Corollary \ref{main_lemma} (with $r=0$) to the samples $\{(x_i, h(x_i))\}$: We note that $h\in L^2(q_X)$ and $h$ is a function of $x$ only so 
$M^h_0=0$ and thus
$$\mathbb{E}_q[\nu_X((h - s^{h}_m)^2)]\le C_\kappa \nu_X(h^2)\le C_\kappa \varepsilon$$
where $C_\kappa=2\kappa^4+2$.\\
Adding these two parts we obtain
$$\mathbb{E}_q[\nu_X((f - s_m)^2)] \le 2 C_\kappa(\varepsilon+m^{-1}\nu_X((L_q^{-1}g)^2)).$$
Finally, we conclude that
$$\mathbb{E}_q\left[(\hat{\theta}_{CF}-\theta)^2\right]\le C_1(\|L_q^{-1}g\|^2_{L^2(q_X)} m^{-1} +\varepsilon).$$
\end{proof}

Theorem \ref{mainCF0-3} in Section \ref{sec:DRSK main} is part (a) of Theorem \ref{mainCF3}.
%We observe that since $f_m(x_j,y_j)$ is not necessarily an unbiased estimator of $\theta$, the dataset $D_1$ has not any contributions to improving the MSE. The bound in Theorem $\ref{mainCF3}$ has nothing to do with $n-m$. So it is reasonable to take $m=n-1$ and use the entire data to construct $s_m$ for maximizing the effect of the CF estimator.

\begin{theorem}[CF in the ``Both" case]\label{mainCF4}
Suppose Assumptions \ref{CSA}, \ref{noise1}, and \ref{lrupper2} hold. Take an RLS estimate with $\lambda = m^{-\frac{1}{2}}$. Let $m=\alpha n$ where $0< \alpha< 1$. Then the CF estimator $\hat{\theta}_{CF}$ satisfies the following bound.\\
(a) If $\bar{f} \in \text{Range}(L_q^{r})$ ($0\le r \le 1$), then $\mathbb{E}_q [(\hat{\theta}_{CF}-\theta)^2] \le C_1(C_{f} n^{-r}+M_0)$ where $C_{f}=\|L_q^{-r}\bar{f}\|^2_{L^2(q_X)}$ (which is a constant indicating the regularity of $\bar{f}$ in $\mathcal{H}_+$), $C_1$ only depends on $\alpha$, $\kappa$, $\mathbb{E}_{x\sim \pi_X}[\frac{\pi_X(x)}{q_X(x)}]$. In particular, $\mathbb{E}_q [(\mu_X(s_m)-\theta)^2] \le C_1(C_{f} m^{-r}+M_0)$.\\
(b) Suppose that there exists a $\varepsilon>0$ and $g \in \text{Range}(L_q)$, such that $\|\bar{f}-g\|^2_{L^2(q_X)}\le \varepsilon$. This assumption holds, for instance, if $\bar{f} \in \overline{\mathcal{H}_+}^q$ (the closure of $\mathcal{H}_+$ in the space $L^2(q_X)$).
Then $\mathbb{E}_q [(\hat{\theta}_{CF}-\theta)^2] \le C_1(\|L_q^{-1}g\|^2_{L^2(q_X)} n^{-1} +M_0 +\varepsilon)$ where $C_1$ only depends on $\alpha$, $\kappa$, $\mathbb{E}_{x\sim \pi_X}[\frac{\pi_X(x)}{q_X(x)}]$. 
\end{theorem}

\begin{proof}
It follows from combining the proofs of Theorem \ref{mainCF1} and \ref{mainCF3}. In fact,
$$\hat{\theta}_{CF}-\theta=\frac{1}{n-m}\sum_{j=m+1}^{n}(\bar{f}(x_j)-s_m(x_j)+\mu_X(s_m)-\theta+\epsilon(x_j,y_j))$$
so we only need to analyze the following two terms separately:
$$\mathbb{E}_q\left[\left(\frac{1}{n-m}\sum_{j=m+1}^{n}(\bar{f}_m(x_j)-\theta)\right)^2\right], \ \mathbb{E}_q\left[\left(\frac{1}{n-m}\sum_{j=m+1}^{n}\epsilon(x_j,y_j)\right)^2\right].$$
The first term can be studied as in Theorem \ref{mainCF3}:
\begin{align*}
&\mathbb{E}_q\left[\left(\frac{1}{n-m}\sum_{j=m+1}^{n}(\bar{f}_m(x_j)-\theta)\right)^2\right]\\
\le &\mathbb{E}_q[\nu_X((\bar{f}_m - \theta)^2)]\\
\le & 2\left(\mathbb{E}_{x\sim \pi_X}\left[\frac{\pi_X(x)}{q_X(x)}\right]+1\right)  \mathbb{E}_q[\nu_X((\bar{f} - s_m)^2)].    
\end{align*}
For the second term,
we note that with Assumption \ref{CSA}, we have
$\mathbb{E}_q[\epsilon(x_j,y_j)]=0$ and thus
$$\mathbb{E}_q\left[\left(\frac{1}{n-m}\sum_{j=m+1}^{n}\epsilon(x_j,y_j)\right)^2\right]= \frac{1}{n-m} \mathbb{E}_q\left[\epsilon(x_{m+1},y_{m+1})^2\right]\le \frac{M_0}{n-m}$$
by Assumption \ref{noise1}.
\end{proof}

Theorem \ref{mainCF0} in Section \ref{sec:DRSK main} is part (a) of Theorem \ref{mainCF4}. Although $M_0$ appears non-vanishing in Theorem \ref{mainCF4}, it is merely a matter of the bound that we use on the term $\mathbb{E}_q[\nu_X((\bar{f} - s_m)^2)]$. There are several ways to overcome this. The first approach is to choose $\lambda = m^{-\frac{1}{2+2r}}$ (instead of $\lambda=m^{-\frac{1}{2}}$) to obtain, by Corollary \ref{main_lemma},
$$\mathbb{E}_q[\nu_X((f_q - s_m)^2)] \le C_1 m^{-\frac{r}{1+r}} (\nu_X(( L_q^{-r} f_q)^2)+M_0)$$
with a vanishing rate $m^{-\frac{r}{1+r}}$ for the $M_0$ term, but at the cost of a worse rate $m^{-\frac{r}{1+r}}$ than $m^{-r}$ for the $C_{f}$ term. The second approach is to leverage refined error bounds in \citet{SW,SW2}, e.g., by keeping $\lambda=m^{-\frac{1}{2}}$, we can obtain
$$\mathbb{E}_q[\nu_X((f_q - s_m)^2)] \le C_1 (m^{-r} \nu_X(( L_q^{-r} f_q)^2)+m^{-\frac{1}{2}}M_0).$$
However, none of these approaches can provide a bound that is better than the bound $o(n^{-\frac{1}{2}-r}) + M_0n^{-1}$ in our DRSK.

\subsection{Black-box Importance Sampling}\label{sec:IS}

In this section, we prove the theorems in Section \ref{sec:BBIS1}. 

%\begin{theorem}[``Standard" \& ``Biased"] \label{mainIS1}
%Suppose $f\in \mathcal{H}_+$. The BBIS estimator $\hat{\theta}_{IS}$ satisfies the following bound.\\
%(a) Assume Assumption \ref{lrupper3} holds. Then $\mathbb{E}_q[(\hat{\theta}_{IS}-\theta)^2] =O(n^{-1})$.\\
%(b) Assume Assumption \ref{lrupper}, \ref{lrupper4} holds. Then $\mathbb{E}_q[(\hat{\theta}_{IS}-\theta)^2] =o(n^{-1})$.\\
%\end{theorem}

\begin{proof} [Proof of Theorem \ref{mainIS0-1} in Section \ref{sec:DRSK main}]
See \citet{LL}.
\end{proof}

To prove Theorem \ref{mainIS0}, we split our proof into three lemmas. First we demonstrate that with the upper bound on each weight, we can control the noise term. Then since the optimization problem (\ref{COP0}) we consider here is different from \citet{LL}, Theorem 3.2 and 3.3 in \citet{LL} cannot be applied straightforwardly to Theorem \ref{mainIS0} which thus needs to be redeveloped. This is done in Lemma \ref{varifyweight1} and \ref{varifyweight2} which consider Part (a) and Part (b) of Theorem \ref{mainIS0} respectively.

\begin{lemma} \label{noiselemma}
Suppose Assumptions \ref{CSA} and \ref{noise1} hold. Then 
$$\mathbb{E}_q\left[\left(\sum_{j=1}^{n} \hat{w}_j \epsilon(x_j,y_j)\right)^2\right] \le \frac{M_0 B_0^2}{n}.$$
\end{lemma}

\begin{proof}
The covariate shift assumption implies that
$\mathbb{E}_q[\epsilon(x_j,y_j)|x_j]=\mathbb{E}_\pi[\epsilon(x_j,y_j)|x_j]=0$. Since $\hat{w}_j$ is a function of the $X$ factor $\textbf{x}=(x_{1},\cdots,x_n)$, it follows that $\mathbb{E}_q[\hat{w}_j\epsilon(x_j,y_j)|\textbf{x}]=0$ and conditional on $\textbf{x}$, $\hat{w}_j\epsilon(x_j,y_j)$ is conditionally independent of each other.
So we assert that
$$\mathbb{E}_q\left[\left(\sum_{j=1}^{n} \hat{w}_j \epsilon(x_j,y_j)\right)^2\bigg|\textbf{x}\right]=\sum_{j=1}^{n}\mathbb{E}_q\left[\left( \hat{w}_j \epsilon(x_j,y_j)\right)^2\bigg|\textbf{x}\right],$$
and thus %\le \frac{M_0 B_0^2}{n}
$$\mathbb{E}_q\left[\left(\sum_{j=1}^{n} \hat{w}_j \epsilon(x_j,y_j)\right)^2\right]= \sum_{j=1}^{n}\mathbb{E}_q \left[\mathbb{E}_q\left[\left( \hat{w}_j \epsilon(x_j,y_j)\right)^2\bigg|\textbf{x}\right] \right].$$
%\frac{M_0 B_0^2}{n}
The upper bound on $\hat{w}_j$ (in the BBIS construction) and $\epsilon(x_j,y_j)$ (in Assumption \ref{noise1}) implies that \begin{align*}
\mathbb{E}_q\left[\mathbb{E}_q\left[\hat{w}_j^2 \epsilon(x_j,y_j)^2|\textbf{x}\right]\right]&= \mathbb{E}_q\left[\hat{w}_j^2 
\mathbb{E}_q\left[\epsilon(x_j,y_j)^2|\textbf{x}\right]\right]\\
&\le \mathbb{E}_q\left[\frac{B_0^2}{n^2} 
\mathbb{E}_q\left[\epsilon(x_j,y_j)^2|\textbf{x}\right]\right]\\
&= \frac{B_0^2}{n^2}\mathbb{E}_q\left[ 
\mathbb{E}_q\left[\epsilon(x_j,y_j)^2|\textbf{x}\right]\right]\\
&= \frac{B_0^2}{n^2} 
\mathbb{E}_q\left[\epsilon(x_j,y_j)^2\right] \le \frac{M_0 B_0^2}{n^2}.
\end{align*}
Hence we obtain that
$$\mathbb{E}_q\left[\left(\sum_{j=1}^{n} \hat{w}_j \epsilon(x_j,y_j)\right)^2\right]\le \frac{M_0 B_0^2}{n}.$$

%we also have the following two facts.
%Since, uniform $L^1(q_X)$-boundedness implies uniform $L^1(q_X)$-tightness, then $$\left(\sum_{j=1}^{n} \hat{w}_j \epsilon(x_j,y_j)\right)^2=O_q(n^{-1}).$$
%by Markov's inequality, we have for $\delta>0$,
%$$\left(\sum_{j=1}^{n} \hat{w}_j \epsilon(x_j,y_j)\right)^2\le \frac{4 M_0 B^2}{n \delta}.$$
%with probability at least $1-\delta$.
\end{proof}

\begin{lemma} \label{varifyweight1}
Suppose Assumption \ref{lrupper} holds. Take $B_0=2B$ in (\ref{COP0}). We have $$\mathbb{E}_q[\mathbb{S}(\{\hat{w}_j,x_j\},\pi_X)]= O(n^{-1}).$$
\end{lemma}

\begin{proof}
%First we note that by Markov'e inequality, $$\frac{\pi_X(x)}{q_X(x)}\le \frac{1}{\delta}$$
%with probability at least $1-\delta$ since $\mathbb{E}_q[\frac{\pi_X(x)}{q_X(x)}]=1$.
%We only consider $x$ in such a subset. 
We define the weights as constructed in \citet{LL}:
$$w_j^*=\frac{1}{Z} \frac{\pi_X(x_j)}{q_X(x_j)}, \quad Z=\sum_{j=1}^{n} \frac{\pi_X(x_j)}{q_X(x_j)}.$$
%Notice that since $\mathbb{E}_q[\frac{\pi_X(x_j)}{q_X(x_j)}-1]=0$ and they are independent, then
%$$\mathbb{E}_q[(\frac{1}{n}Z-1)^2]=\frac{1}{n}\mathbb{E}_q[(\frac{\pi_X(x_1)}{q_X(x_1)}-1)^2]=\frac{1}{n}\mathbb{E}_q[(\frac{\pi_X(x_1)}{q_X(x_1)})^2-1]:=\frac{1}{n}M_2$$
Notice that $\frac{1}{n}\frac{\pi_X(x_j)}{q_X(x_j)} \in [0, \frac{B}{n}]$.
By Hoeffding's inequality \citep{Hoeffding}, we have
$$\mathbb{P}\left(\frac{1}{n}Z-1 \le -\frac{1}{2}\right)\le \exp(-\frac{n}{2B^2})$$
where the probability is with respect to the sampling distribution $q_X$ (the same below). Let $$\mathcal{E} :=\left\{\frac{1}{n}Z-1 \ge -\frac{1}{2}\right\}.$$ 
The above statement demonstrates that
$\mathbb{P}(\mathcal{E}^c)\le \exp(-\frac{n}{2B^2}).$ Furthermore,
note that given $\mathcal{E}$, we have
$$w^*_j\le \frac{B}{n/2}=\frac{2B}{n}.$$
So
$$\mathcal{E}\subset \left\{0\le w^*_i \le \frac{2B}{n}, \ \forall i= 1,\cdots,n\right\}.$$
This demonstrates that given $\mathcal{E}$, $w^*_j$ is a feasible solution to the problem (\ref{COP0}) and thus
$$\mathbb{S}(\{\hat{w}_j,x_j\},\pi_X)\le \mathbb{S}(\{w^*_j,x_j\},\pi_X).$$ 
Moreover, we observe that since $0\le \hat{w}_j\le \frac{2B}{n}$,
\begin{align*}
0&\le \mathbb{S}(\{\hat{w}_j,x_j\},\pi_X)=\sum_{j,k=1}^{n} \hat{w}_j\hat{w}_k k_0(x_j,x_k)\\
&\le \sum_{j,k=1}^{n} |\hat{w}_j\hat{w}_k k_0(x_j,x_k)|\le \sum_{j,k=1}^{n}(\frac{2B}{n})^2\kappa_0^2=4B^2\kappa_0^2.    
\end{align*}
where $\kappa_0:=\sup_{x\in \Omega} \sqrt{k_0(x,x)}<\infty$ by our construction of $\mathcal{H}_0$.
Therefore we can express $\mathbb{E}_q[\mathbb{S}(\{\hat{w}_j,x_j\},\pi_X)]$ as
\begin{align*}
\mathbb{E}_q[\mathbb{S}(\{\hat{w}_j,x_j\},\pi_X)] &= \mathbb{E}_q[\mathbb{S}(\{\hat{w}_j,x_j\},\pi_X)|\mathcal{E}]\cdot \mathbb{P}[\mathcal{E}] + \mathbb{E}_q[\mathbb{S}(\{\hat{w}_j,x_j\},\pi_X)|\mathcal{E}^c]\cdot \mathbb{P}[\mathcal{E}^c]\\
&\le \mathbb{E}_q[\mathbb{S}(\{w^*_j,x_j\},\pi_X)|\mathcal{E}]\cdot \mathbb{P}[\mathcal{E}] + \mathbb{E}_q[\mathbb{S}(\{\hat{w}_j,x_j\},\pi_X)|\mathcal{E}^c]\cdot \exp(-\frac{n}{2B^2})\\
&\le \mathbb{E}_q[\mathbb{S}(\{w^*_j,x_j\},\pi_X)] + 4B^2\kappa_0^2\cdot \exp(-\frac{n}{2B^2}).
\end{align*}
It follows from Theorem B.2 in \citet{LL} that 
$$\mathbb{E}_q[\mathbb{S}(\{w^*_j,x_j\},\pi_X)]= O(n^{-1}).$$
Obviously we also have
$$4B^2\kappa_0^2\cdot \exp(-\frac{n}{2B^2})= O(n^{-1}).$$
Hence, we finally obtain
$$\mathbb{E}_q[\mathbb{S}(\{\hat{w}_j,x_j\},\pi_X)]= O(n^{-1}).$$
\end{proof}

\begin{lemma} \label{varifyweight2}
Suppose Assumptions \ref{lrupper} and \ref{lrupper4}  hold. Take $B_0=4B$ in (\ref{COP0}). We have 
$$\mathbb{E}_q[\mathbb{S}(\{\hat{w}_j,x_j\},\pi_X)]= o(n^{-1}).$$
\end{lemma}

\begin{proof}
Without loss of generality, assume $n$ is an even number, and we partition the index of the data set $D_1$ into two parts $D_2=\{1, \cdots,\frac{n}{2}\}$ and $D_3=\{\frac{n}{2}+1,\cdots, n\}$. We define the weights as constructed in \citet{LL}:
$$w_i^* = \left\{ \begin{array}{rl} 
\frac{1}{n}\frac{\pi_X(x_i)}{q_X(x_i)}-\frac{2}{n^2}\sum_{j\in D_3} \frac{\pi_X(x_i)}{q_X(x_i)}\frac{\pi_X(x_j)}{q_X(x_j)} k_L(x_j,x_i) & \forall i\in D_2,\\
\frac{1}{n}\frac{\pi_X(x_i)}{q_X(x_i)}-\frac{2}{n^2}\sum_{j\in D_2} \frac{\pi_X(x_i)}{q_X(x_i)}\frac{\pi_X(x_j)}{q_X(x_j)} k_L(x_j,x_i) & \forall i\in D_3, \end{array} \right. $$
where $k_L(x,x')=\sum_{l=1}^{L} \phi_l(x) \phi_l(x')$ and $L=n^{1/4}$.
The proof here follows the proof of Theorem B.5. in \citet{LL}. Obviously, we only need to consider $i \in D_2$. Let $$T= \frac{2}{n}\sum_{j\in D_3} \frac{\pi_X(x_j)}{q_X(x_j)} k_L(x_j,x_i).$$
Lemma B.8 in \citet{LL} implies that
$$\mathbb{P}\left(\sum_{i=1}^{n}w^*_i<\frac{1}{2}\right)\le 2\exp(-\frac{n}{4L^2M_s}) \quad \text{where } M_s=M_2^2(M_2^2+\sqrt{2})^2/4,$$
$$\mathbb{P}(w^*_i<0)=\mathbb{P}(T>1)\le \exp(-\frac{n}{L^2M_2^4}).$$
Note that $\frac{1}{n}\frac{\pi_X(x_j)}{q_X(x_j)} \in [0, \frac{B}{n}]$ and $w_i^*(x)=\frac{1}{n}\frac{\pi_X(x_j)}{q_X(x_j)}(1-T).$ So $w_i^*(x)\ge \frac{2B}{n}$ implies that $T\le -1$. Using the similar argument, we obtain (by Hoeffding's inequality)
$$\mathbb{P}(w^*_i\ge\frac{2B}{n})\le Q(T\le -1)\le \exp(-\frac{n}{L^2M_2^4}).$$
Let $$\mathcal{E} = \left\{\sum_{i=1}^{n} w^*_i\ge 1/2, \ 0\le w^*_i \le \frac{2B}{n}, \ \forall i= 1,\cdots,n\right\}.$$ 
The above statement demonstrates that
$Q(\mathcal{E}^c)\le 2n\exp(-\frac{n}{L^2M_2^4})+2\exp(-\frac{n}{4L^2M_s}).$ Next we consider the following weights
$$w^+_i=\frac{\max(0,w^*_i)}{\sum_{i=1}^{n}\max(0,w^*_i)}.$$
Given the event $\mathcal{E}$, $0\le w^+_j\le \frac{4B}{n}$ is a feasible solution to the problem (\ref{COP0}) and thus,
$$\mathbb{S}(\{\hat{w}_j,x_j\},\pi_X)\le\mathbb{S}(\{w_j^+,x_j\},\pi_X).$$
Then we can express $\mathbb{E}_q[\mathbb{S}(\{\hat{w}_j,x_j\},\pi_X)]$ as
\begin{align*}
&\mathbb{E}_q[\mathbb{S}(\{\hat{w}_j,x_j\},\pi_X)]\\
=& \mathbb{E}_q[\mathbb{S}(\{\hat{w}_j,x_j\},\pi_X)|\mathcal{E}]\cdot \mathbb{P}[\mathcal{E}] + \mathbb{E}_q[\mathbb{S}(\{\hat{w}_j,x_j\},\pi_X)|\mathcal{E}^c]\cdot \mathbb{P}[\mathcal{E}^c]\\
\le& \mathbb{E}_q[\mathbb{S}(\{w^+_j,x_j\},\pi_X)|\mathcal{E}]\cdot \mathbb{P}[\mathcal{E}] + \mathbb{E}_q[\mathbb{S}(\{\hat{w}_j,x_j\},\pi_X)|\mathcal{E}^c]\cdot \mathbb{P}[\mathcal{E}^c]\\
\le& \mathbb{E}_q[\mathbb{S}(\{w^+_j,x_j\},\pi_X)] + 16B^2\kappa_0^2\cdot \left(2n\exp(-\frac{n}{L^2M_2^4})+2\exp(-\frac{n}{4L^2M_s})\right).
\end{align*}
by noting that $0\le \mathbb{S}(\{\hat{w}_j,x_j\},\pi_X)\le 16B^2\kappa_0^2.$
It follows from the remark below Theorem B.5 in \citet{LL} that 
$$\mathbb{E}_q[\mathbb{S}(\{w^+_j,x_j\},\pi_X)]= o(n^{-1}).$$
Obviously we also have
$$2n\exp(-\frac{n}{L^2M_2^4})+2\exp(-\frac{n}{4L^2M_s})= o(n^{-1}).$$
Hence, we finally obtain
$$\mathbb{E}_q[\mathbb{S}(\{\hat{w}_j,x_j\},\pi_X)]= o(n^{-1}).$$
\end{proof}

Now we are ready to prove Theorem \ref{mainIS0}.
\begin{proof} [Proof of Theorem \ref{mainIS0} in Section \ref{sec:DRSK main}]
Proposition 3.1 in \citet{LL} shows that
\begin{align*}
(\hat{\theta}_{IS}-\theta)^2&=\left(\sum_{j=1}^{n} \hat{w}_j (\bar{f}(x_j)+\epsilon(x_j,y_j)-\theta)\right)^2\\
&\le 2\left(\left(\sum_{j=1}^{n} \hat{w}_j (\bar{f}(x_j)-\theta)\right)^2+\left(\sum_{j=1}^{n} \hat{w}_j \epsilon(x_j,y_j)\right)^2\right)\\
&\le 2\left(\|\bar{f}-\theta\|^2_{\mathcal{H}_0} \cdot \mathbb{S}(\{\hat{w}_j,x_j\},\pi_X)+\left(\sum_{j=1}^{n} \hat{w}_j \epsilon(x_j,y_j)\right)^2\right)
\end{align*}
Note that $\|\bar{f}-\theta\|^2_{\mathcal{H}_0}$ is a constant. Combining Lemma \ref{noiselemma} and Lemma \ref{varifyweight1}, we obtain part (a).  Combining Lemma \ref{noiselemma} and Lemma \ref{varifyweight2}, we obtain part (b).
\end{proof}

\subsection{Doubly Robust Stein-Kernelized Estimators} \label{sec:DRSK}

We present and prove the following theorems that subsume the ones on our DRSK estimator in Section \ref{sec:DRSK main}.

\begin{theorem}[DRSK in all cases under weak assumptions]\label{mainDRSK1}
Suppose Assumptions \ref{CSA}, \ref{noise1}, and \ref{lrupper} hold. Take an RLS estimate with $\lambda = m^{-\frac{1}{2}}$ and $B_0=2B$ in (\ref{COP3}). Let $m=\alpha n$ where $0< \alpha< 1$. The DRSK estimator $\hat{\theta}_{DRSK}$ satisfies the following bound.\\
(a) If $\bar{f} \in \text{Range}(L_q^{r})$ ($\frac{1}{2}\le r \le 1$), then $\mathbb{E}_q [(\hat{\theta}_{DRSK}-\theta)^2] \le C_1(C_{f} n^{-\frac{1}{2}-r}+M_0 n^{-1})$ where $C_{f}=\|L_q^{-r}\bar{f}\|^2_{L^2(q_X)}$ (which is a constant indicating the regularity of $\bar{f}$ in $\mathcal{H}_+$), $C_1$ only depends on $\alpha, \kappa, B$.\\
(b) %If $\bar{f}\in \text{Range}(L_q^{\frac{1}{2}})$, then for any given $\varepsilon>0$, $\mathbb{E}_q [(\hat{\theta}_{DRSK}-\theta)^2] \le C_1(C_\varepsilon n^{-\frac{3}{2}}+M_0 n^{-1}+\varepsilon n^{-1})$ where $C_\varepsilon=\|L_q^{-1}g\|_{L^2(q_X)}$ and $g\in \text{Range}(L_q)$ is an $\varepsilon$-approximation of $\bar{f}$ in $\mathcal{H}_+$, $C_1$ only depends on $\alpha, \kappa, B$.\\
Suppose that there exists a $\varepsilon>0$ and $g \in \text{Range}(L_q)$, such that $\|\bar{f}-g\|^2_{\mathcal{H}_+}\le \varepsilon$. This assumption holds, for instance, if $\bar{f} \in \text{Range}(L_q^{\frac{1}{2}})$.
Then under this assumption, $\mathbb{E}_q [(\hat{\theta}_{DRSK}-\theta)^2] \le C_1(\|L_q^{-1}g\|^2_{L^2(q_X)} n^{-\frac{3}{2}}+M_0 n^{-1}+\varepsilon n^{-1})$ where $C_1$ only depends on $\alpha, \kappa, B$. \\
%In particular, if the noise term $Y$ does not exist, then we set $M_0=0$ in the result.
\end{theorem}

\begin{proof}
Similarly to the proof of Theorem \ref{mainIS0}, we express $(\hat{\theta}_{DRSK}-\theta)^2$ as
\begin{align*}
(\hat{\theta}_{DRSK}-\theta)^2&=\left(\sum_{j=m+1}^{n} \hat{w}_j (\bar{f}_m(x_j)+\epsilon(x_j,y_j)-\theta)\right)^2\\
&\le 2\left(\left(\sum_{j=m+1}^{n} \hat{w}_j (\bar{f}_m(x_j)-\theta)\right)^2+\left(\sum_{j=m+1}^{n} \hat{w}_j \epsilon(x_j,y_j)\right)^2\right)\\
&\le 2\left(\|\bar{f}_m-\theta\|^2_{\mathcal{H}_0} \cdot \mathbb{S}(\{\hat{w}_j,x_j\},\pi_X)+\left(\sum_{j=m+1}^{n} \hat{w}_j \epsilon(x_j,y_j)\right)^2\right)
\end{align*}
where $\bar{f}_m=\bar{f}-s_m+\mu_X(s_m)$. To see that $\bar{f}_m-\theta\in \mathcal{H}_0$ in the above equation, we note that $\bar{f}\in \text{Range}(L_q^{r}) \subset \mathcal{H}_+$ whenever $\frac{1}{2}\le r \le 1$ and $s_m\in \mathcal{H}_+$, which implies that $\bar{f}_m-\theta=\bar{f}-s_m+\mu_X(s_m)-\theta$ is a function in $\mathcal{H}_+$. We can further claim that $\bar{f}_m-\theta$ is a function in $\mathcal{H}_0$ since $\mu_X(\bar{f}_m-\theta)=\mu_X(\bar{f}-s_m+\mu_X(s_m)-\theta)=\theta-\mu_X(s_m)+\mu_X(s_m)-\theta=0$. Note that
$$\bar{f}-s_m=(\bar{f}-s_m+\mu_X(s_m)-\theta)+(\theta-\mu_X(s_m)),$$
so we can express
$$\|\bar{f}-s_m\|^2_{\mathcal{H}_+}=\|\bar{f}_m-\theta\|^2_{\mathcal{H}_0}+\|\theta-\mu_X(s_m)\|^2_{\mathcal{C}}.$$
Thus we have
$$\|\bar{f}_m-\theta\|^2_{\mathcal{H}_0}\le \|\bar{f}-s_m\|^2_{\mathcal{H}_+}.$$
Since $\mathbb{S}(\{\hat{w}_j,x_j\},\pi_X)$ depends only on $D_1$ and $\|\bar{f}_m-\theta\|^2_{\mathcal{H}_0}$ depends only on $D_0$, they are independent of each other. Hence,
\begin{align} \label{equa1}
&\mathbb{E}_q[(\hat{\theta}_{DRSK}-\theta)^2] \nonumber\\
\le &2\left(\mathbb{E}_q[\|\bar{f}_m-\theta\|^2_{\mathcal{H}_0}] \cdot \mathbb{E}_q[\mathbb{S}(\{\hat{w}_j,x_j\},\pi_X)]+\mathbb{E}_q\left[\left(\sum_{j=m+1}^{n} \hat{w}_j \epsilon(x_j,y_j)\right)^2\right]\right) \nonumber\\
\le &2\left(\mathbb{E}_q[\|\bar{f}-s_m\|^2_{\mathcal{H}_+}] \cdot O((n-m)^{-1})+M_0B_0^2(n-m)^{-1}\right)
\end{align}
by Lemma \ref{noiselemma} and Lemma \ref{varifyweight1}.

(a) We apply Corollary \ref{main_lemma2} (with $r$) to the samples $\{(x_i,f(x_i,y_i))\}$ to get
\begin{equation} \label{equa3}
\mathbb{E}_q[\|\bar{f}-s_m\|^2_{\mathcal{H}_+}]  \le C_\kappa \nu_X(( L_q^{-r} \bar{f})^2)m^{-r+\frac{1}{2}}+2\kappa^2 M_0    
\end{equation}
where $C_\kappa=2\kappa^2+2$.
Plugging (\ref{equa3}) into (\ref{equa1}), we obtain the desired result.\\

(b) We note that Proposition \ref{ease2} and Lemma \ref{isomorphism} show that $\text{Range}(L^{\frac{1}{2}}_q)= (\mathcal{H}_+)^q_1=\overline{\text{Range}(L_q)}^{\mathcal{H}_+}$. Therefore, the assumption holds if $\bar{f} \in \text{Range}(L_q^{\frac{1}{2}})$.

%Fix $\varepsilon>0$. Since $(\mathcal{H}_+)^q_1= \overline{\text{Range}(L_q)}^{\mathcal{H}_+}$ by Proposition \ref{ease2}, there exists a $g\in\text{Range}(L^{\frac{1}{2}}_q)= (\mathcal{H}_+)^q_1$ (Lemma \ref{isomorphism}) such that $\|\bar{f}-g\|^2_{\mathcal{H}_+}\le \varepsilon$.
Let $h=f-g$ so $\bar{h}=\bar{f}-g$.  Let $s^{h}_m$, $s^{g}_m$ be the RLS functional approximation of $h$, $g$ respectively. Note that $s^{h}_m$ is a linear functional of $h$, so we can write
$$h - s^{h}_m=(f - s_m)-(g - s^{g}_m).$$
We apply Corollary \ref{main_lemma2} (with $r=1$) to the samples $\{(x_i, g(x_i))\}$: Since $g\in\text{Range}(L_q)$ and $g$ is a function of $x$ only, then $M^g_0=0$ and
$$\mathbb{E}_q[\|g-s_m^g\|^2_{\mathcal{H}_+}] \le C_\kappa m^{-\frac{1}{2}} \nu_X((L_q^{-1}g)^2).$$
Again we apply Corollary \ref{main_lemma2} (with $r=\frac{1}{2}$) to the samples $\{(x_i,h(x_i,y_i))\}$: We note that $\bar{h}=\bar{f}-g\in \text{Range}(L_q^{\frac{1}{2}})=(\mathcal{H}_+)^q_1$ and $\bar{g}=g$ so 
$$M^h_0:=\mathbb{E}_\pi [(h(x,y)-\bar{h}(x))^2]=\mathbb{E}_\pi [(f(x,y)-\bar{f}(x))^2]=M_0$$
and thus
$$\mathbb{E}_q[\|h-s^h_m\|^2_{\mathcal{H}_+}]  \le C_\kappa \nu_X(( L_q^{-\frac{1}{2}} \bar{h})^2)+2\kappa^2 M_0= C_\kappa \|\bar{h}\|^2_{\mathcal{H}_+}+2\kappa^2 M_0\le C_\kappa \varepsilon+2\kappa^2 M_0$$
where $C_\kappa=2\kappa^2+2$.\\
Adding these two parts we obtain
\begin{equation} \label{equa2}
\mathbb{E}_q[\|\bar{f} - s_m\|^2_{\mathcal{H}_+}] \le 2 C_\kappa (\varepsilon+ m^{-\frac{1}{2}} \nu_X((L_q^{-1}g)^2))+4\kappa^2 M_0=C_1(\varepsilon+M_0+\|L_q^{-1}g\|^2_{L^2(q_X)} m^{-\frac{1}{2}}).    
\end{equation}
Plugging (\ref{equa2}) into (\ref{equa1}), we obtain the desired result.
\end{proof}

Theorem \ref{mainDRSK0} in Section \ref{sec:DRSK main} is part (a) of Theorem \ref{mainDRSK1}.

%As in the BBIS estimator, we can obtain a better result if we impose a stronger assumption. Therefore, the DRSK estimator is guaranteed to be better than BBIS estimator in terms of the MSE.

\begin{theorem}[DRSK in all cases under strong assumptions]\label{mainDRSK2}
Suppose Assumptions \ref{CSA}, \ref{noise1}, \ref{lrupper}, and \ref{lrupper4} hold. Take an RLS estimate with $\lambda = m^{-\frac{1}{2}}$ and $B_0=4B$ in (\ref{COP3}). Let $m=\alpha n$ where $0< \alpha< 1$. The DRSK estimator $\hat{\theta}_{DRSK}$ satisfies the following bound.\\
(a) If $\bar{f} \in \text{Range}(L_q^{r})$ ($\frac{1}{2}\le r \le 1$), then $\mathbb{E}_q [(\hat{\theta}_{DRSK}-\theta)^2] \le C_1(C_{f,n} n^{-\frac{1}{2}-r}+M_0 n^{-1})$ where $C_{f,n}=\|L_q^{-r}\bar{f}\|^2_{L^2(q_X)} \cdot o(1)$ as $n\to \infty$, $C_1$ only depends on $\alpha, \kappa, B$.\\
(b) %If $\bar{f}\in \text{Range}(L_q^{\frac{1}{2}})$, then for any given $\varepsilon>0$, $\mathbb{E}_q [(\hat{\theta}_{DRSK}-\theta)^2] \le C_1(C_{\varepsilon,n} n^{-\frac{3}{2}}+M_0 n^{-1}+\varepsilon n^{-1})$ where $C_{\varepsilon,n}=\|L_q^{-1}g\|_{L^2(q_X)} \cdot o(1)$  as $n\to \infty$ and $g\in \text{Range}(L_q)$ is an $\varepsilon$-approximation of $\bar{f}$ in $\mathcal{H}_+$, $C_1$ only depends on $\alpha, \kappa, B$. \\
Suppose that there exists a $\varepsilon>0$ and $g \in \text{Range}(L_q)$, such that $\|\bar{f}-g\|^2_{\mathcal{H}_+}\le \varepsilon$. This assumption holds, for instance, if $\bar{f} \in \text{Range}(L_q^{\frac{1}{2}})$.
Then under this assumption, $\mathbb{E}_q [(\hat{\theta}_{DRSK}-\theta)^2] \le C_1(C_{g,n} n^{-\frac{3}{2}}+M_0 n^{-1}+\varepsilon n^{-1})$ where $C_{g,n}=\|L_q^{-1}g\|^2_{L^2(q_X)} \cdot o(1)$  as $n\to \infty$, $C_1$ only depends on $\alpha, \kappa, B$. 
%In particular, if the noise term $Y$ does not exist, then we set $M_0=0$ in the result.
\end{theorem}

\begin{proof}
Similarly to the proof of Theorem \ref{mainDRSK1} but leveraging Lemma \ref{noiselemma} and Lemma \ref{varifyweight2} in this Theorem, we obtain that 
\begin{align*}
&\mathbb{E}_q[(\hat{\theta}_{DRSK}-\theta)^2]\\
\le &2\left(\mathbb{E}_q[\|\bar{f}_m-\theta\|^2_{\mathcal{H}_0}] \cdot \mathbb{E}_q[\mathbb{S}(\{\hat{w}_j,x_j\},\pi_X)]+\mathbb{E}_q\left[\left(\sum_{j=m+1}^{n} \hat{w}_j \epsilon(x_j,y_j)\right)^2\right]\right)\\
\le &2\left(\mathbb{E}_q[\|\bar{f}-s_m\|^2_{\mathcal{H}_+}] \cdot o((n-m)^{-1})+M_0B_0^2(n-m)^{-1}\right).
\end{align*}
The rest of the proof is similar to Theorem \ref{mainDRSK1}.
\end{proof}

Theorem \ref{mainDRSK0-2} in Section \ref{sec:DRSK main} is part (a) of Theorem \ref{mainDRSK2}.

\acks{We gratefully acknowledge support from the National Science Foundation under grants CAREER CMMI-1834710 and IIS-1849280. The research of Haofeng Zhang is supported in part by the Cheung-Kong Innovation Doctoral Fellowship. The authors also thank the reviewers and editor for their constructive comments, which have helped improve our paper tremendously.}

% The research of Haofeng Zhang is supported in part by the Cheung-Kong Innovation Doctoral Fellowship.

\bibliography{bibliography}

\begin{thebibliography}{75}
\providecommand{\natexlab}[1]{#1}
\providecommand{\url}[1]{\texttt{#1}}
\expandafter\ifx\csname urlstyle\endcsname\relax
  \providecommand{\doi}[1]{doi: #1}\else
  \providecommand{\doi}{doi: \begingroup \urlstyle{rm}\Url}\fi

\bibitem[Asmussen and Glynn(2007)]{asmussen2007stochastic}
S{\o}ren Asmussen and Peter~W Glynn.
\newblock \emph{Stochastic Simulation: Algorithms and Analysis}, volume~57.
\newblock Springer Science \& Business Media, 2007.

\bibitem[Barton et~al.(2022)Barton, Lam, and Song]{barton2022input}
Russell~R Barton, Henry Lam, and Eunhye Song.
\newblock Input uncertainty in stochastic simulation.
\newblock In \emph{The Palgrave Handbook of Operations Research}, pages
  573--620. Springer, 2022.

\bibitem[Belomestny et~al.(2021)Belomestny, Iosipoi, Moulines, Naumov, and
  Samsonov]{belomestny2021variance}
Denis Belomestny, Leonid Iosipoi, Eric Moulines, Alexey Naumov, and Sergey
  Samsonov.
\newblock Variance reduction for dependent sequences with applications to
  stochastic gradient {MCMC}.
\newblock \emph{SIAM/ASA Journal on Uncertainty Quantification}, 9\penalty0
  (2):\penalty0 507--535, 2021.

\bibitem[Berlinet and Thomas-Agnan(2011)]{berlinet2011reproducing}
Alain Berlinet and Christine Thomas-Agnan.
\newblock \emph{Reproducing Kernel Hilbert Spaces in Probability and
  Statistics}.
\newblock Springer Science \& Business Media, 2011.

\bibitem[Blanchet and Lam(2012)]{BLANCHET201238}
Jose Blanchet and Henry Lam.
\newblock State-dependent importance sampling for rare-event simulation: An
  overview and recent advances.
\newblock \emph{Surveys in Operations Research and Management Science},
  17\penalty0 (1):\penalty0 38--59, 2012.

\bibitem[Blei et~al.(2017)Blei, Kucukelbir, and McAuliffe]{blei2017variational}
David~M Blei, Alp Kucukelbir, and Jon~D McAuliffe.
\newblock Variational inference: A review for statisticians.
\newblock \emph{Journal of the American statistical Association}, 112\penalty0
  (518):\penalty0 859--877, 2017.

\bibitem[Bucklew(2013)]{bucklew2013introduction}
James Bucklew.
\newblock \emph{Introduction to Rare Event Simulation}.
\newblock Springer Science \& Business Media, 2013.

\bibitem[Christmann and Steinwart(2007)]{christmann2007consistency}
Andreas Christmann and Ingo Steinwart.
\newblock Consistency and robustness of kernel-based regression in convex risk
  minimization.
\newblock \emph{Bernoulli}, 13\penalty0 (3):\penalty0 799--819, 2007.

\bibitem[Chwialkowski et~al.(2016)Chwialkowski, Strathmann, and Gretton]{CSG}
Kacper Chwialkowski, Heiko Strathmann, and Arthur Gretton.
\newblock A kernel test of goodness of fit.
\newblock In \emph{International Conference on Machine Learning}, volume~48,
  pages 2606--2615, 2016.

\bibitem[Corlu et~al.(2020)Corlu, Akcay, and Xie]{corlu2020stochastic}
Canan~G Corlu, Alp Akcay, and Wei Xie.
\newblock Stochastic simulation under input uncertainty: A review.
\newblock \emph{Operations Research Perspectives}, 7:\penalty0 100162, 2020.

\bibitem[Cucker and Smale(2002{\natexlab{a}})]{CS}
Felipe Cucker and Steve Smale.
\newblock On the mathematical foundations of learning.
\newblock \emph{Bulletin of the American Mathematical Society}, 39\penalty0
  (1):\penalty0 1--49, 2002{\natexlab{a}}.

\bibitem[Cucker and Smale(2002{\natexlab{b}})]{CS2}
Felipe Cucker and Steve Smale.
\newblock Best choices for regularization parameters in learning theory: On the
  bias-variance problem.
\newblock \emph{Foundations of Computational Mathematics}, 2\penalty0
  (4):\penalty0 413--428, 2002{\natexlab{b}}.

\bibitem[Cucker and Zhou(2007)]{CZ}
Felipe Cucker and Ding~Xuan Zhou.
\newblock \emph{Learning Theory: An Approximation Theory Viewpoint}, volume~24.
\newblock Cambridge University Press, 2007.

\bibitem[Debruyne et~al.(2008)Debruyne, Hubert, and Suykens]{debruyne2008model}
Michiel Debruyne, Mia Hubert, and Johan~AK Suykens.
\newblock Model selection in kernel based regression using the influence
  function.
\newblock \emph{Journal of machine learning research}, 9:\penalty0 2377--2400,
  2008.

\bibitem[Dud{\'\i}k et~al.(2011)Dud{\'\i}k, Langford, and Li]{dudik2011doubly}
Miroslav Dud{\'\i}k, John Langford, and Lihong Li.
\newblock Doubly robust policy evaluation and learning.
\newblock In \emph{International Conference on Machine Learning}, pages
  1097--1104, 2011.

\bibitem[Dud{\'\i}k et~al.(2014)Dud{\'\i}k, Erhan, Langford, and
  Li]{dudik2014doubly}
Miroslav Dud{\'\i}k, Dumitru Erhan, John Langford, and Lihong Li.
\newblock Doubly robust policy evaluation and optimization.
\newblock \emph{Statistical Science}, 29\penalty0 (4):\penalty0 485--511, 2014.

\bibitem[Farajtabar et~al.(2018)Farajtabar, Chow, and
  Ghavamzadeh]{farajtabar2018more}
Mehrdad Farajtabar, Yinlam Chow, and Mohammad Ghavamzadeh.
\newblock More robust doubly robust off-policy evaluation.
\newblock In \emph{International Conference on Machine Learning}, pages
  1447--1456, 2018.

\bibitem[Gaunt et~al.(2019)Gaunt, Mijoule, and Swan]{gaunt2019algebra}
Robert~E Gaunt, Guillaume Mijoule, and Yvik Swan.
\newblock An algebra of {S}tein operators.
\newblock \emph{Journal of Mathematical Analysis and Applications},
  469\penalty0 (1):\penalty0 260--279, 2019.

\bibitem[Glasserman(2003)]{Gl}
Paul Glasserman.
\newblock \emph{{M}onte {C}arlo Methods in Financial Engineering}.
\newblock Springer-Verlag New York, 2003.

\bibitem[Glasserman and Yu(2005)]{doi:10.1287/opre.1040.0148}
Paul Glasserman and Bin Yu.
\newblock Large sample properties of weighted {M}onte {C}arlo estimators.
\newblock \emph{Operations Research}, 53\penalty0 (2):\penalty0 298--312, 2005.

\bibitem[Glynn and Szechtman(2002)]{glynn2002some}
Peter~W Glynn and Roberto Szechtman.
\newblock Some new perspectives on the method of control variates.
\newblock In \emph{{M}onte {C}arlo and Quasi-{M}onte {C}arlo Methods 2000},
  pages 27--49. Springer, 2002.

\bibitem[Gretton et~al.(2009)Gretton, Smola, Huang, Schmittfull, Borgwardt, and
  Sch{\"o}lkopf]{gretton2009covariate}
Arthur Gretton, Alex Smola, Jiayuan Huang, Marcel Schmittfull, Karsten
  Borgwardt, and Bernhard Sch{\"o}lkopf.
\newblock Covariate shift by kernel mean matching.
\newblock \emph{Dataset Shift in Machine Learning}, 3\penalty0 (4):\penalty0 5,
  2009.

\bibitem[Gretton et~al.(2012)Gretton, Borgwardt, Rasch, Sch{\"o}lkopf, and
  Smola]{gretton2012kernel}
Arthur Gretton, Karsten~M Borgwardt, Malte~J Rasch, Bernhard Sch{\"o}lkopf, and
  Alexander Smola.
\newblock A kernel two-sample test.
\newblock \emph{Journal of Machine Learning Research}, 13\penalty0
  (Mar):\penalty0 723--773, 2012.

\bibitem[Han and Liu(2018)]{han2018stein}
Jun Han and Qiang Liu.
\newblock {S}tein variational gradient descent without gradient.
\newblock In \emph{International Conference on Machine Learning}, pages
  1900--1908, 2018.

\bibitem[Hastie et~al.(2009)Hastie, Tibshirani, and
  Friedman]{hastie2009elements}
T.~Hastie, R.~Tibshirani, and J.H. Friedman.
\newblock \emph{The Elements of Statistical Learning: Data Mining, Inference,
  and Prediction}.
\newblock Springer, 2009.

\bibitem[Henderson and Glynn(2002)]{henderson2002approximating}
Shane~G Henderson and Peter~W Glynn.
\newblock Approximating martingales for variance reduction in {M}arkov process
  simulation.
\newblock \emph{Mathematics of Operations Research}, 27\penalty0 (2):\penalty0
  253--271, 2002.

\bibitem[Henderson and Simon(2004)]{henderson2004adaptive}
Shane~G Henderson and Burt Simon.
\newblock Adaptive simulation using perfect control variates.
\newblock \emph{Journal of Applied Probability}, 41\penalty0 (3):\penalty0
  859--876, 2004.

\bibitem[Hodgkinson et~al.(2020)Hodgkinson, Salomone, and
  Roosta]{hodgkinson2020reproducing}
Liam Hodgkinson, Robert Salomone, and Fred Roosta.
\newblock The reproducing {S}tein kernel approach for post-hoc corrected
  sampling.
\newblock \emph{arXiv preprint arXiv:2001.09266}, 2020.

\bibitem[Hoeffding(1963)]{Hoeffding}
Wassily Hoeffding.
\newblock Probability inequalities for sums of bounded random variables.
\newblock \emph{Journal of the American Statistical Association}, 58\penalty0
  (301):\penalty0 13--30, 1963.

\bibitem[Jacob et~al.(2011)Jacob, Robert, and Smith]{jacob2011using}
Pierre Jacob, Christian~P Robert, and Murray~H Smith.
\newblock Using parallel computation to improve independent
  {M}etropolis--{H}astings based estimation.
\newblock \emph{Journal of Computational and Graphical Statistics}, 20\penalty0
  (3):\penalty0 616--635, 2011.

\bibitem[Jiang and Li(2016)]{jiang2016doubly}
Nan Jiang and Lihong Li.
\newblock Doubly robust off-policy value evaluation for reinforcement learning.
\newblock In \emph{International Conference on Machine Learning}, pages
  652--661, 2016.

\bibitem[Juneja and Shahabuddin(2006)]{JUNEJA2006291}
Sandeep Juneja and Perwez Shahabuddin.
\newblock Rare-event simulation techniques: an introduction and recent
  advances.
\newblock \emph{Handbooks in Operations Research and Management Science},
  13:\penalty0 291--350, 2006.

\bibitem[Kim and Henderson(2007)]{kim2007adaptive}
Sujin Kim and Shane~G Henderson.
\newblock Adaptive control variates for finite-horizon simulation.
\newblock \emph{Mathematics of Operations Research}, 32\penalty0 (3):\penalty0
  508--527, 2007.

\bibitem[Kpotufe and Martinet(2021)]{KM}
Samory Kpotufe and Guillaume Martinet.
\newblock Marginal singularity and the benefits of labels in covariate-shift.
\newblock \emph{The Annals of Statistics}, 49\penalty0 (6):\penalty0
  3299--3323, 2021.

\bibitem[Lam(2016)]{lam2016advanced}
Henry Lam.
\newblock Advanced tutorial: Input uncertainty and robust analysis in
  stochastic simulation.
\newblock In \emph{Winter Simulation Conference}, pages 178--192, 2016.

\bibitem[Lam and Qian(2021)]{lam2021subsampling}
Henry Lam and Huajie Qian.
\newblock Subsampling to enhance efficiency in input uncertainty
  quantification.
\newblock \emph{Operations Research}, 2021.

\bibitem[Lam and Zhang(2019)]{HZ}
Henry Lam and Haofeng Zhang.
\newblock On the stability of kernelized control functionals on partial and
  biased stochastic inputs.
\newblock In \emph{Winter Simulation Conference}, 2019.

\bibitem[Leluc et~al.(2021)Leluc, Portier, and Segers]{leluc2021control}
R{\'e}mi Leluc, Fran{\c{c}}ois Portier, and Johan Segers.
\newblock Control variate selection for {M}onte {C}arlo integration.
\newblock \emph{Statistics and Computing}, 31\penalty0 (4):\penalty0 50, 2021.

\bibitem[Ley et~al.(2017)Ley, Reinert, and Swan]{ley2017stein}
Christophe Ley, Gesine Reinert, and Yvik Swan.
\newblock Stein’s method for comparison of univariate distributions.
\newblock \emph{Probability Surveys}, 14:\penalty0 1--52, 2017.

\bibitem[Li et~al.(2020)Li, Lam, and Prusty]{lam2019robust}
Fengpei Li, Henry Lam, and Siddharth Prusty.
\newblock Robust importance weighting for covariate shift.
\newblock In \emph{International Conference on Artificial Intelligence and
  Statistics}, pages 352--362, 2020.

\bibitem[Lin et~al.(2015)Lin, Song, and Nelson]{lin2015single}
Y~Lin, Eunhye Song, and Barry~L Nelson.
\newblock Single-experiment input uncertainty.
\newblock \emph{Journal of Simulation}, 9\penalty0 (3):\penalty0 249--259,
  2015.

\bibitem[Liu(2017)]{liu2017stein}
Qiang Liu.
\newblock {S}tein variational gradient descent as gradient flow.
\newblock In \emph{Advances in Neural Information Processing Systems}, pages
  3115--3123, 2017.

\bibitem[Liu and Lee(2017)]{LL}
Qiang Liu and Jason Lee.
\newblock Black-box importance sampling.
\newblock In \emph{International Conference on Artificial Intelligence and
  Statistics}, volume~54, pages 952--961, 2017.

\bibitem[Liu and Wang(2016)]{liu2016stein}
Qiang Liu and Dilin Wang.
\newblock {S}tein variational gradient descent: A general purpose {B}ayesian
  inference algorithm.
\newblock In \emph{Advances in neural information processing systems}, pages
  2378--2386, 2016.

\bibitem[Liu et~al.(2016)Liu, Lee, and Jordan]{liu2016kernelized}
Qiang Liu, Jason~D. Lee, and Michael Jordan.
\newblock A kernelized {S}tein discrepancy for goodness-of-fit tests.
\newblock In \emph{International Conference on Machine Learning}, volume~48,
  pages 276--284, 2016.

\bibitem[Liu et~al.(2017)Liu, Ramachandran, Liu, and Peng]{liu2017stein2}
Yang Liu, Prajit Ramachandran, Qiang Liu, and Jian Peng.
\newblock {S}tein variational policy gradient.
\newblock In \emph{Uncertainty in Artificial Intelligence}, 2017.

\bibitem[Maire(2003)]{maire2003reducing}
Sylvain Maire.
\newblock Reducing variance using iterated control variates.
\newblock \emph{Journal of Statistical Computation and Simulation}, 73\penalty0
  (1):\penalty0 1--30, 2003.

\bibitem[Mijoule et~al.(2018)Mijoule, Reinert, and Swan]{mijoule2018stein}
Guillaume Mijoule, Gesine Reinert, and Yvik Swan.
\newblock Stein operators, kernels and discrepancies for multivariate
  continuous distributions.
\newblock \emph{arXiv preprint arXiv:1806.03478}, 2018.

\bibitem[Nelson(1990)]{nelson1990control}
Barry~L Nelson.
\newblock Control variate remedies.
\newblock \emph{Operations Research}, 38\penalty0 (6):\penalty0 974--992, 1990.

\bibitem[Oates et~al.(2017)Oates, Girolami, and Chopin]{CF1}
Chris~J Oates, Mark Girolami, and Nicolas Chopin.
\newblock Control functionals for {M}onte {C}arlo integration.
\newblock \emph{Journal of the Royal Statistical Society: Series B (Statistical
  Methodology)}, 79\penalty0 (3):\penalty0 695--718, 2017.

\bibitem[Oates et~al.(2019)Oates, Cockayne, Briol, and Girolami]{CF2}
Chris~J. Oates, Jon Cockayne, Fran{\c{c}}ois-Xavier Briol, and Mark Girolami.
\newblock Convergence rates for a class of estimators based on {S}tein's
  method.
\newblock \emph{Bernoulli}, 25\penalty0 (2):\penalty0 1141--1159, 2019.

\bibitem[Owen and Zhou(2000)]{owen2000safe}
Art Owen and Yi~Zhou.
\newblock Safe and effective importance sampling.
\newblock \emph{Journal of the American Statistical Association}, 95\penalty0
  (449):\penalty0 135--143, 2000.

\bibitem[Portier and Segers(2018)]{portier2019monte}
François Portier and Johan Segers.
\newblock {M}onte {C}arlo integration with a growing number of control
  variates.
\newblock \emph{Journal of Applied Probability}, 56, 2018.

\bibitem[Rahimi and Recht(2008)]{RR}
Ali Rahimi and Benjamin Recht.
\newblock Random features for large-scale kernel machines.
\newblock In \emph{Advances in Neural Information Processing Systems}, pages
  1177--1184, 2008.

\bibitem[Rosenthal(2000)]{rosenthal2000parallel}
Jeffrey~S Rosenthal.
\newblock Parallel computing and {M}onte {C}arlo algorithms.
\newblock \emph{Far East Journal of Theoretical Statistics}, 4\penalty0
  (2):\penalty0 207--236, 2000.

\bibitem[Rubino and Tuffin(2009)]{rubino2009rare}
Gerardo Rubino and Bruno Tuffin.
\newblock \emph{Rare Event Simulation using Monte Carlo Methods}, volume~73.
\newblock Wiley Online Library, 2009.

\bibitem[Rubinstein and Kroese(2016)]{rubinstein2016simulation}
Reuven~Y Rubinstein and Dirk~P Kroese.
\newblock \emph{Simulation and the {M}onte {C}arlo Method}, volume~10.
\newblock John Wiley \& Sons, 2016.

\bibitem[Smale and Zhou(2004)]{SZ}
Steve Smale and Ding-Xuan Zhou.
\newblock Shannon sampling and function reconstruction from point values.
\newblock \emph{Bulletin of the American Mathematical Society}, 41\penalty0
  (3):\penalty0 279--306, 2004.

\bibitem[Smale and Zhou(2005)]{SZ2}
Steve Smale and Ding-Xuan Zhou.
\newblock Shannon sampling {II}: Connections to learning theory.
\newblock \emph{Applied and Computational Harmonic Analysis}, 19\penalty0
  (3):\penalty0 285--302, 2005.

\bibitem[Smale and Zhou(2007)]{SZ3}
Steve Smale and Ding-Xuan Zhou.
\newblock Learning theory estimates via integral operators and their
  approximations.
\newblock \emph{Constructive Approximation}, 26\penalty0 (2):\penalty0
  153--172, 2007.

\bibitem[So{\l}tan(2018)]{SO}
P.~So{\l}tan.
\newblock \emph{A Primer on Hilbert Space Operators}.
\newblock Birkh{\"{a}}user Cham, Springer, 2018.

\bibitem[Song et~al.(2014)Song, Nelson, and Pegden]{song2014advanced}
Eunhye Song, Barry~L Nelson, and C~Dennis Pegden.
\newblock Advanced tutorial: Input uncertainty quantification.
\newblock In \emph{Winter Simulation Conference}, pages 162--176, 2014.

\bibitem[South et~al.(2022{\natexlab{a}})South, Oates, Mira, and
  Drovandi]{south2022regularized}
Leah~F South, Chris~J Oates, Antonietta Mira, and Christopher Drovandi.
\newblock Regularized zero-variance control variates.
\newblock \emph{Bayesian Analysis}, 1\penalty0 (1):\penalty0 1--24,
  2022{\natexlab{a}}.

\bibitem[South et~al.(2022{\natexlab{b}})South, Riabiz, Teymur, and
  Oates]{south2022postprocessing}
Leah~F South, Marina Riabiz, Onur Teymur, and Chris~J Oates.
\newblock Postprocessing of {MCMC}.
\newblock \emph{Annual Review of Statistics and Its Application}, 9:\penalty0
  529--555, 2022{\natexlab{b}}.

\bibitem[Stein et~al.(2004)Stein, Diaconis, Holmes, and Reinert]{stein2004use}
Charles Stein, Persi Diaconis, Susan Holmes, and Gesine Reinert.
\newblock Use of exchangeable pairs in the analysis of simulations.
\newblock \emph{Lecture Notes-Monograph Series}, pages 1--26, 2004.

\bibitem[Stuart(2010)]{stuart2010inverse}
Andrew~M Stuart.
\newblock Inverse problems: A {B}ayesian perspective.
\newblock \emph{Acta Numerica}, 19:\penalty0 451--559, 2010.

\bibitem[Sun and Wu(2009)]{SW}
Hongwei Sun and Qiang Wu.
\newblock Application of integral operator for regularized least-square
  regression.
\newblock \emph{Mathematical and Computer Modelling}, 49\penalty0
  (1-2):\penalty0 276--285, 2009.

\bibitem[Sun and Wu(2010)]{SW2}
Hongwei Sun and Qiang Wu.
\newblock Regularized least square regression with dependent samples.
\newblock \emph{Advances in Computational Mathematics}, 32\penalty0
  (2):\penalty0 175--189, 2010.

\bibitem[Wainwright and Jordan(2008)]{wainwright2008graphical}
Martin~J Wainwright and Michael~I Jordan.
\newblock Graphical models, exponential families, and variational inference.
\newblock \emph{Foundations and Trends{\textregistered} in Machine Learning},
  1\penalty0 (1--2):\penalty0 1--305, 2008.

\bibitem[Wang and Liu(2019)]{wang2019stein}
Dilin Wang and Qiang Liu.
\newblock Nonlinear {S}tein variational gradient descent for learning
  diversified mixture models.
\newblock In \emph{International Conference on Machine Learning}, pages
  6576--6585, 2019.

\bibitem[Welling(2013)]{welling2013kernel}
Max Welling.
\newblock Kernel ridge regression.
\newblock \emph{Max Welling’s Classnotes in Machine Learning}, pages 1--3,
  2013.

\bibitem[Xie et~al.(2014)Xie, Nelson, and Barton]{xie2014bayesian}
Wei Xie, Barry~L Nelson, and Russell~R Barton.
\newblock A {B}ayesian framework for quantifying uncertainty in stochastic
  simulation.
\newblock \emph{Operations Research}, 62\penalty0 (6):\penalty0 1439--1452,
  2014.

\bibitem[Yu and Szepesv{\'{a}}ri(2012)]{YS}
Yaoliang Yu and Csaba Szepesv{\'{a}}ri.
\newblock Analysis of kernel mean matching under covariate shift.
\newblock In \emph{International Conference on Machine Learning}, pages
  607--614, 2012.

\bibitem[Zhang et~al.(2013)Zhang, Duchi, and Wainwright]{pmlr-v30-Zhang13}
Yuchen Zhang, John Duchi, and Martin Wainwright.
\newblock Divide and conquer kernel ridge regression.
\newblock In \emph{Annual Conference on Learning Theory}, pages 592--617, 2013.

\bibitem[Zouaoui and Wilson(2003)]{zouaoui2003accounting}
Faker Zouaoui and James~R Wilson.
\newblock Accounting for parameter uncertainty in simulation input modeling.
\newblock \emph{IIE Transactions}, 35\penalty0 (9):\penalty0 781--792, 2003.

\end{thebibliography}

\end{document}